\documentclass[12pt]{article}

\usepackage[margin=1in]{geometry} %For changing margins
\usepackage{amsmath} %Initates various math environment, including align
\usepackage{amsthm} % Theorem environment
\usepackage{amsfonts, amssymb} %Either provides Real number and other font'related stuff.
\usepackage{graphicx} %For adding graphics
\usepackage{epstopdf} %for adding eps in pdfLatex
\usepackage{natbib} % bibliography package
\usepackage{setspace} %sets spacing options for entire text
\usepackage{titlesec} %centering section headings
\usepackage[usenames]{color}
\usepackage{bm} %Bolds everything when uses \bm as a command
\usepackage[justification=justified,singlelinecheck=false]{caption} %Captioning in tables
\usepackage{morefloats}
 
\setcitestyle{aysep={}}
\titleformat{\section}[block]{\Large\bfseries\filcenter}{\thesection}{1em}{}

\newtheorem{theorem}{Theorem}
\newtheorem{corollary}{Corollary}
\newtheorem{lemma}{Lemma}

\theoremstyle{definition}
\newtheorem{definition}{Definition}

\theoremstyle{remark}

\newcommand{\iid}{\stackrel{\mathrm{iid}}{\sim}}

\newcommand{\amin}[1]{\underset{#1}{\operatorname{argmin~}}} %~ gives spacing

\newcommand{\reals}{\mathbb{R}}

\newcommand{\Z}{\mathbf{Z}}
\newcommand{\D}{\mathbf{D}}
\newcommand{\PP}{\mathbf{P}}
\newcommand{\aalpha}{\bm{\alpha}}
\newcommand{\M}{\mathbf{M}}

\begin{document}
\title{Instrumental Variables Estimation With Some Invalid Instruments and its Application to Mendelian Randomization\footnote{\emph{Address for correspondence:} Hyunseung Kang, Department of Statistics, The Wharton School, University of Pennsylvania, Jon M. Huntsman Hall, 3730 Walnut Street, Philadelphia, PA 19104-6340, USA. E-mail: khyuns@wharton.upenn.edu. Hyunseung Kang is Ph.D. student (E-mail: khyuns@wharton.upenn.edu); Anru Zhang is Ph.D. student (E-mail: anrzhang@wharton.upenn.edu); T. Tony Cai is Dorothy Silberberg Professor of Statistics (E-mail: tcai@wharton.upenn.edu); and Dylan S. Small is Professor (E-mail: dsmall@wharton.upenn.edu). The research of Tony Cai and Anru Zhang was supported in part by NSF FRG Grant DMS-0854973, NSF Grant DMS-1208982 and NIH Grant R01 CA127334-05. The research of Hyunseung Kang and Dylan Small was supported in part by NSF Grant SES-1260782. The authors thank Jack Bowden, two referees, and the associate editor for helpful suggestions.}}
\author{Hyunseung Kang,  Anru Zhang, T. Tony Cai, Dylan S. Small\\
Department of Statistics\\
The Wharton School\\
University of Pennsylvania
}
 
\date{}
\maketitle
\begin{abstract}
Instrumental variables have been widely used for estimating the causal effect between exposure and outcome. Conventional estimation methods require complete knowledge about all the instruments' validity; a valid instrument must not have a direct effect on the outcome and not be related to unmeasured confounders. Often, this is impractical as highlighted by Mendelian randomization studies where genetic markers are used as instruments and complete knowledge about instruments' validity is equivalent to complete knowledge about the involved genes' functions. 

In this paper, we propose a method for estimation of causal effects when this complete knowledge is absent. It is shown that causal effects are identified and can be estimated as long as less than $50$\% of instruments are invalid, without knowing which of the  instruments are invalid. We also introduce conditions for identification when the 50\% threshold is violated. A fast penalized $\ell_1$  estimation method, called sisVIVE, is introduced for estimating the causal effect without knowing which instruments are valid, with theoretical guarantees on its performance. The proposed method is demonstrated on simulated data and a real Mendelian randomization study concerning the effect of body mass index on health-related quality of life index. An R package \emph{sisVIVE} is available on CRAN. Supplementary materials for this article are available online.
\end{abstract}

\newpage
\section{INTRODUCTION}
Instrumental variables (IV) is a popular method for estimating the causal effect of an exposure on an outcome when there is unmeasured confounding. Conventional IV estimation methods require that the instruments are valid, or informally speaking, that the instruments are (A1) related to the exposure (A2) have no direct pathway to the outcome and (A3) are not related to unmeasured variables that affect the exposure and the outcome (see Figure \ref{fig:dag} and Section 2 for a formal definition of valid IVs). For example, Figure \ref{fig:dag} is an illustration of the IV assumptions and one potential violation of the IV assumptions(see \citet{hernan_instruments_2006} for details on other possible violations). Here, the IV is a genetic marker that is a single nucleotide polymorphism whose value is fixed at birth and the unmeasured variables refer to variables that precede the assignment of the genetic marker, such as population stratification (to be discussed later). The challenge in IV estimation is to find valid instruments that satisfy assumptions (A1)-(A3). Unfortunately, this is a difficult task, especially in the case of Mendelian randomization (MR).

\begin{figure}[h!]
\centering
\includegraphics[scale=0.3]{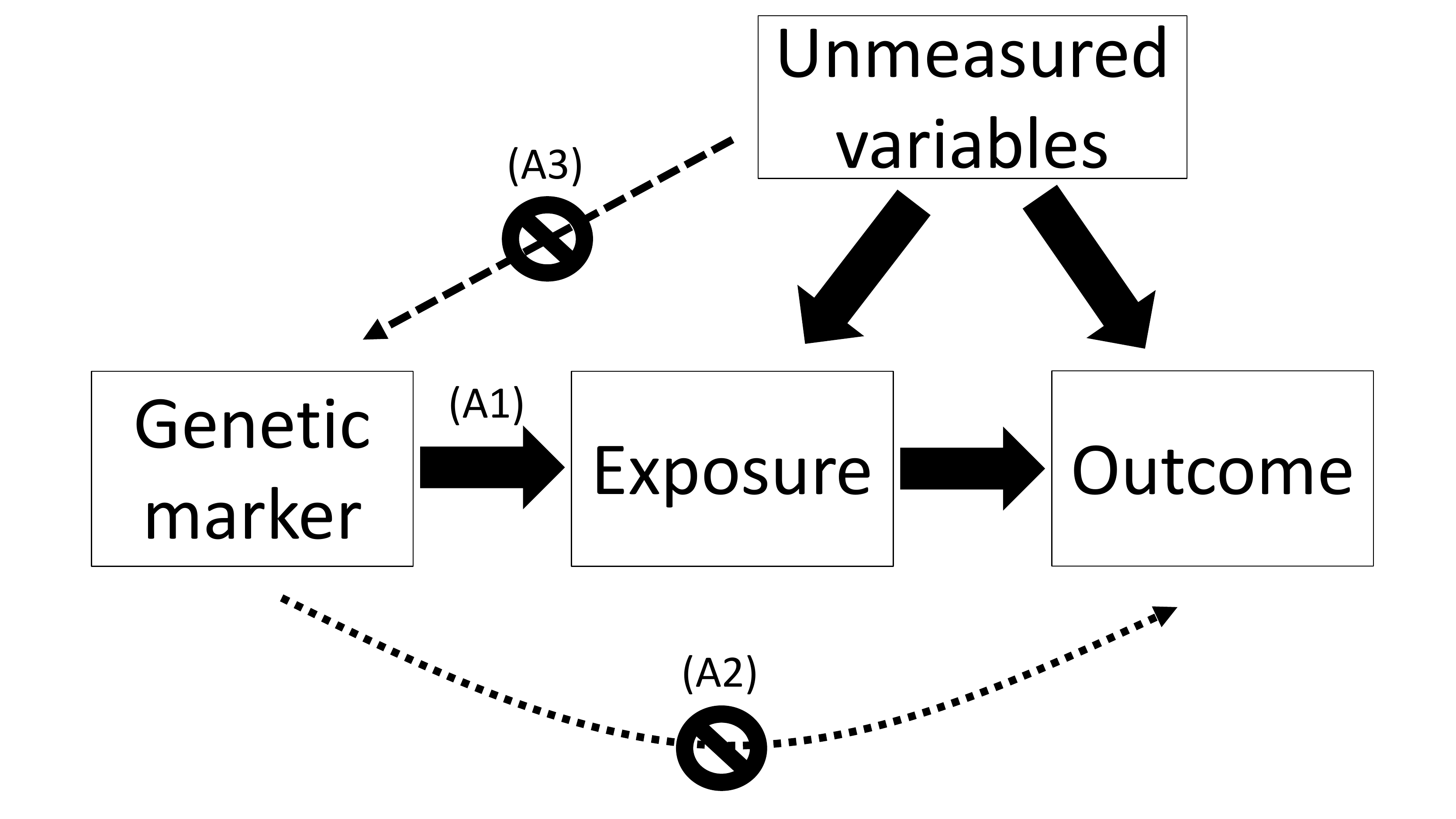}
\caption{Diagram of One Possible Violation of Instrumental Variables Assumptions. Arrows represent associations between variables. Absence of arrows indicates no relationship. Numbers (A1), (A2), and (A3) indicate different instrumental variables assumptions. In this example, the unmeasured variable refers to variables that precede the assignment of the genetic marker, such as population stratification, and the genetic marker is a single nucleotide polymorphism whose value is fixed at birth. As such, the arrows from the unmeasured variable originate from the unmeasured variable and the arrow from the genetic marker goes from it to the outcome (A2) since the genetic marker is fixed at the time of conception.}
\label{fig:dag}
\end{figure}

In MR, the goal is to estimate the causal effect of an exposure on an outcome by using genetic markers, specifically single nucleotide polymorphisms (SNPs), as instruments \citep{smith_mendelian_2003, smith_mendelian_2004, lawlor_mendelian_2008, wehby_mendelian_2008}. For example, \citet{timpson_c-reactive_2005} studied the causal effect of C-reactive protein (CRP), the exposure, on various metabolic outcomes, such as body mass index (BMI) and cholesterol biomarkers (e.g. tryglycerides), using four haplotypes constructed from three SNPs (rs1800947, rs1130864, rs1205) as instruments. The instruments have been previously associated with plasma CRP levels, thereby agreeing with (A1). However, agreement with (A2) and (A3) is less certain. As the authors of the study noted, it is plausible that one or more of the genes that contain the SNPs, rs1800947, rs1130864, and rs1205, may have multiple functions, known as pleiotropy, where, in addition to changing CRP levels (the exposure), the gene containing one of these SNPs would change triglyceride levels or BMI (the outcome) and (A2) would not hold. Indeed, recent work by \citet{martinez_haplotypes_2012} suggested that one of the instruments used, rs1130864, is directly linked to BMI, one of the outcomes, raising doubts about causal estimates when this SNP is assumed to be a valid instrument.

As another example, \citet{katan_apolipoprotein_1986}, in one of the first discussions of MR, proposed to estimate the causal effect of serum cholesterol level on cancer by using the apolipoprotein E polymorphism (APOE)'s effect on serum cholesterol levels. However, as \citet{smith_mendelian_2004} argued, the current knowledge about the APOE gene and its multiple pleiotropic effects on longevity, cholesterol biomarkers, and several other variables, would invalidate the APOE gene as a valid instrument, specifically due to its violation of (A2), and make an IV analysis based on it biased.

Both examples highlight a fundamental limitation with MR studies. For one, pleiotropy and its impact on (A2) is a concern in most MR studies \citep{little_mendelian_2003, smith_mendelian_2003, smith_mendelian_2004, thomas_commentary_2004, brennan_commentary_2004, lawlor_mendelian_2008}. \citet{lawlor_mendelian_2008} also list other biological phenomena associated with genetic instruments such as linkage disequilibrium and population stratification that may violate (A2) and (A3). Unfortunately, verifying genetic instruments as valid IVs requires having complete knowledge of the instruments' biological function and pleoitropic effects. As both examples highlight, the biological understanding of many genetic markers and their potential pleiotropic effects are typically incomplete at the time of the study \citep{solovieff_pleiotropy_2013}. In the face of incomplete biological knowledge and possible instrument invalidity, can valid causal estimates be derived? 

Previous work in IV estimation in the presence of possibly invalid instruments is limited. Traditional instrumental variables literature has stated that to estimate the causal effect of an exposure on an outcome when there are unmeasured confounders, one needs to have at least one instrument that one \emph{knows} is valid \citep{wooldridge_econometrics_2010}. \citet{kolesar_identification_2011} considered the possibility of identifying causal effects when all the instruments are invalid because of direct effects on the outcome. The authors showed that if the direct effects are orthogonal to the instruments' effects on the treatment, then the causal effect can be identified. \citet{kolesar_identification_2011} describes conditions under which this orthogonality is plausible. But, for MR, this stringent structure on the instruments would not hold in most cases as it would mean that the pleiotropic effects of the IVs are orthogonal to the effects of the IVs on the treatment. \citet{gautier_high_2011} analyzed instrumental variables regression in the presence of possibly invalid instruments. However, for their procedure to work, one must have a pre-defined set of known valid instruments. 

Our paper adds to the prior literature as follows. First, we show that it is indeed possible to identify and estimate the causal effect without a known pre-defined set of valid instruments. In particular, under a weaker condition where the proportion of invalid instruments is strictly less than 50\% of the total instruments, we show that identification and estimation is possible. For example, given four possible haplotypes/instruments in the previous example by \citet{timpson_c-reactive_2005}, estimation of the causal effect of CRP on metabolic phenotypes is still possible if no more than one instrument is invalid, without knowing exactly which of the four is invalid. We also show conditions for identification when the 50\% threshold may not hold.

Second, we develop a fast $\ell_1$ estimation procedure to estimate the causal effect of the exposure on the outcome in the presence of possibly invalid instruments. The procedure has provable theoretical guarantees on estimation performance and is computationally as fast as ordinary least squares. The procedure is implemented and available on CRAN as an R package \emph{sisVIVE}, which stands for Some Invalid Some Valid IV Estimator.

Third, we conduct a simulation study that compares our method to two-stage least squares (TSLS), the most popular estimation procedure in IV estimation. We show that our procedure dominates TSLS when the instruments may be invalid. We also conduct a real MR study concerning the effect of BMI on health-related quality of life (HRQL) measure using our new method.

\section{CAUSAL MODEL AND INSTRUMENTAL VARIABLES}
\subsection{Notation}
To define valid instruments, the potential outcomes approach \citep{neyman_application_1990, rubin_estimating_1974} for instruments laid out in \citet{holland_causal_1988} is used. For each individual $i \in \{1,\ldots,n\}$, let $Y_i^{(d,\mathbf{z})} \in \reals$ be the potential outcome if the individual were to have exposure $d \in \reals$ and instruments $\mathbf{z} \in \reals^L$. Let $D_i^{(\mathbf{z})} \in \reals$ be the potential exposure if the individual had instruments $\mathbf{z} \in \reals^L$. For each individual, only one possible realization of $Y_i^{(d,\mathbf{z})}$ and $D_{i}^{(\mathbf{z})}$ is observed, denoted as $Y_i$ and $D_i$, respectively, based on his observed instrument values $\mathbf{Z}_{i.} \in \reals^L$ and exposure $D_{i}$. In total, $n$ sets of outcome, exposure, and instruments, denoted as $(Y_i, D_i,\mathbf{Z}_{i.})$, are observed in an i.i.d. fashion. 

We denote $\mathbf{Y} = (Y_1,\ldots,Y_n)$ to be an $n$-dimensional vector of observed outcomes, $\mathbf{D} = (D_1,\ldots,D_n)$ to be an $n$-dimensional vector of observed exposures, and $\mathbf{Z}$ to be a $n$ by $L$ matrix of instruments where row $i$ consists of $\mathbf{Z}_{i.}$. 

For any vector $\bm{\alpha} \in \reals^L$, let $\alpha_j$ denote the $j$th element of $\bm{\alpha}$. Let $\|\bm{\alpha}\|_1$, $\|\bm{\alpha}\|_2$, and $\|\bm{\alpha}\|_\infty$ be the usual $1,2$ and $\infty$-norms, respectively. Let $\|\bm{\alpha}\|_0$ denote the $0$-norm, i.e. the number of non-zero elements in $\bm{\alpha}$. The support of $\bm{\alpha}$, denoted as supp$(\bm{\alpha}) \subseteq \{1,\ldots,L\}$, is defined as the set containing the non-zero elements of the vector $\bm{\alpha}$, i.e. $j \in$ supp$(\bm{\alpha})$ if and only if $\alpha_j \neq 0$. A vector $\bm{\alpha}$ is called $s$-sparse if it has no more than $s$ non-zero entries. Also, for a vector $\bm{\alpha} \in \reals^L$ and set $A \subseteq \{1,\ldots,L\}$, we denote $\bm{\alpha}_{A} \in \reals^L$ to be the vector where all the elements except whose indices are in $A$ are zero. 

For any $n$ by $L$ matrix $\mathbf{M} \in \reals^{n \times L}$, we denote the $(i,j)$ element of matrix $\mathbf{M}$ as $M_{ij}$, the $i$th row as $\mathbf{M}_{i.}$, and the $j$th column as $\mathbf{M}_{.j}$. Let $\mathbf{M}^T$ be the transpose of $\mathbf{M}$. Let $\mathbf{P}_{\mathbf{M}}$ be the $n$ by $n$ orthogonal projection matrix onto the column space of $\mathbf{M}$, specifically $\mathbf{P}_{\mathbf{M}} = \mathbf{M} (\mathbf{M}^T \mathbf{M})^{-1} \mathbf{M}^T$; it is assumed that $\mathbf{M}^T \mathbf{M}$ has a proper inverse, unless otherwise noted. Let $\mathbf{P}_{\mathbf{M}^\perp}$ be the residual projection matrix, specifically $\mathbf{P}_{\mathbf{M}^{\perp}} = \mathbf{I} - \mathbf{P}_{\mathbf{M}}$ where $\mathbf{I}$ is an $n$ by $n$ identity matrix.

For any sets $A \subseteq \{1,\ldots,L\}$, we denote $A^C$ to be the complement of set $A$. Also,  we denote $|A|$ to be the cardinality of set $A$. 
\subsection{Model}
We consider the Additive LInear, Constant Effects (ALICE) model of \citet{holland_causal_1988} and extend it to allow for multiple valid and possibly invalid instruments as in \citet{small_sensitivity_2007}. Let $d',d \in \reals$ be possible values of the exposure and $\mathbf{z}',\mathbf{z} \in \reals^L$ be possible values of the instruments. Let $\epsilon_i = Y_{i}^{(0,\mathbf{0})} - E[Y_{i}^{(0,\mathbf{0})} | \mathbf{Z}_{i.}]$ and the collection of $\epsilon_i$ be denoted as $\bm{\epsilon} = (\epsilon_1,\ldots,\epsilon_n)$. Suppose we have the following potential outcomes model for the outcome
\begin{align}
Y_{i}^{(d',\mathbf{z}')} - Y_{i}^{(d,\mathbf{z})} &=  (\mathbf{z}' - \mathbf{z})^T \bm{\phi}^* + (d' - d) \beta^* \label{eq:model1} \\
E(Y_i^{(0,\mathbf{0})} | \mathbf{Z}_{i.}) &=  \mathbf{Z}_{i.}^T \bm{\psi}^* \label{eq:model2}
\end{align}
where $\bm{\phi}^*, \bm{\psi}^* \in \reals^L$, and $\beta^* \in \reals$ are unknown parameters. In equation \eqref{eq:model1}, the parameter $\beta^*$ represents the causal parameter of interest, the causal effect on the outcome of changing the exposure by one unit. Also in equation \eqref{eq:model1}, the parameter $\bm{\phi}^*$ represents the direct effect of the instruments on the outcome; changing instruments from $\mathbf{z}'$ to $\mathbf{z}$ results in a direct effect on the outcome of $(\mathbf{z}' - \mathbf{z})^T \bm{\phi}^*$. In equation \eqref{eq:model2}, the parameter $\bm{\psi}^*$ represents the confounders that affect the instrument and the outcome. In particular, without any confounders, there should not be any relationship between the instruments $\mathbf{Z}_{i.}$ and the potential outcome $Y_{i}^{(0,\mathbf{0})}$. Instead, in equation \eqref{eq:model2}, they are related via  $\bm{\psi}^*$.

Let $\bm{\alpha}^* = \bm{\phi}^* + \bm{\psi}^*$. When we combine equations \eqref{eq:model1} and \eqref{eq:model2} along with the definition of $\epsilon_i$, we have the observed data model
\begin{equation} \label{eq:model3}
Y_i =  \mathbf{Z}_{i.}^T \bm{\alpha}^* +  D_i \beta^* + \epsilon_i, \quad{} E(\epsilon_i | \mathbf{Z}_{i.}) = 0
\end{equation}
We make the following remarks regarding the model \eqref{eq:model3}. First, the model can include exogenous measured covariates, say $\mathbf{X}_{i.} \in \reals^p$ which may include the intercept term, and we can replace the variables $Y_i$, $D_i$, and $\mathbf{Z}_{i.}$ with the residuals after regressing them on $\mathbf{X}$ (e.g. replace $\mathbf{Y}$ by $(\mathbf{I} - \mathbf{P}_{\mathbf{X}})\mathbf{Y}$) where $\mathbf{X}$ is the $n$ by $p$ matrix of covariates \citep{wang_inference_1998}. The results in this paper will hold generally when working with such data that is transformed by regressing out the effect of $\mathbf{X}$. In the same spirit, the model can be extended to non-linear models by including appropriate basis transformations of $\mathbf{Z}_{i.}$. However, for simplicity of exposition, we will focus on a model without any measured covariates or non-linear terms. We will also assume that $\mathbf{Y}$, $\mathbf{D}$, and the columns of $\mathbf{Z}$ are centered, which can also result from a residual transformation with $\mathbf{X}$ containing only the intercept term.

Second, following \citet{heckman_alternative_1985}, \citet{bjorklund_estimation_1987}, and \citet{small_sensitivity_2007}, we can incorporate heterogeneous effects as follows. Suppose, instead of equation \eqref{eq:model1}, the potential outcomes model for the outcome is 
\begin{equation} \label{eq:heteromodel1}
Y_{i}^{(d',\mathbf{z}')} - Y_{i}^{(d,\mathbf{z})} =  (\mathbf{z}' - \mathbf{z})^T \bm{\phi}^* + (d' - d) \beta_i^*
\end{equation}
where $\beta^* = E(\beta_i^*)$ is the average effect of the exposure for everyone in the population. Then, the observed data model can be derived from \eqref{eq:heteromodel1} as follows.
\begin{equation} \label{eq:heteromodel2}
Y_i =  \mathbf{Z}_{i.}^T \bm{\alpha}^* +  D_i \beta^* + (\beta_i^* - \beta^*)D_i + \epsilon_i, \quad{} E(\epsilon_i | \mathbf{Z}_{i.}) = 0
\end{equation}
If $(\beta_i^* - \beta^*)$ is independent of $D_i$ given $\mathbf{Z}_{i.}$, the heterogeneous model in \eqref{eq:heteromodel2} is identical to model \eqref{eq:model3} and our result for Theorem \ref{prop:1} in Section \ref{sec:identModel} hold. Also, as \citet{small_sensitivity_2007} notes in page 1055, the assumption that $(\beta_i^* - \beta^*)$ is independent of $D_{i}$ given $\mathbf{Z}_{i.}$ is equivalent to that ``units do not select their treatment levels $D_i$ given $\mathbf{Z}_{i.}$ based on the gains they would experience from treatment $D_i$ given $\mathbf{Z}_{i.}$.'' If this assumption is violated, different groups of people will have different treatment effects, which in turn would lead to possibly non-zero $\bm{\alpha}^*$ (see \citet{angrist_two-stage_1995} and \citet{small_sensitivity_2007} for details). For simplicity of exposition, we'll focus on a model with constant linear effect $\beta^*$.

\subsection{Definition of Valid Instruments}
Based on the observed model in \eqref{eq:model3}, the parameter $\bm{\alpha}^*$ combines both the direct effect, represented by $\bm{\phi}^*$, and the effect of confounders on the $\mathbf{Z}_{i.}$ and $Y_i^{(0,0)}$ relationship, represented by $\bm{\psi}^*$. If there is no direct effect and no effect of the confounders, then $\bm{\alpha}^* = 0$. Hence, the value of $\bm{\alpha}^*$ captures the notion of valid and invalid instruments. The definition below formalizes this idea: 
\begin{definition} \label{def:validIV} Suppose we have the models in \eqref{eq:model1} -\eqref{eq:model3} with $L$ instruments. We say instrument $j \in \{1,\ldots,L\}$ is valid if $\alpha_j^* = 0$ and invalid if $\alpha_j^* \neq 0$.
\end{definition}
Definition \ref{def:validIV} distinguishes valid and invalid instruments based on supp$(\bm{\alpha}^*)$, the support of $\bm{\alpha}^*$. If instrument $j = 1,\ldots,L$ is not in the support, it is valid. If the instrument is in the support of $\bm{\alpha}^*$, it is invalid. Consequently, not knowing which instruments are valid and invalid directly translates to not knowing the support of $\bm{\alpha}^*$ in model \eqref{eq:model3}. 

In the case of only one instrument (i.e. $L = 1$), Definition \ref{def:validIV} of a valid instrument matches with the informal definition (A2) and (A3) in the Introduction and the formal definition in \citet{holland_causal_1988}. Specifically, the notion of exclusion restriction (A2), $Y_{i}^{(d,z)} = Y_{i}^{(d,z')}$ for all $z, z' \in \reals$ is equivalent to the parameter $\phi^*$ in equation \eqref{eq:model1} being zero. Also, the assumption of no unmeasured confounding of the IV-outcome relationship (A3) where $Y_{i}^{(d,z)}$ and $D_i^{(z)}$ are independent of $Z_i$ for all $d,z \in \reals$, is encoded by $\psi^*$ in \eqref{eq:model2} being zero. Hence, $\phi^* = \psi^* = 0$, which implies $\alpha^* = 0$ and a valid IV in \citet{holland_causal_1988} is also a valid IV in our definition. Also, for one instrument, our model and definition is a special case of the definition of valid instrument discussed in \citet{angrist_identification_1996} where our model assumes an additive, linear, and constant treatment effect $\beta^*$.

For more than one instruments (i.e. $L > 1$), our model \eqref{eq:model1}-\eqref{eq:model3} and definition of valid IVs can be viewed as a generalization of \citet{holland_causal_1988}. It is important to note that in this generalization, Definition \ref{def:validIV} defines the validity of an instrument $j$ in the context of the set of instruments $\{1,\ldots,L\}$ being considered. Specifically, an instrument $j$ could be valid in the context of the set $\{1,\ldots,L\}$ (i.e. $\alpha_j^* = 0$), but invalid if considered alone because $\mathbf{Z}_{.j}$ may be associated with or causally affect another IV $\mathbf{Z}_{.j'}$, $j \neq j'$ where $\alpha_{j'}^* \neq 0$.  

\section{ESTIMATION OF CAUSAL EFFECT}
\subsection{Identifiability of Model} \label{sec:identModel}
We first address whether the model in equation \eqref{eq:model3} is identifiable, that is whether we can estimate the unknown parameters if we were given infinite data, even without any knowledge about which instruments are valid and invalid. We begin by making the assumptions.
\begin{enumerate}
\item[(a)] $E(\mathbf{Z}^T \mathbf{Z})$ is full rank;
\item[(b)] For $E(\mathbf{Z}^T \mathbf{D}) = E(\mathbf{Z}^T \mathbf{Z}) \bm{\gamma}^*$, the components of $\bm{\gamma}^*$ are all not equal to zero, i.e. $\gamma_j^* \neq 0$ for $j=1,\ldots, L$.
\end{enumerate} 
Assumption (a) states that the matrix of instruments $\mathbf{Z}$ is full rank, a common assumption in the instrumental variables literature \citep{wooldridge_econometrics_2010}. Assumption (b) states that the instruments are associated with the exposure, akin to assumption (A1), that the instruments are relevant to the exposure; note that there does not need to be a causal relationship between the instrument $\mathbf{Z}$ and the exposure $\mathbf{D}$, just an association \citep{hernan_instruments_2006,didelez_mendelian_2007, glymour_credible_2012}. As one reviewer remarked, assumption (b) requires that all $L$ instruments are related to the exposure, $\gamma_j^* \neq 0$ for all $j$. If we have instruments that are not relevant to the exposure, $\gamma_j^* = 0$, we can exclude them from further analysis and concentrate only on those instruments that affect the exposure. 

Now, the model in \eqref{eq:model3} implies the following moment condition.
\begin{equation} \label{eq:momentCond1}
E(\mathbf{Z}^T(\mathbf{Y} - \mathbf{Z}\bm{\alpha}^{*} - \mathbf{D}\beta^{*})) = 0 
\end{equation}
Suppose the assumptions (a) and (b) hold. Then, the moment equation in equation \eqref{eq:momentCond1} simplifies to 
\begin{equation} \label{eq:momentCond2}
\bm{\Gamma}^* = \bm{\alpha}^* + \bm{\gamma}^* \beta^*
\end{equation}
where $\bm{\Gamma}^* = E(\mathbf{Z}^T \mathbf{Y}) E(\mathbf{Z}^T \mathbf{Z})^{-1}$. Since both $\bm{\Gamma}^*$ and $\bm{\gamma}^*$, defined by (b), can be identified by their moments based on observed data $E(\mathbf{Z}^T \mathbf{Y}) E(\mathbf{Z}^T \mathbf{Z})^{-1}$ and $E(\mathbf{Z}^T \mathbf{D}) E(\mathbf{Z}^T \mathbf{Z})^{-1}$, respectively, $\bm{\alpha}^*$ and $\beta^*$ are identified if we can find a bijective mapping between $\bm{\alpha}^*, \beta^{*}$ and $\bm{\Gamma}^*, \bm{\gamma}^*$, i.e. a unique solution $\bm{\alpha}^*, \beta^*$ given $\bm{\Gamma}^*, \bm{\gamma}^*$. 

If we know exactly which instruments are invalid $A^* = $ supp$(\bm{\alpha}^*) = \{j : \alpha_j^* \neq 0\}$ and hence, know the set of valid instruments $(A^*)^C = \{j : \alpha_j^* = 0\}$, equation \eqref{eq:momentCond2} becomes
\[
\bm{\alpha}_{(A^*)^C} + \bm{\gamma}_{(A^*)^C}^* \beta^* =\bm{\gamma}_{(A^*)^C}^* \beta^* = \Gamma_{(A^*)^C}^*
\]
There is a unique $\beta^*$ so long as $|(A^*)^C| > 0$, or there is at least one known valid instrument. This is a special case of the classic identification result for linear simultaneous equation models \citep{koopmans_measuring_1950}.

If we know that there is a valid instrument, but are not sure of the identity of the valid instrument(s), then a unique solution to \eqref{eq:momentCond2} and hence, identification, is not guaranteed. For example, let there be four instruments, $L = 4$ with $\bm{\gamma}^* = (1,2,3,4)$ and $\bm{\Gamma}^* = (1,2,3,8)$. Then, depending on the set of valid instruments $(A^*)^C$, which is unknown, we have two different $\beta^*$ that satisfy equation \eqref{eq:momentCond2}. If the set of valid instruments $(A^*)^C$ is $(A^*)^C = \{1,2,3\}$, we have $\bm{\gamma}_{(A^*)^C}^* \beta^* = \Gamma_{(A^*)^C}^*$ and $\beta^* = 1$. However, if the set of valid instruments is $(A^*)^C = \{4\}$, $\beta^* = 2$. Without knowing exactly which $(A^*)^C$ is the true set of valid instruments, we can't choose between the two $\beta^*$s and hence, there is not a unique solution to \eqref{eq:momentCond2}. 

But, suppose we impose constraints on $A^*$. Specifically, suppose the number of invalid instruments, $s = |A^*|$, has to be less than some number $U$, $s < U$, without knowing which instruments are invalid or knowing exactly the number of invalid instruments. For example, geneticists may have a rough idea on the maximum number of invalid instruments, $U$, but not know exactly the number of invalid instruments nor do they know exactly which instruments are invalid. Note that this condition of knowing the maximum number of invalid instruments is a much weaker requirement than what is traditionally required in IV and MR literature where one must know exactly which instruments are invalid, i.e. know exactly the set $A^*$; here, we only need an upper bound on the cardinality of $A^*$. Under the weaker condition $s < U$, a unique solution to \eqref{eq:momentCond2} can exist and this is stated in Theorem \ref{prop:1}.

\begin{theorem}[Uniqueness of Solution] \label{prop:1}
Suppose we assume assumptions (a) and (b) and the modeling assumption \eqref{eq:model3}. Let $s \in \{0,1,\ldots,L\}$ with $s < U$ where $U = 1,\ldots,L$. 
Consider all sets $C_m \subseteq \{1,\ldots,L\}, m=1,\ldots,M$ of size $|C_m| = L - U +1$ with the property 
\[
\gamma_j^* q_m = \Gamma_j^* \quad{} j\in C_m
\] 
where $q_m$ is a constant. There is a unique solution $\bm{\alpha}^*$ and $\beta^*$ to \eqref{eq:momentCond2} if and only if $q_m = q_{m'}$ for all $m, m' \in \{1,\ldots,M\}$. 
\end{theorem}
To understand Theorem \ref{prop:1}, note that if the valid instruments are those in the set $C_m$, then the causal effect $\beta^* =q_m$. Theorem \ref{prop:1} says that $\beta^*$ is identified as long as there are not two subsets of the instruments of cardinality $L - U + 1$ that give internally consistent estimates of $\beta^*$ (i.e. all instruments in each subset give the same estimate of $\beta^*$), but are externally inconsistent (i.e. the estimates of $\beta^*$ from the two subsets are different). We call the property in Theorem \ref{prop:1} that there is a unique solution to $\bm{\alpha}^*$ and $\beta^*$ to \eqref{eq:momentCond2} if and only if $q_m = q_{m'}$ for all $m, m' \in \{1,\ldots,M\}$ the \emph{consistency criterion}. We thank Jack Bowden for his insight and suggestions on terminology for interpreting Theorem \ref{prop:1}.

As an example of applying Theorem \ref{prop:1}, consider our numerical example above with $\bm{\gamma}^* = (1,2,3,4)$ and $\bm{\Gamma}^* = (1,2,3,8)$ and $U = 3$. Then, by Theorem \ref{prop:1} we have 3 sets $C_1 = \{1,2\}, C_2 = \{1,3\}, C_3 = \{2,3\}$ with $q_1 = q_2 = q_3 = 1$. Hence, $\bm{\gamma}^*$ and $\bm{\Gamma}^*$ satisfy the consistency criterion of Theorem \ref{prop:1} and we have a unique solution $\bm{\alpha}^*$ and $\beta^*$ to \eqref{eq:momentCond2}. In contrast, if $\bm{\gamma}^* = (1,2,3,4)$ and $\bm{\Gamma}^* = (1,2,6,8)$, we would have two sets $C_1 = \{1,2\}, C_2 = \{3,4\}$ with $q_1 = 1$ and $q_2=2$, respectively. These $\bm{\gamma}^*$ and $\bm{\Gamma}^*$ do not satisfy the consistency criterion of Theorem \ref{prop:1} because $q_1 \neq q_2$ and there are no unique solutions $\bm{\alpha}^*$ and $\beta^*$ to \eqref{eq:momentCond2}. Further discussion of this particular example is discussed in the Supplementary Materials along with discussion of the implications of Theorem \ref{prop:1} when the additional linearity and normality assumptions of the classical linear simultaneous/structural equation model \citep{koopmans_measuring_1950} are considered.

Checking the consistency criterion can be computationally difficult, especially if $U$ is large; it requires looking at ${L \choose L - U +1}$ possible subsets of $\{1,\ldots,L\}$ and the constants $q_m$ associated with $\bm{\Gamma}^*$ and $\bm{\gamma}^*$. Corollary \ref{coro:0} says that the consistency criterion is automatically satisfied if $U \leq L/2$ (i.e. if 50\% of the total candidate of $L$ instruments are invalid) regardless of the values of $\bm{\gamma}^*$ and $\bm{\Gamma}^*$.
\begin{corollary} \label{coro:0} If $U \leq L/2$, there is always a unique solution to \eqref{eq:momentCond2}
\end{corollary}
In addition to the computational benefits, compared to Theorem \ref{prop:1}, Corollary \ref{coro:0} is simpler to interpret. For example, for a geneticist, without knowing the entire biology of genetic instruments, specifically knowing which instruments are valid and invalid, as long as the number of invalid instruments is less than 50\% of the total instruments, then the geneticist can rest assured that the parameters can always be identified. If this is not the case, the geneticist can always check the consistency criterion stated in Theorem \ref{prop:1}.

We would like to mention two final points about Theorem \ref{prop:1}. First, Theorem \ref{prop:1} is a statement about uniqueness of solutions for the parameters $\bm{\alpha}^*$, and $\beta^*$ in equation \eqref{eq:momentCond2}. A natural question to ask is whether the uniqueness is guaranteed for just $\beta^*$, the causal effect of interest, at the expense of non-uniqueness of $\bm{\alpha}^*$. In the proof of Theorem \ref{prop:1}, we show that this cannot be the case. Specifically, regardless of the condition on $s$, the parameter $\beta^*$ is a unique solution to \eqref{eq:momentCond2} if and only if the parameter $\bm{\alpha}^*$ is a unique solution to \eqref{eq:momentCond2}. Second, Theorem \ref{prop:1} supposes the existences of the sets $C_m$ and proceeds to compare their corresponding $q_m$. However, one may ask whether these sets $C_m$ even exist in the first place. In the proof of Theorem \ref{prop:1}, we provide a rigorous argument that, indeed, under model \eqref{eq:model3} and $s < U$, at least one set $C_m$ has to exist.

\subsection{Estimation of the Causal Effect of Exposure on Outcome} \label{sec:est}
Given the model \eqref{eq:model3} and $s < U$, Theorem \ref{prop:1} lays out the sufficient and necessary condition for finding a unique solution to the moment equation \eqref{eq:momentCond1}. Specifically, if the model is identified, the moment equation \eqref{eq:momentCond1} is zero at exactly one value, the true value of $\bm{\alpha}^*$ and $\beta^*$. Naturally then, a method to estimate the one true value is to find the values of $\bm{\alpha}^*$ and $\beta^*$ that minimize \eqref{eq:momentCond1} subject to the parameter constraint that $s < U$. Formally, we can write this estimation strategy as 
\begin{equation} \label{eq:L0}
\amin{\bm{\alpha},\beta} \frac{1}{2} \|\mathbf{P}_{\mathbf{Z}} (\mathbf{Y} - \mathbf{Z}\bm{\alpha} - \mathbf{D}\beta)\|_2^2, \quad{} s.t. \quad{}  ||\bm{\alpha}||_0 < U
\end{equation}
where $||\bm{\alpha}||_0$ is the number of non-zero entries of $\bm{\alpha}$ and by Definition \ref{def:validIV}, $s = ||\bm{\alpha}||_0$. However, it is computationally infeasible to go through all subsets of size less than $U$ and this type of problem has been shown to be NP-hard \citep{natarajan_sparse_1995}. Instead, a computationally tractable version of estimation strategies like \eqref{eq:L0} has been proposed in the literature using a convex surrogate of the $\ell_0$ norm \citep{candes_decoding_2005, tropp_relax_2006, donoho_most_2006}. Specifically, the computationally feasible version of the estimation strategy in \eqref{eq:L0} can be written as
\begin{equation} \label{eq:L1}
\amin{\bm{\alpha},\beta} \frac{1}{2} \|\mathbf{P}_{\mathbf{Z}} (\mathbf{Y} - \mathbf{Z}\bm{\alpha} - \mathbf{D}\beta)\|_2^2, \quad{} s.t. \quad{}  ||\bm{\alpha}||_1 \leq t
\end{equation}
where the $\ell_0$ norm is replaced by the convex norm $\ell_1$ and $U$ is replaced by a user-specified tuning parameter $t > 0$. In this paper, we propose the equivalent Lagrangian form as our estimator of the causal effect, called \emph{some invalid some valid IV estimator}, or sisVIVE, as follows
\begin{equation} \label{eq:proceduresimple}
(\hat{\bm{\alpha}}_{\lambda},\hat{\beta}_{\lambda}) \in \amin{\bm{\alpha},\beta} \frac{1}{2} \|\mathbf{P}_{\mathbf{Z}} (\mathbf{Y} - \mathbf{Z}\bm{\alpha} - \mathbf{D}\beta)\|_2^2 + \lambda\|\bm{\alpha}\|_1
\end{equation}
for some tuning parameter $\lambda > 0$ where $\lambda$ corresponds to $t$ in \eqref{eq:L1}. If $\lambda = 0$ in \eqref{eq:proceduresimple}, then \eqref{eq:proceduresimple} is the popular two stage least squares (TSLS) estimator, which is equivalent to the GMM estimator when the $\bm{\epsilon}$ are assumed to be homoscedastic \citep{hansen_large_1982}. Hence, sisVIVE can be viewed as a generalization of TSLS or GMM.

sisVIVE also bears some resemblance to the traditional $\ell_1$ penalization procedure, in particular the Lasso \citep{tibshirani_regression_1996} or the recent $\ell_1$ penalty procedures in IV estimation by \citet{gautier_high_2011} and \citet{belloni_sparse_2012}. However, there are a few important differences. First, with regards to traditional Lasso and the procedure proposed by \citet{gautier_high_2011}, our procedure in \eqref{eq:proceduresimple} only penalizes $\bm{\alpha}^*$. The estimator \eqref{eq:proceduresimple} does not penalize $\beta^*$, the causal effect of the exposure on the outcome, because the causal effect may be far from zero. In contrast, the prior works we mentioned penalize all the parameters in the model. Second, the traditional Lasso only considers regression with all exogenous regressors, which are regressors that are assumed to be independent of the error term or assumed to be fixed. The regressors in our model \eqref{eq:model3} are not all exogenous; specifically, model \eqref{eq:model3} contains one random endogenous variable, $D_i$, which is dependent on the error term. Third, \citet{gautier_high_2011} and \citet{belloni_sparse_2012} assume that either all the $L$ instruments are valid or we know exactly which subset of them are valid. In contrast, our procedure does not assume this.

Finally, a careful reader may have recognized that there may be multiple minimizers to the equation \eqref{eq:proceduresimple}, specifically $\hat{\beta}_\lambda$, because $||\bm{\alpha}||_1$ is not strictly convex and hence, we use the set notation instead of the equality sign in \eqref{eq:proceduresimple}. This might seem to be a concern as there are multiple estimates of $\beta^*$. However, as we will show in Section \ref{sec:estTheory}, all minimizers of \eqref{eq:proceduresimple} are close to the true values $\beta^*$. Also, if the entries of the matrix $\mathbf{P}_{\hat{\mathbf{D}}^\perp} \mathbf{Z}$ where $\hat{\mathbf{D}} = \mathbf{P}_{\mathbf{Z}}\mathbf{D}$ (i.e. the predicted value of the exposure given the instruments) are drawn from a continuous distribution, then the solution to \eqref{eq:proceduresimple} is unique \citep{tibshirani_lasso_2013}. 

Without loss of generality, we assume that the columns of $\mathbf{Z}$ are scaled to unit length. This allows all $L$ instruments to have identical units so no columns of $\mathbf{Z}$ gets unfairly penalized by the penalty term in \eqref{eq:proceduresimple} simply due to their original units.

\subsection{Choice of $\lambda$} \label{sec:lambdaCV}
Like many penalization procedures, the choice of the tuning parameter $\lambda$ affects the performance of the estimation procedure and this is certainly the case with sisVIVE. High values of $\lambda$ force heavy penalization on $\bm{\alpha}$, which will put most elements of $\bm{\hat{\alpha}}_{\lambda}$ to zero and most instruments will be estimated as valid instruments. In contrast, low values of $\lambda$ will put few elements of $\bm{\hat{\alpha}}_{\lambda}$ to zero and most instruments will be estimated as invalid instruments. In short, the optimal choice of $\lambda$ depends on knowing the exact number of invalid and valid instruments, something not implied by the condition $s < U$.

In practice, cross validation is a popular data-driven method to choose $\lambda$. In the same spirit, we use a $K$-fold cross validation where we minimize the estimating equation $||\mathbf{P}_{\mathbf{Z}}(\mathbf{Y} - \mathbf{Z} \bm{\alpha} - \mathbf{D} \beta)||_2$ instead of the predictive error $||(\mathbf{Y} - \mathbf{Z} \bm{\alpha} - \mathbf{D} \beta)||_2$. We minimize the estimating equation instead of the predictive error since the parameter of interest is the causal effect $\beta^*$ that sets the expected value of the estimating equation to zero (see equation \eqref{eq:momentCond1}, Sections \ref{sec:identModel} and \ref{sec:est}). We use the ``one standard error'' rule used in most cross-validation procedures \citep{hastie_elements_2009} and choose the smallest $\lambda$ that is no more than one standard error above the minimum of the estimating equation. In Section \ref{sec:sim}, we discuss the performance of $\hat{\beta}_{\lambda_{cv}}$, where $\lambda_{cv}$ is the cross-validated $\lambda$ based on the estimating equation through various simulation studies. Also, in the Supplementary Materials, we discuss another method of choosing $\lambda$, in particular, choosing $\lambda$ based on the theoretical guidance from Theorem \ref{prop:2} and Corollary \ref{coro:1}. In short, the Supplementary Materials show that for better estimation performance of $\hat{\beta}_{\lambda}$, it is important not to incorrectly set invalid IVs to be valid (i.e. let $\hat{\alpha}_j$ to be zero when the true  $\alpha_j^*$ is not zero), while the reverse is not as important. This observation argues for choosing $\lambda$ that tends to set relatively few elements of $\hat{\bm{\alpha}}_{\lambda}$ to be zero and in the Supplementary Materials, we demonstrate that cross validation achieves this goal in a wide variety of settings.

%%%%%%%%%%%%%%%%%%%%%%%
\subsection{Estimation Performance} 
\label{sec:estTheory}
%%%%%%%%%%%%%%%%%%%%%%%

How well does sisVIVE estimate the causal effect $\beta^*$? In order to analyze the performance of sisVIVE, we first introduce some basic notations and definitions. 

\begin{definition}
 For any matrix $\mathbf{M}$, the upper and lower restricted isometry property (RIP) constants of order $k$, denoted as $\delta_k^+(\mathbf{M})$ and $\delta_k^-(\mathbf{M})$ respectively, are the smallest $\delta_k^+(\mathbf{M})$ and largest $\delta_k^-(\mathbf{M})$ such that
\begin{equation}\label{eq:RIP}
\delta_k^-(\mathbf{M}) \|\bm{\alpha}\|_2^2\leq \|\mathbf{M} \bm{\alpha}\|_2^2\leq \delta_k^+(\mathbf{M})\|\bm{\alpha}\|_2^2
\end{equation}
 holds for all $k$-sparse vectors $\bm{\alpha}$.
\end{definition}
RIP conditions have been widely used in the literature on compressed sensing and high-dimensional linear regression.  See \cite{cai_compressed_2013} and the references therein. The following theorem characterizes the performance of sisVIVE in finite samples using the RIP conditions. Note that this characterizes all the minimizers $\hat{\beta}_{\lambda}$ from sisVIVE in \eqref{eq:proceduresimple}.

\begin{theorem}[Estimation performance of sisVIVE] 
Suppose we have the model given in \eqref{eq:model3}. Let $\hat{\mathbf{D}} = \mathbf{P}_{\mathbf{Z}}  \mathbf{D}$. Let the restricted isometry constants $\delta^{+}_{2s}(\mathbf{Z})$, $\delta^-_{2s}(\mathbf{Z})$, $\delta_{2s}^+(\mathbf{P}_{\hat {\mathbf{D}}}\mathbf{Z})$ be defined as in \eqref{eq:RIP}, where $s$ is the number of invalid instruments. Suppose
\begin{equation}
\label{RIP.Cond.}
2\delta^-_{2s}(\mathbf{Z}) > \delta_{2s}^+(\mathbf{Z}) + 2\delta_{2s}^+(\mathbf{P}_{\hat{\mathbf{D}}}\mathbf{Z})
\end{equation}
holds, then the estimate $\hat{\beta}_{\lambda}$ given by \eqref{eq:proceduresimple} with tuning parameter $\lambda \ge 3 \|\mathbf{Z}^T \mathbf{P}_{\hat{\mathbf{D}}^\perp}  \bm{\epsilon}\|_\infty $ has the following performance guarantee 
\begin{equation} \label{eq:boundRIP}
|\hat{\beta}_{\lambda} - \beta^*| \leq \frac{ | \hat{\mathbf{D}}^T \bm{\epsilon} |}{\|\hat{\mathbf{D}}\|_2^2} + \frac{1}{\|\hat{\mathbf{D}}\|_2} \left( \frac{(4/3\sqrt{5}) \lambda \sqrt{s\delta_{2s}^+(\mathbf{P}_{\hat{\mathbf{D}}}\mathbf{Z})}}{2\delta_{2s}^-(\mathbf{Z}) - \delta_{2s}^+(\mathbf{Z}) - 2\delta_{2s}^+(\mathbf{P}_{\hat{\mathbf{D}}}\mathbf{Z})}\right).
\end{equation}
\label{prop:2}
\end{theorem}

Condition \eqref{RIP.Cond.} includes the RIP constants, $\delta_{2s}^-(\mathbf{Z})$, $\delta_{2s}^+(\mathbf{Z})$, and $\delta_{2s}^+(\mathbf{P}_{\hat{\mathbf{D}}}\mathbf{Z})$. Unfortunately, these RIP constants in \eqref{RIP.Cond.} are difficult to evaluate. Hence, in some applications, it is more convenient to use a slightly stronger but much simpler and interpretable condition called the ``mutual incoherence property" (MIP). Specifically, let $\hat{\mathbf{D}} = \mathbf{P}_{\mathbf{Z}}  \mathbf{D}$ and $\|\mathbf{Z}_{.j}\|_2 = 1$ for all $j = 1,\ldots,L$. Define the constants $\mu$ and $\rho$ as 
\begin{equation}
\label{MIP.Cond}
\mu =  \max_{ i \neq j} |\mathbf{Z}_{.i}^T \mathbf{Z}_{.j} |\quad\mbox{and} \quad \rho =  \max_{j} | \hat{\mathbf{D}}^T \mathbf{Z}_{.j} | /  \|\hat{\mathbf{D}}\|_2. 
\end{equation}
First, the constant $\mu$ measures the maximum correlation between any two columns of the matrix of instruments $\mathbf{Z}$. This is related to Assumption (a) in Section \ref{sec:identModel} where a full rank $\mathbf{Z}$ means the columns of $\mathbf{Z}$ are linearly independent. In fact, if $\mu < 1/(L-1)$, $\mathbf{Z}$ is full rank. Second, the constant $\rho$ measures the maximum strength of individual instruments. A high $\rho$ doesn't necessarily imply that all $L$ instruments are individually strong; it just implies that one of the $L$ instruments is strong (i.e. has a high correlation to $\mathbf{D}$); it's possible that the rest of the $L - 1$ instruments are weak. This notion of strength by $\rho$ is slightly different than the concentration parameter, which measures the overall strength of all the $L$ instruments (see Section 4 for details). Also, $\rho$ stands in contrast to Condition (b) in Theorem \ref{prop:1} which looks at the individual values of $\gamma_j, j=1,\ldots,L$, instead of the maximum of $\gamma_j$s. 

Given the two MIP constants $\mu$ and $\rho$, we have the following result on estimation performance. Like Theorem \ref{prop:2}, Corollary \ref{coro:1} characterizes all the minimizers $\hat{\beta}_{\lambda}$ from sisVIVE in \eqref{eq:proceduresimple}.
\begin{corollary}[Estimation performance of sisVIVE under MIP]  
Let the MIP constants $\mu$ and $\rho$ be given in \eqref{MIP.Cond}. If the number of invalid instruments, $s$, satisfies
\begin{equation} \label{eq:constraintMIP}
s < \min( \frac{1}{12 \mu}, \frac{1}{10 \rho^2} )
\end{equation}
the estimate $\hat{\beta}_{\lambda}$ given by \eqref{eq:proceduresimple} with tuning parameter $\lambda\ge 3 \|\mathbf{Z}^T \mathbf{P}_{\hat{\mathbf{D}}^\perp}  \bm{\epsilon}\|_\infty$ has the following performance guarantee 
\begin{equation} \label{eq:boundMIP}
|\hat{\beta}_{\lambda} - \beta^*| \leq \frac{ | \hat{\mathbf{D}}^T \bm{\epsilon} |}{\|\hat{\mathbf{D}}\|_2^2} + \frac{1}{\|\hat{\mathbf{D}}\|_2} \left( \frac{4 \sqrt{105}/9 \lambda s \rho}{1 - s(5 \rho^2 + 6 \mu)}\right).
\end{equation}
\label{coro:1}
\end{corollary}
We make the following remarks. First, in the Supplementary Materials, we show the condition in equation \eqref{eq:constraintMIP} directly implies the condition in equation \eqref{eq:RIP}. We also provide an example of a matrix of instruments $\mathbf{Z}$ where the RIP condition is satisfied, but the MIP condition is not satisfied. Second, the constraint on the number of invalid instruments, $s$, in Corollary \ref{coro:1} is strict, but is required to precisely characterize the bound on estimation performance. As two reviewers pointed out, if the instruments are even slightly correlated at $\mu = 0.1$, $s < 10/12$, no invalid instruments are allowed, and Corollary \ref{coro:1} is not useful in characterizing the performance of sisVIVE. In Section \ref{sec:sim} and in the Supplementary Materials, we study the behavior of sisVIVE when this constraint in \eqref{eq:constraintMIP} may not hold. Third, in the case where all the instruments are uncorrelated with each other so that $\mu = 0$, a small $\rho$ provides a less restrictive upper bound on $s$. At first glance, this may be counterintuitive since a small $\rho$ implies that all the instruments' individual correlation to the exposure is weak and, therefore, having weak instruments allow one to have more invalid instruments. However, we note that the denominator of the bound \eqref{eq:boundMIP}, specifically $\|\hat{\mathbf{D}}\|_2^2$ is a function of the correlation of the instruments, and having a small $\rho$ would translate to having a small $\|\hat{\mathbf{D}}\|_2^2$. Hence, even though the condition \eqref{eq:constraintMIP} allows for more invalid instruments, the upper bound \eqref{eq:boundMIP} becomes worse and our estimator $\hat{\beta}_\lambda$ will be far from $\beta^*$. Finally, we emphasize that the conditions in both Theorem \ref{prop:2} and Corollary \ref{coro:1} are sufficient, but not necessary conditions for the performance bounds to hold. In particular, a violation of these conditions does not imply that sisVIVE will perform badly (see Section 4 and the Supplementary Materials).

\subsection{Fast Numerical Algorithm} \label{sec:computation}
In addition to the theoretical guarantees on estimation performance, in practice, a fast, scalable numerical algorithm for estimation is desirable, especially for MR where genetic data can be large. Theorem \ref{prop:3} outlines a two-step numerical method whose solution is identical to sisVIVE in \eqref{eq:proceduresimple}, but is as fast as ordinary least squares.
\begin{theorem}[Fast two-step numerical algorithm] \label{prop:3}
Let $\mathbf{P}_{\hat{\mathbf{D}}}$ be the projection matrix onto the vector $\hat{\mathbf{D}}$ and $\mathbf{P}_{\hat{\mathbf{D}}^{\perp}} = \mathbf{I} - \mathbf{P}_{\hat{\mathbf{D}}}$. We propose the two-step algorithm as follows.
\begin{enumerate}
\item[]Step 1: For a given $\lambda > 0$, solve: 
\[
\hat{\bm{\alpha}}_{\lambda} \in \amin{\bm{\alpha}} \frac{1}{2}||\mathbf{P}_{\hat{\mathbf{D}}^{\perp}}\mathbf{P}_{\mathbf{Z}} \mathbf{Y} - \mathbf{P}_{\hat{\mathbf{D}}^{\perp}}\mathbf{Z}\bm{\alpha} ||_2^2 + \lambda ||\bm{\alpha}||_1
\]
\item[]Step 2: Use $\hat{\bm{\alpha}}_{\lambda}$ from Step 1 to estimate $\hat{\beta}_{\lambda}$ by 
\[
\hat{\beta}_{\lambda} = \frac{ \hat{\mathbf{D}}^T (\mathbf{Y} - \mathbf{Z}\hat{\bm{\alpha}}_{\lambda})}{||\hat{\mathbf{D}}||_2^2} 
\]
\end{enumerate}
The solution to the two-step algorithm is identical to the solution to sisVIVE in \eqref{eq:proceduresimple}
\end{theorem}
In the two-step algorithm, step 1 is the standard Lasso problem with outcome $\mathbf{P}_{\hat{\mathbf{D}}^{\perp}}\mathbf{P}_{\mathbf{Z}} \mathbf{Y}$ and $\mathbf{P}_{\hat{\mathbf{D}}^{\perp}}\mathbf{Z}$; remember, sisVIVE in \eqref{eq:proceduresimple} is not the standard Lasso problem as discussed in Section \ref{sec:est}. Fast algorithms for the Lasso exist, most notably LARS \citep{efron_least_2004}. In fact, LARS is able to solve $\hat{\bm{\alpha}}_{\lambda}$ for all values of $\lambda >0$ at the same computational efficiency as ordinary least squares. Step 2 is also numerically efficient, requiring a simple dot product operation between $\hat{\mathbf{D}}$ and $\mathbf{Y} - \mathbf{Z}\hat{\bm{\alpha}}_{\lambda}$. Thus, the proposed two-step algorithm is, practically speaking, as fast as ordinary least squares. Best of all, the estimate from this two-step algorithm is identical to sisVIVE.

\section{SIMULATION STUDY} \label{sec:sim}
We conduct various simulation studies to study the estimation performance, measured by $|\hat{\beta}- \beta^*|$, for different methods. Specifically, we compare sisVIVE with TSLS, the most popular estimator in IV and MR, and ordinary least squares (OLS) under various settings that vary the instruments' absolute/overall and relative strength, their validity and correlation among each other, and endogeneity.

Let there be $n = 2000$ individuals and $L = 10$ potential candidate instruments. The observations $(Y_i,D_i,\mathbf{Z}_{i.}), i=1,\ldots,n$ are generated by
\begin{align*}
\left.
\begin{array}{r@{\mskip\thickmuskip}l}
Y_i &=  \pi^* + \mathbf{Z}_{i.}^T \bm{\alpha}^* +  D_i \beta^* + \epsilon_i \\
D_i &= \gamma_0^* + \mathbf{Z}_{i.}^T \bm{\gamma}^* + \xi_i 
\end{array}
\quad{},
\begin{array}{r@{\mskip\thickmuskip}l}
\begin{pmatrix}
\epsilon_i \\
\xi_i 
\end{pmatrix} &\iid N \left( \begin{bmatrix} 0 \\ 0 \end{bmatrix}, \begin{bmatrix} 1 & \sigma_{\epsilon \xi}^* \\ \sigma_{\epsilon \xi}^* & 1\end{bmatrix} \right)
\end{array}\right.
\end{align*}
where $\mathbf{Z}_{i.}$ is drawn from a multivariate normal with mean $\mathbf{0}$ and covariance matrix where the diagonals are all one. Throughout the simulation, the parameters $\pi^*,  \beta^*$, and $\gamma_0^*$  are fixed. However, we vary (i) the endogeneity parameter $\sigma_{\epsilon \xi}^*$, (ii) the direct effect parameter $\bm{\alpha}^* = (1,1,\ldots,0,0)$ where we change $s$ in $\| \bm{\alpha}^*\|_0 = s$, (iii) the pairwise correlation between instruments, i.e. $\mu$ in equation \eqref{MIP.Cond}, (iv) the absolute/overall strength of instruments, and (v) the relative strength of instruments, the latter two by changing the parameter $\bm{\gamma}^*$. 

In particular, for (i), we vary $\sigma_{\epsilon \xi}^*$ from $0$ to $0.9$. For (ii), we vary $s$ from $0$ to $9$. For (iii), we set $\mu$ at four different values, $0, 0.25, 0.5,$ and $0.75$, by setting all the off-diagonal elements of the covariance matrix of $\mathbf{Z}_{i.}$ to this value. For (iv), we vary the absolute/overall instrument strength by the concentration parameter. The concentration parameter is a popular measure for instrument strength; high values of the concentration parameter indicate the overall strength of all $L$ instruments are strong and vice versa. The concentration parameter is also the population value of the first stage F statistic for the instruments when the exposure is regressed on them; this first stage F statistic is often used to check instrument strength \citep{stock_survey_2002}. Based on Table 1 in \citet{stock_survey_2002}, a set of instruments with a concentration parameter (scaled by the number of valid instruments) of around 10 is considered weak in the absolute/overall sense and instruments with a concentration parameter (scaled by the number of valid instruments) of around 100 is considered strong in the absolute/overall sense. Finally for (v), we vary the relative instrument strength by changing the individual entries of the vector $\bm{\gamma}^*$ while keeping the concentration parameter fixed. Specifically, for a particular concentration parameter, we consider instruments to have equal relative strength if $\gamma_j^* = \gamma_k^*$ for all $j \neq k$ and variable relative strength if $\gamma_{j}^* = 2*\gamma_{k}^*$ for various values of $j \neq k$.

For each simulation setting, we repeat the simulation 1000 times. For each repetition, we compute sisVIVE's estimate of the causal effect, $\hat{\beta}_{\lambda}$, where $\lambda$ is chosen by 10-fold cross validation outlined in Section \ref{sec:lambdaCV}. We also compute estimates from TSLS and OLS. For TSLS, we run two types of TSLS. First, we run the ``naive'' TSLS as if all the instruments are valid. This is quite common in MR studies where all the instruments are assumed to be valid and the causal estimate is computed using TSLS. When some of the instruments are in fact invalid, naive TSLS should give biased estimates. Second, we run TSLS as if we knew exactly which instruments are valid, i.e. the ``oracle'' TSLS. Specifically, we use the knowledge of the support of $\bm{\alpha}^*$ and run TSLS controlling for the invalid instruments that are in the support of $\bm{\alpha}^*$ as covariates. Finally, we run OLS with $\mathbf{Z}$ and $\mathbf{D}$ as our regressors and $\mathbf{Y}$ as our outcome. We expect OLS to perform poorly when there is substantial endogeneity by $\mathbf{D}$ since OLS cannot control for endogenous variables. But, OLS should be more efficient than IV methods if there is no endogeneity \citep{richardson_note_1971}.

Figure \ref{fig:combinedPlot-endo} shows the estimation error when endogeneity is varied. The number of invalid instruments is fixed at $s = 3$ and we consider 16 different sets of instruments based on their absolute and relative strength as well as their pairwise correlations. For example, the top lefthand plot of Figure \ref{fig:combinedPlot-endo} corresponds to instruments whose overall strength is strong (i.e. scaled concentration parameter is around $100$) , their relative strength is equal (i.e. $\gamma_j^*$ are identical for all $j=1,\ldots,L$), and their pairwise correlations are $0$. In contrast, the bottom right plot of Figure \ref{fig:combinedPlot-endo} corresponds to instruments whose their overall strength is weak (i.e. scaled concentration parameter is around $10$), their relative strength is variable (i.e. $\gamma_j^* = 2*\gamma_k^*$ for $j \neq k$) and their pairwise correlations are equal to $0.75$. 

As expected, OLS dominates naive TSLS, oracle TSLS, and sisVIVE when the endogeneity is small and close to zero, with the dominance being greater for weak instruments. Once there is a sufficient amount of endogeneity, oracle TSLS, which knows exactly which instruments are valid and invalid, does best. However, sisVIVE, which is a feasible rather than infeasible oracle estimator, is close to the oracle TSLS; the gap between oracle TSLS and sisVIVE gets larger as the instruments' absolute strength gets weaker. Regardless of instrument strength, naive TSLS, which assumes all the $L$ instruments are valid, has a high error since it cannot take into account the bias introduced by invalid instruments. 

\begin{figure}[htbp!]
\centering
\includegraphics[width=7in,height=7.3in]{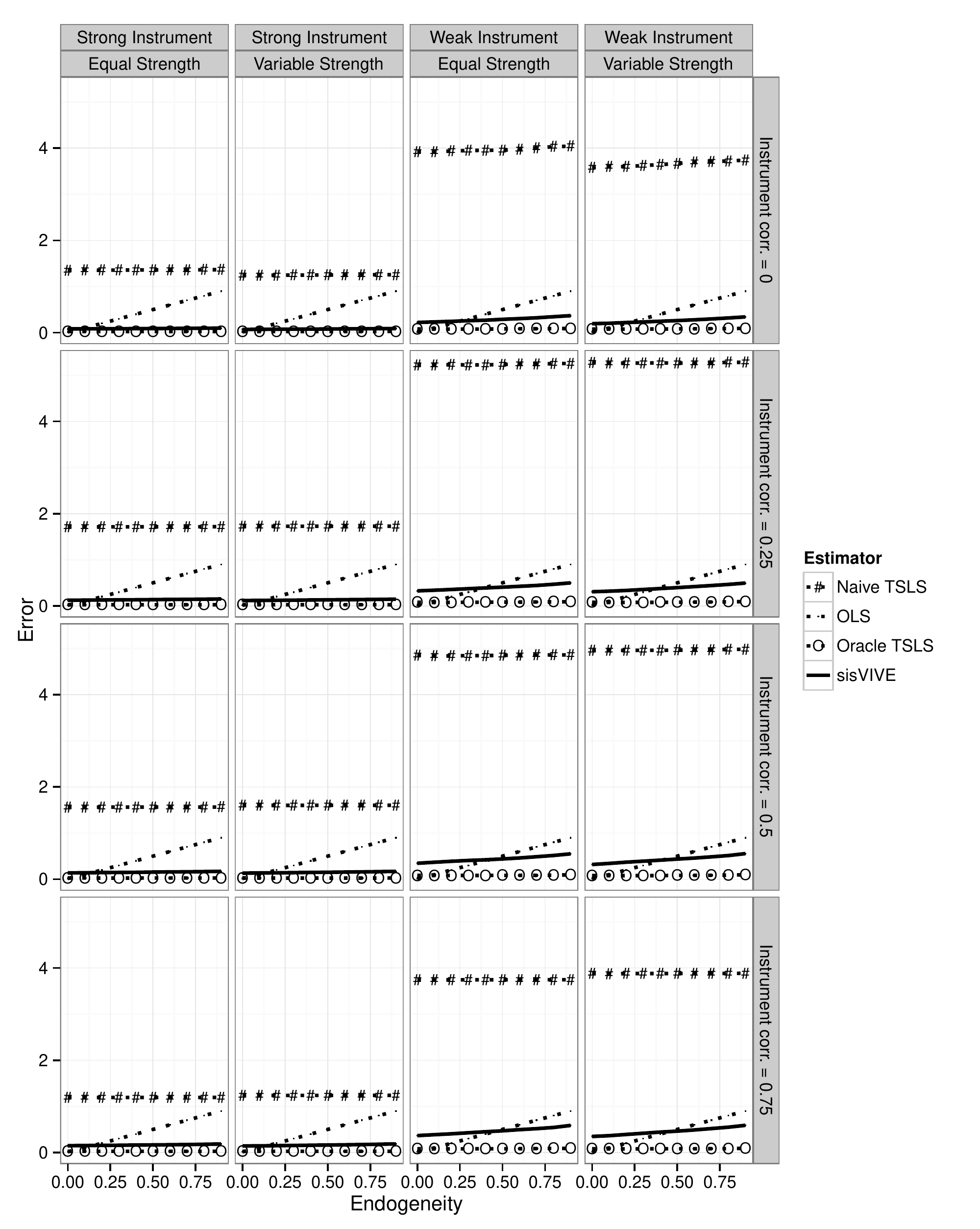}
\caption{Simulation Study of Estimation Performance Varying Endogeneity. There are ten $(L = 10)$ instruments. Each line represents median absolute estimation error ($|\beta^* - \hat{\beta}|$) after 1000 simulations. We fix the number of invalid instruments to $s = 3$. Each column in the plot corresponds to a different variation of instruments' absolute and relative strength. There are two types of absolute strengths, ``Strong'' and ``Weak'', measured by the concentration parameter. There are two types of relative strengths, ``Equal'' and ``Variable'', measured by varying $\bm{\gamma}^*$ while holding the absolute strength (i.e. concentration parameter) fixed. Each row corresponds to the maximum correlation between instruments.}
\label{fig:combinedPlot-endo}
\end{figure}

Figure \ref{fig:combinedPlot-s} shows the estimation error when the number of invalid instruments is varied. The endogeneity, $\sigma_{\epsilon \xi}^*$, is fixed at $0.8$. Like Figure \ref{fig:combinedPlot-endo}, we consider the same 16 sets of instruments. We first see that at $s =0$, i.e. when there are no invalid instruments, sisVIVE's performance is nearly identical to naive and oracle TSLS. However, sisVIVE does not use the knowledge that one knows exactly which instruments are valid while the two TSLS estimators do. Also, sisVIVE's performance degrades slightly for instruments with weak absolute strength when the correlation between instruments increases. 

When $s < L/2 = 5$, sisVIVE's performance is comparable to oracle TSLS and better than naive TSLS. However, for instruments with weak absolute strength, sisVIVE does slightly worse compared to the oracle TSLS than for instruments with strong absolute strength. Once we reach the identification boundary in Corollary \ref{coro:0}, $s < L/2 = 5$, sisVIVE's performance becomes similar to naive TSLS. This is the case regardless of the instruments' absolute and relative strength. Finally, for any $s$, oracle TSLS performs much better than all the other estimators.
\begin{figure}[htbp!]
\centering
\includegraphics[width=7in,height=7.3in]{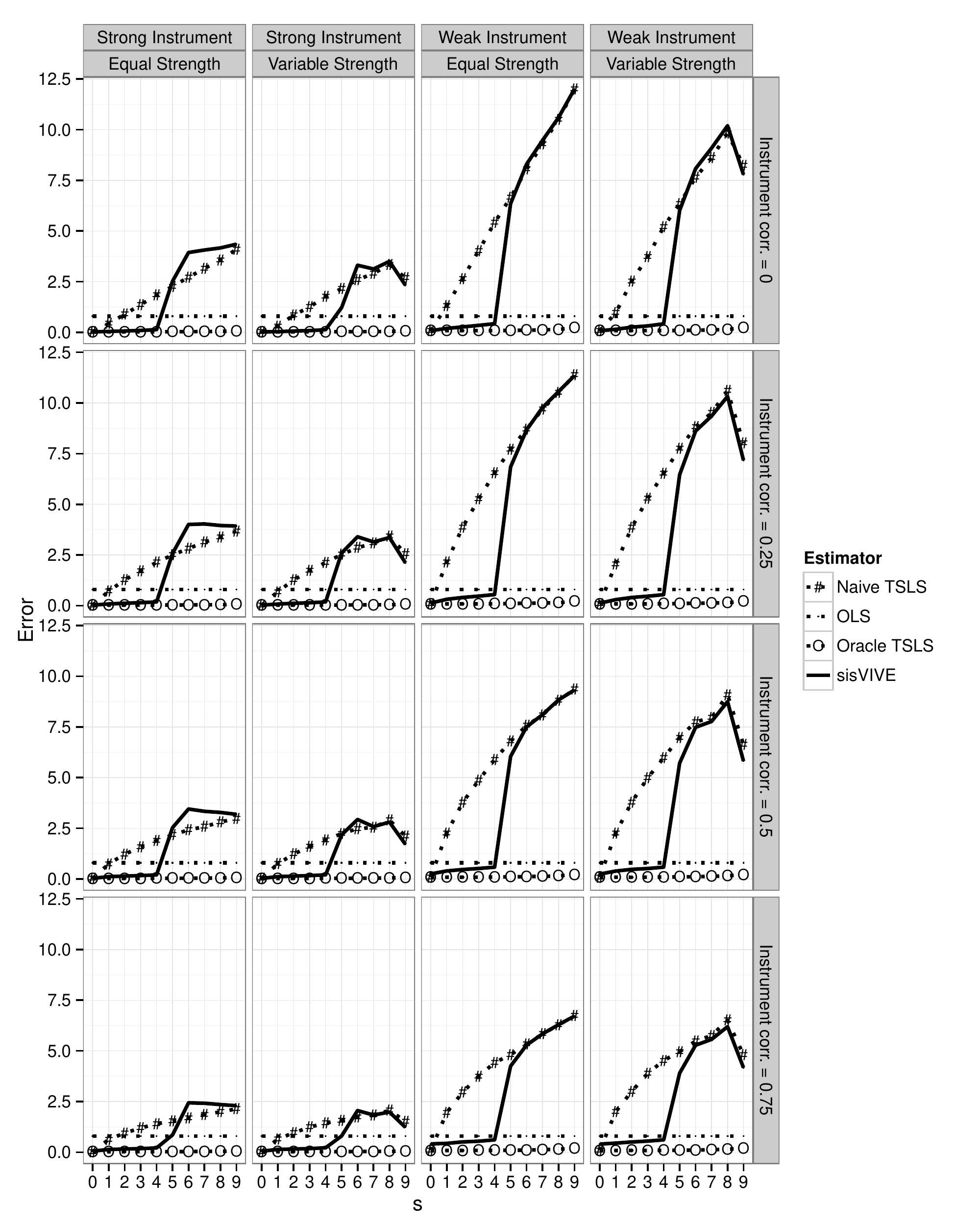}
\caption{Simulation Study of Estimation Performance Varying the Number of Invalid Instruments ($s$). There are ten $(L = 10)$ instruments. Each line represents median absolute estimation error ($|\beta^* - \hat{\beta}|$) after 1000 simulations. We fix the endogeneity $\sigma_{\epsilon \xi}^*$ to $\sigma_{\epsilon \xi}^* = 0.8$. Each column in the plot corresponds to a different variation of instruments' absolute and relative strength. There are two types of absolute strengths, ``Strong'' and ``Weak'', measured by the concentration parameter. There are two types of relative strengths, ``Equal'' and ``Variable'', measured by varying $\bm{\gamma}^*$ while holding the absolute strength fixed. Each row corresponds to maximum correlation between instruments. }
\label{fig:combinedPlot-s}
\end{figure}

Also, in all 16 sets of instruments, we compute the $\rho$ and $\mu$ found in the condition for Corollary \ref{coro:1} from the simulated data and this is detailed in the Supplementary Materials. For example, the top lefthand plot of Figure \ref{fig:combinedPlot-endo} has $\rho$ of approximately $0.31$ and $\mu = 0$. Based on this, the upper bound on $s$ in Corollary \ref{coro:1} is 1.04. However, since $s = 3$ for the simulations in Figure \ref{fig:combinedPlot-endo}, the condition \eqref{eq:constraintMIP} in Corollary \ref{coro:1} is violated and cannot be used to characterize the behavior of sisVIVE. Regardless, in our simulation study presented in this Section, sisVIVE performs just as well as the oracle TSLS. 

In the Supplementary Materials, we expand the simulation study to cover different types of instrument strength, correlation structure between instruments, and total number of potential instruments. We also explore different metrics of error, such as the proportion of correctly selected valid instruments and invalid instruments, to analyze the relationship between these proportion-based error metrics and the median bias error metric used in this Section. In addition, we also compute the conditions for Corollary \ref{coro:1}, specifically $\rho$, $\mu$, and $\lambda$ required to achieve the performance bound. The Supplementary Materials show that in every case considered, sisVIVE performs no worse than the next best alternative, naive TSLS. In fact, in most cases, sisVIVE beats naive TSLS and performs similarly to the oracle TSLS. The only case where sisVIVE's performance deviated greatly from the oracle TSLS was when the invalid instruments were weaker than the valid instruments and $s = 4$. In addition, the Supplementary Materials show that a good estimate of $\beta^*$ depends strongly on correctly selecting the invalid instruments more than correctly selecting the valid instruments and choosing $\lambda$ based on cross validation seems to favor this situation. We also find that choosing $\lambda$ based on Corollary \ref{coro:1} leads to a higher $\lambda$ than one based on cross validation. Finally, we find that sisVIVE based on $\lambda$ chosen by cross validation always performed at least as well as sisVIVE based on $\lambda$ chosen by Corollary \ref{coro:1}. In fact, in most cases, sisVIVE with a cross-validated $\lambda$ performs better than sisVIVE with a $\lambda$ chosen by Corollary \ref{coro:1}. 

Overall, sisVIVE using a cross-validated $\lambda$ does much better than naive TSLS, the most frequently used estimator in MR and IV. In many cases, sisVIVE beats the naive TSLS and it is comparable to oracle TSLS. The promising simulation results suggest that sisVIVE should be used whenever there is concern about invalid instruments.

\section{DATA ANALYSIS}
%An important question in heath economics is assessing the effect of health-related indicators, such as body mass index (BMI), mobility scores, and medical costs, on utility-based quality of life. Having accurate assessment of such effects drive design and implementation of health-related policy decisions that attempt to modify such health-related indicators to maximize individual utility. However,  As one example, if the study were to examine the relationship between BMI and a health-based utility measure, diet is a possible confounder. Not only does diet impact BMI, but it also affects one's perceived utility of health. To make matters worse, there's no straightforward way to numerically quantify diet, making it difficult to control for this confounder, even if it were observed.

%Mendelian randomization (MR) can be a powerful tool in these studies. MR uses genotypes as instruments to deduce causal relationships between exposure and outcome that are confounded by observed and unobserved variables. Traditional approaches to MR require that all candidate instruments be checked for validity. However, the method we developed allows for possibly invalid instruments amongst the candidate instruments. This feature is attractive since the biology of the genes is not always fully understood, especially when dealing with complex outcomes like utility-based quality of life and its biological association with the candidate genetic instruments. 
 
We demonstrate the potential benefit of using sisVIVE in MR by analyzing the effect of obesity, the exposure, on health-related quality of life, the outcome. An individual quality of life is the general well-being of the individual; an individual's health quality of life is the subset of quality of life related to the individual's health \citep{torrance_utility_1987}. Previous non-MR studies by \citet{trakas_health_2001} and \citet{sach_relationship_2006} have shown that there is a negative association between obesity and health-related quality of life. However, a fundamental difficulty with these studies is that the outcome, health-related quality of life, encompasses various factors about the individual, making it difficult to control for all possible confounders that may affect the causal effect \citep{cawley_medical_2012}. An MR approach offers the potential of controlling for unmeasured confounders.

For the analysis, we use the data from the Wisconsin Longitudinal Study (WLS), a well-known longitudinal study that has kept track of American high school graduates from Wisconsin since 1957. We look at graduates that were reinterviewed in 2003-2005 \citep{hauser_survey_2005} and who have been genotyped. Similar to another analysis with the WLS genetic data, we remove individuals with more than 10\% missing genotype data \citep{roetker_multigene_2012}. Our analysis of the data set contains $n = 3712$ individuals with 1913 females and 1799 males born mostly between 1938 to 1940.

To measure health-related quality of life, we use the Health Utility Index Mark 3 (HUI3) which was also used in \citet{trakas_health_2001}. HUI3 is a composite score of utility between 0 and 1, with 1 indicating highest health state and 0 indicating a health state equivalent to death; negative utility is possible and indicates that the person is alive, but in a state worse than death. To measure obesity, we looked at the body mass index (BMI) across several categories of obesity. The categories were based on US National Institute of Health clinical guidelines \citep{initiative_clinical_1998} and were also used in \citet{trakas_health_2001} and \citet{sach_relationship_2006}. Table 1 summarizes the different classes of obesity and their associations to HUI3. Different classes of obesity have different median HUI3 scores and simply classifying individuals by obese versus not obese would not capture the magnitude of the differences between the obesity classes. To account for this, we define the exposure as a censored BMI that takes the maximum of BMI $- 30$ and $0$ (i.e. $\max(BMI - 30, 0)$), to not only indicate obesity, but also to measure its severity.

\begin{table}[h!]	
\centering
\begin{tabular}{l r r r r } 
\multicolumn{5}{c}{Table 1. Relationship Between Obesity and Health Utility Index Mark 3 (HUI3)} \\ \hline \hline
& &  \multicolumn{3}{c}{Health Utility Index Mark 3} \\ \cline{3-5}
Obesity Categories & $N$ & 1st quartile & Median & 3rd quartile \\ \hline 
Not obese (BMI $< 30$) & 2581 & 0.84 & 0.92 & 0.97  \\
Obese class I ($30 \leq$ BMI $ < 35$)& 777 & 0.73 & 0.91 & 0.97 \\
Obese class II ($35 \leq$ BMI $ < 40$) & 246 & 0.66 & 0.85 & 0.97 \\
Obese class III ($40 \leq$ BMI ) & 108 & 0.51 & 0.72 & 0.91 \\ 
All categories & 3712 & 0.78 & 0.92 & 0.97 
\end{tabular}
\end{table}

For potential candidate instruments, we use the following single nucleotide polymorphisms (SNPs) in the WLS that have been previously shown to be associated with obesity: rs1421085, rs1501299, and rs2241766 (see Table 2). rs1421085 is in the FTO gene and it has been shown to be strongly associated with obesity \citep{dina_variation_2007, price_fto_2008}. rs1501299 (i.e. $+276$G$>$T) is in the ADIPOQ gene that encodes adiponectin, a protein encoding for lipid metabolism, and has been associated with obesity \citep{bouatia_acdc_2006, yang_adiponectin_2007}. Finally, rs2241766 is also in the ADIPOQ gene that has been associated with obesity \citep{ukkola_mutations_2003, yang_allele-specific_2003, beckers_association_2009}. For all the SNPs, we follow an MR study done by \citet{timpson_c-reactive_2005} and assume an additive model. Although we have no particular reason to think any of the SNPs is an invalid IV, we are uncertain due to the lack of complete knowledge about the biological functions of the SNPs, a common scenario in MR studies. Our sisVIVE estimator will provide a good estimate as long as least two of the three SNPs are valid IVs.

\begin{table}[h!]
\centering
\begin{tabular}{l r r r r r  } 
\multicolumn{5}{c}{Table 2. Summary of Instruments in the Data Analysis. MAF stands for minor allele frequency} \\ \hline \hline
Instruments & Major alleles & Heterozygote & Minor alleles & MAF (SE) \\ \hline
rs1421085 & 1281 (34.5\%; TT) & 1818 (49.0\%; CT) & 613 (16.5\%; CC) & 0.39 (0.0057)  \\
rs1501299 & 1950 (52.5\%; CC) & 1502 (40.5\%; AC) & 260 (7.0\%; AA) & 0.24 (0.0049)\\
rs2241766 & 2956 (79.6\%; TT) & 719 (19.4\%; TG) & 37 (1.0\%; GG) & 0.10 (0.0036)\\
rs6265 & 2437 (65.7\%; GG) & 1112 (30.0\%; AG)& 163 (4.4\%; AA) & 0.19 (0.0046) \\
\end{tabular}
\end{table}

A simple ordinary least squares analysis estimates that an increase in the censored BMI leads to a $-0.013$ $($SE: $0.0010)$ decrease in HUI-3 score. This is consistent with \citet{trakas_health_2001} which found that obese individuals (i.e. BMI $> 30$), on average, have lower HUI-3 scores, 0.04 to be exact, than non-obese individuals. The reduced form estimates are summarized in the Supplementary Materials.

If we use TSLS, under the operating assumption that all the instruments are valid, the estimated effect is $-0.00019$ $($SE: $ 0.022)$. Our estimator, sisVIVE, which operates only under the assumption that a proportion of instruments are invalid, estimates $-0.00019$ as the causal effect, which is identical to the estimate by TSLS. Also, sisVIVE does not select any SNPs as an invalid IV. The overidentifying restrictions test is summarized in the Supplementary Materials.
 
To further validate our method, we include another instrument, rs6265 (i.e. Val66Met). rs6265 is in the brain-derived neurotrophic factor BDNF gene and has been shown to not only be associated with BMI \citep{thorleifsson_genome_2008, shugart_two_2009}, but also neurological and cognitive function \citep{hwang_val66met_2006, rybakowski_prefrontal_2006}. Hence, there is some reason to believe that rs6265 may be pleiotropic; rs6265 may impact obesity, but also affect health-related quality of life through mechanisms other than obesity. sisVIVE should be able to pick up on this instrument being invalid in contrast to TSLS, which will always assume that all the instruments used are valid.

If we use TSLS under the operating assumption that all the four instruments are valid, the estimated effect is $0.00091$ $($SE:$ 0.022)$. sisVIVE, on the other hand, estimates the causal effect to be $-0.00011$, similar to the estimates when we used three instruments. sisVIVE also throws out the instrument, rs6265, which we suspect to be invalid. The reduced form estimates and the overidentifying restrictions test are summarized in the Supplementary Materials.

In both data analyses, sisVIVE operates under the assumption of possibly invalid instruments, which are typical in MR studies, while TSLS operates under the assumption of all  valid instruments. In the first data analysis where there was no reason to believe that the instruments were invalid, sisVIVE provides the same answer as TSLS, but without assuming that all the instruments were valid. In the second data analysis where one instrument was suspect, sisVIVE removed the suspected instrument. In both cases, sisVIVE was robust to possibly invalid instruments compared to TSLS.
 
\section{DISCUSSION}
This paper demonstrates that proper estimation of causal effects using the IV method is possible without knowledge of all the instruments' validity. Our results show that simply knowing a proportion of the instrument is valid, without knowing which are valid, is sufficient and we construct the sisVIVE estimator that dominates the naive TSLS in almost every aspect while performing similarly to the oracle TSLS. Both the simulation result and data analysis show that sisVIVE is a robust alternative to TSLS in the presence of possibly invalid instruments.

Future work could involve generalizing the model considered. In particular, the current paper discusses a model in which treatment effects are constant. \citet{angrist_identification_1996} discusses the setting in which the treatment effects are not constant and individuals may select into treatment based on expected gains from treatment. Then, $q_{m}$ and $q_{m'}$ in Theorem \ref{prop:1} might not be equal to each other for different sets of valid instruments and Theorem \ref{prop:1} does not apply. It would be useful to understand what sisVIVE is estimating under this setting of treatment effect heterogeneity. Other useful directions for future work are relaxing the conditions on Corollary \ref{coro:1} to encompass more invalid instruments $s$ and deriving tests for identification. Also, we have focused on the applications of our method to Mendelian randomization. In economic applications, it is also common to have multiple candidate instruments and be concerned that some proportion of the instruments are invalid \citep{murray_avoiding_2006}. Our current work demonstrates that instrumental variable estimation is definitely possible even in the presence of possibly invalid instruments.

\appendix
\section{Additional Discussion About Theorem 1}
\subsection{Numerical Example}
In Section 3.1 of the main manuscript, we discussed the identification result and illustrated it with a numerical example where $L = 4$, $\bm{\gamma}^* = (1,2,3,4)$, $\bm{\Gamma}^* = (1,2,6,8)$, and $s < U$ where $U = 3$. We showed that there are two sets $C_1 = \{1,2\}$ and $C_2 = \{3,4\}$ with $q_1 = 1$ and $q_2 = 2$. Since $q_1 \neq q_2$, by Theorem 1, identification is not possible with this numerical example. 

One of the reviewers, however, mentioned an interesting numerical example where the setup is identical to our numerical example above, except $\bm{\Gamma}^*$ is perturbed by $\epsilon > 0$ such that $\tilde{\bm{\Gamma}}^* = (1,2,6, 8+ \epsilon)$. With $\tilde{\bm{\Gamma}}^*$, there is only one set $C_1 = \{1,2\}$ where $q_1 = 1$ and we have identification for any $\epsilon$. However, we can shrink $\epsilon$ to be arbitrary small such that $\bm{\Gamma}^*$ and $\tilde{\bm{\Gamma}}^* = (1,2,6, 8 + \epsilon)$, are arbitrarily close to each other. As the reviewer stated ``As a result, in any finite sample, it will be impossible to distinguish between the two cases, and hence no estimation or inference results that rely on Theorem 1 can be uniformly valid.'' 

However, consider the identical setup as before, except $\bm{\Gamma}^* = (1,2,7,9)$.  Then, there is only one subset $C_1 = \{1,2\}$ where $q_1 = 1$ and identification is achieved. Furthermore, any small perturbation of $\bm{\Gamma}^*$ by $\delta > 0$ and $\epsilon > 0$, i.e. $\tilde{\bm{\Gamma}}^* = (1,2,7+\delta,9+\epsilon)$, will still produce only subset $C_1 = \{1,2\}$ and identification is maintained. 

The two numerical examples with $\bm{\Gamma}^* = (1,2,6,8)$ and $\bm{\Gamma}^* = (1,2,7,9)$ illustrate what we call the \emph{identification boundary}. The vector $\bm{\Gamma}^* = (1,2,6,8)$ lies just at the identification boundary where any small perturbation can render the model unidentified or identified. In contrast, for $\bm{\Gamma}^* = (1,2,7,9)$, the vector $\bm{\Gamma}^*$ lies far from the identification boundary and any small perturbation can still make the model identifiable. Exploration of the identification boundary for different values of $\bm{\Gamma}^*$ and $\bm{\gamma}^*$ is a topic for future research.

\subsection{Normality Assumption and Identification}
We consider two additional modeling assumptions which are not needed for identification, but are part of the classical linear simultaneous/structural equations model \citep{koopmans_measuring_1950} and discuss the identification result in Section 3.1 of the main manuscript. First, we assume that the relationship between $D_i$ and $\mathbf{Z}_{i.}$ is assumed to be linear
\begin{equation} \label{eq:modelFirstStage}
D_i = \mathbf{Z}_{i.}^T \bm{\gamma}^* + \xi_i, \quad{} E(\xi_i | \mathbf{Z}_{i.}) = 0
\end{equation}
where $\bm{\gamma}^*$ relates the instruments to the exposure and the error terms are bivariate Normal
\begin{equation} \label{eq:bivariateNormal}
(\epsilon_i, \xi_i) \iid N(\mathbf{0},\mathbf{\Sigma})
\end{equation}
Under these assumptions in \eqref{eq:modelFirstStage} and \eqref{eq:bivariateNormal}, the distributions of $Y_{i}$ and $D_{i}$ conditional on $\mathbf{Z}_{i.}$ are fully characterized by finite-dimensional parameters $\bm{\alpha}^*, \beta^*, \bm{\gamma}^*$, and $\mathbf{\Sigma}$ known as ``structural'' parameters in econometrics \citep{wooldridge_econometrics_2010}. Let $\epsilon'_i = \beta^* \xi_i + \epsilon_i$. Then, we have the ``reduced forms'' \citep{wooldridge_econometrics_2010}
\begin{align*}
Y_i &= \mathbf{Z}_{i.}^T\bm{\Gamma}^* + \epsilon'_i \\
D_i &= \mathbf{Z}_{i.}^T \bm{\gamma}^* + \xi_i
\end{align*}
where $\bm{\Gamma}^* = \bm{\alpha}^* + \beta^* \bm{\gamma}^*$ and the covariance matrix of $(\epsilon'_i,\xi_i)$ is $\mathbf{\Sigma'} = \mathbf{M} \mathbf{\Sigma} \mathbf{M}^T$ with
\[
M = \begin{pmatrix}
1 & \beta^* \\
0 & 1
\end{pmatrix}
\] 
We see that the distribution of $Y_i$ and $D_i$ are also fully characterized by the reduced form parameters $\bm{\Gamma}^*, \bm{\gamma}^*$ and $\mathbf{\Sigma'}$. By \citet{rothenberg_identification_1971}, the reduced form parameters, $ \bm{\Gamma}^*, \bm{\gamma}^*$, and $\mathbf{\Sigma'}$, are globally identified. Also, by \citet{rothenberg_identification_1971}, the structural parameters, $\bm{\alpha}^*$, $\beta^*$, $\bm{\gamma}^*$, and $\mathbf{\Sigma}$, are identified if and only if the mapping between the reduced form parameters, $\bm{\Gamma}^*, \bm{\gamma}^*, \mathbf{\Sigma'}$, and the structural parameters, $\bm{\alpha}^*$, $\beta^*$, $\bm{\gamma}^*,\mathbf{\Sigma}$, represented by equations $\mathbf{\Sigma'} = \mathbf{M} \mathbf{\Sigma} \mathbf{M}^T$, $\bm{\gamma}^* = \bm{\gamma}^*$, and $\bm{\Gamma}^* = \bm{\alpha}^* + \beta^* \bm{\gamma}^*$, is bijective. We see that $\mathbf{M}$ is an invertible matrix for any $\beta^*$ and hence there is a bijective map between $\mathbf{\Sigma}$ and $\mathbf{\Sigma}'$. For $\bm{\gamma}^*$, it maps onto itself between the structural and reduced form parameters. Consequently, whether there is a bijection between the structural parameters and reduced form parameters is determined only by whether there is a unique solution $\bm{\alpha}^*$ and $\beta^*$ to the equation \eqref{eq:momentCond2} given $\bm{\gamma}^*$ and $\bm{\Gamma}^*$. Theorem 1 in the main manuscript states that a unique solution $\bm{\alpha}^*$ and $\beta^*$ of \eqref{eq:momentCond2} exists if and only if the consistency criterion holds, that $q_{m} = q_{m'}$ for all $m,m' \in \{1,\ldots,M\}$. Hence, with the modeling assumptions \eqref{eq:modelFirstStage} and \eqref{eq:bivariateNormal}, we have identification of the structural parameters if and only if the consistency criterion holds.

\section{Simulation}
\subsection{Values of $\rho$ and $\mu$} \label{sec:rhomuCorollary2}
In Section 4 of the main manuscript, we conduct a simulation study to study the performance of sisVIVE compared to other competitors such as two stage least squares. In addition, in Section 3.4 of the main manuscript, Corollary 2 characterizes the performance of sisVIVE theoretically if certain conditions based on constants $\rho$ and $\mu$ are satisfied. In this section, we check whether these theoretical conditions are met for the simulation setup we considered in the main manuscript. 

We first computed $\rho$ from each simulated data set and take the median value of it after 1000 simulations. To compute $\mu$, we use the true values of the correlation of $\mathbf{Z}_{i.}$, specifically $\mu = 0, 0.25, 0.5$, and $0.75$. Table 1 shows the value for $\mu$ and $\rho$ for the simulation setup in the main manuscript.
\begin{table}[h!]
\centering
\begin{tabular}{p{2cm} p{3.4cm} p{3.4cm} p{3.4cm} p{3.4cm} }
\multicolumn{5}{c}{Table 1. Values of $\rho$ defined in Corollary 2 for the Simulation Study} \\ \hline \hline
Instrument Corr. $(\mu)$ & Strong Instrument, Equal Strength & Strong Instrument, Variable Strength & Weak Instrument, Equal Strength & Weak Instrument, Variable Strength \\ \hline
0 & 0.31 & 0.39 & 0.20 & 0.22 \\
0.25 & 0.54 & 0.58 & 0.36 & 0.37 \\
0.5 & 0.72 & 0.73 & 0.53 & 0.53 \\
0.75 & 0.87 & 0.87 & 0.73 & 0.73
\end{tabular}
\label{tab:table1}
\end{table}
Second, based on the values of $\mu$ and $\rho$ in Table 1, we check the condition required in Corollary 2, specifically the upper bound on $s$, $\min(1/(12 \mu), 1/(10 \rho^2) )$, in equation (14) of the main manuscript. These upper bounds are evaluated in Table 2.
\begin{table}[h!]
\centering
\begin{tabular}{p{2cm} p{3.4cm} p{3.4cm} p{3.4cm} p{3.4cm} }
\multicolumn{5}{c}{Table 2. Condition on $s$ in Corollary 2 for the Simulation Study} \\ \hline \hline
Instrument Corr. $(\mu)$ & Strong Instrument, Equal Strength & Strong Instrument, Variable Strength & Weak Instrument, Equal Strength & Weak Instrument, Variable Strength \\ \hline
0 & 1.04 & 0.66 & 2.50 & 2.07 \\
0.25 & 0.33 & 0.33 & 0.33 & 0.33 \\
0.5 & 0.17 & 0.17 & 0.17 & 0.17 \\
0.75 & 0.11 & 0.11 & 0.11 & 0.11
\end{tabular}
\label{tab:table2}
\end{table}
Table 2 shows that in most settings, the condition for Corollary 2 is only satisfied when $s = 0$, i.e. when there are no invalid instruments. For example, when instrument are correlated and $\mu > 0$, Corollary 2 cannot be used to characterize the performance of sisVIVE if invalid instruments are present. Table 2 also illustrates the point we illustrated in the main manuscript, that the condition for Corollary 2, even though it's interpretable, are strict. In the main manuscript, we provide a generalization of Corollary 2 in Theorem 2 at the expense of interpretability.

\subsection{Varying Correlation Structure} \label{sec:varysimstructure}
In this section, we extend the simulation study in Section 4 of the main manuscript by considering other correlation structures between the instruments beyond those considered in the main manuscript. First, Figures \ref{fig:withinZcorrEndo} and \ref{fig:withinZcorrS} of the Supplementary Materials
represent the setting where the pairwise correlation between valid instruments is set to $\mu$ and the pairwise correlation between invalid instruments is also set to $\mu$. However, there is no correlation between any pair consisting of one valid and one invalid instrument. The new setup differs from the main manuscript where all the pairwise correlation between any two instruments is set to $\mu$. Second, Figures \ref{fig:interZcorrEndo} and \ref{fig:interZcorrS} represent the setting where the pairwise correlation between a valid instrument and an invalid instrument is set to $\mu$. However, there is no pairwise correlation between any pair of valid instruments or any pair of invalid instruments. Under the two new correlation structures, we rerun the simulation study in the main manuscript except we reduce the simulation number from $1000$ to $500$ and we only vary $s$ with values $s=1, 3, 4, 5, 7,$ and $9$ for computational reasons. Also, note that as a result of repeating the same simulation, the conditions for Corollary 2 in the main manuscript are similar to those in Tables 1 and 2 of Section \ref{sec:rhomuCorollary2} in the Supplementary Materials.

In both Figures \ref{fig:withinZcorrEndo} and \ref{fig:interZcorrEndo} of the Supplementary Materials where we vary endogeneity, but the number of invalid instruments is fixed at $s = 3$, the behavior of all the estimators are similar to each other and to those in the main manuscript. OLS dominates naive TSLS, oracle TSLS, and sisVIVE when the endogeneity is small and close to zero, with the dominance being greater for weaker instruments. Once there is a sufficient amount of endogeneity, oracle TSLS, which knows exactly which instruments are valid and invalid, does best. sisVIVE also resembles the oracle in terms of performance. Naive TSLS, which assumes all the $L$ instruments are valid, does worst since it assumes that all the $L$ instruments are valid. 

Similarly, in Figures \ref{fig:withinZcorrS} and \ref{fig:interZcorrS} of the Supplementary Materials where we vary the number of invalid instruments, $s$, but fix the endogeneity to $0.8$, the estimators behave similarly across the two Figures and to those in the main manuscript. We first see that at $s =0$, i.e. when there are no invalid instruments, sisVIVE's performance is nearly identical to naive and oracle TSLS, although it degrades slightly for instruments with weak absolute strength. Also, when $s < L/2 = 5$, sisVIVE's performance is comparable to oracle TSLS and better than naive TSLS. Once we reach the identification boundary, $s < L/2 = 5$, sisVIVE's performance becomes similar to naive TSLS. This is the case regardless of the instruments' absolute and relative strength. 

\subsection{Performance of Estimate of $\hat{\bm{\alpha}}_\lambda$} \label{sec:propBasedError}
In this section, we extend the simulation study in Section 4 of the main manuscript by examining the estimation performance of $\bm{\alpha}^*$ for sisVIVE. As we noted in the main manuscript, in Mendelian randomization, the target of estimation is $\beta^*$, the causal effect of the exposure on the outcome, and our procedure, sisVIVE, was designed to estimate $\beta^*$. However, in the process of estimating $\beta^*$, sisVIVE does produce an estimate for $\bm{\alpha}^*$. This section explores the relationship between this intermediate estimate for $\bm{\alpha}^*$, $\hat{\bm{\alpha}}_{\lambda}$, and our desired estimate for $\beta^*$, $\hat{\beta}_{\lambda}$.

To evaluate the estimate $\hat{\bm{\alpha}}_\lambda$, we consider two metrics for error, the proportion of correctly selected valid instruments and the proportion of correctly selected invalid instruments. To illustrate these proportion-based error metrics, consider the following numerical example. Suppose there are $L = 10$ instruments of which the first three instruments are invalid, $\alpha_j^* \neq 0$ for $j =1,2,3$ and the last seven instruments are valid, $\alpha_j^* = 0$ for $j=4, 5,\ldots, 10$. If sisVIVE estimates the first two instruments to be invalid, $\hat{\alpha}_{j} \neq 0$ for $j=1,2$ and the last eight to be valid, $\hat{\alpha}_j = 0$ for $j =3,4,\ldots,10$, the proportion of correctly selected valid instruments is $7/7 = 1$ and sisVIVE makes no error in estimating the valid instruments. However, the proportion of correctly selected invalid instruments is $2/3$ and sisVIVE makes an error in estimating the invalid instruments.

We rerun the simulation setup in Section 4 of the main manuscript and in Section \ref{sec:varysimstructure} in the Supplementary Materials. However, instead of measuring the median absolute deviation, $|\hat{\beta}_{\lambda} - \beta^*|$, we instead measure the two proportion-based error metrics. Similar to Section \ref{sec:varysimstructure} in the Supplementary Materials, we reduce the simulation from $1000$ to $500$ and only consider $s = 1,3,4,5,7,$ and $9$ for computational reasons. The results are in Figures \ref{fig:equalZcorrEndoPercent} to \ref{fig:interZcorrSPercent}. 

When we vary endogeneity but fix the number of invalid instruments to be $s = 3$ (Figures \ref{fig:equalZcorrEndoPercent}, \ref{fig:withinZcorrEndoPercent}, and \ref{fig:interZcorrEndoPercent}),  the proportion of correctly selected invalid instruments is $1$ and sisVIVE never makes a mistake in selecting the invalid instruments. However, sisVIVE does make mistakes in selecting the valid instruments as the proportion of correctly selected valid instruments is mostly below $1$. Also, depending on the correlation structure between instruments, we get different behaviors for the proportion of correctly selected valid instruments. For example, when every pair of instruments has non-zero pairwise correlation (Figure \ref{fig:equalZcorrEndoPercent}), the proportion of correctly selected valid instruments remains roughly the same for different values of endogeneity. When there is only pairwise correlation within valid and invalid instruments (Figure \ref{fig:withinZcorrEndoPercent}), the proportion of correctly selected valid instruments decreases as endogeneity increases, most notably among weak instruments. Finally, when there is only pairwise correlation between valid and invalid instruments (Figure \ref{fig:interZcorrEndoPercent}), the proportion of correctly selected valid instruments increases as endogeneity increases. Despite these differences in the proportion of correctly selected valid instruments between different correlation structures, as the simulations in Section 4 of the main manuscript and Section \ref{sec:varysimstructure} of the Supplementary Materials showed, sisVIVE's median absolute deviation from the truth, $|\hat{\beta}_{\lambda} - \beta^*|$, remains relatively small and constant for all values of the endogeneity. This constant behavior is also present in the proportion of correctly selected invalid instruments, which remains at $1$ for all correlation structures. This suggests that there is a strong relationship between correctly selecting the invalid instruments and sisVIVE's median absolute deviation from $\beta^*$ while there is at most a weak relationship between correctly selecting valid instruments and sisVIVE's median absolute deviation from $\beta^*$. In fact, it appears that correctly selecting invalid instruments is more important than valid instruments if a small median absolute deviation is desired.

When we vary the number of invalid instruments $s$, but fix the endogeneity (Figures \ref{fig:equalZcorrSPercent}, \ref{fig:withinZcorrSPercent}, and \ref{fig:interZcorrSPercent}), the proportion of correctly selected invalid instrument decreases significantly at the $s = 5$ boundary, regardless of the correlation structure between instruments. For example, for strong instruments in the three Figures, when $s < 5$, the proportion of correctly selected invalid instruments remain at $1$. However, when $s \geq 5$, the proportion of correctly selected invalid instruments moves sharply away from $1$. For weak instruments in the three Figures, when $s < 5$, the proportion of correctly selected invalid instruments remains close to $1$, although there is a slightly decrease in the proportion when $s$ moves from $s = 3$ to $s = 4$ and when $\mu$ is away from zero. However, similar to the strong instruments, when $s \geq 5$, the proportion of correctly selected invalid instruments moves away from $1$. In contrast, the proportion of correctly selected valid instruments decreases steadily as $s$ increases, regardless of the type of correlation structure between instruments. For strong instruments in the three Figures, the decrease in the proportion of correctly selected valid instruments begins immediately after $s = 1$. For weak instruments in the three Figures, there is considerable fluctuation of the proportion of correctly selected valid instruments. For Figures \ref{fig:equalZcorrSPercent} and Figures \ref{fig:withinZcorrSPercent}, the proportion of correctly selected valid instruments generally decreases as $s$ increase, with the notable exception in the first row, third column of both Figures. For Figure \ref{fig:interZcorrSPercent}, the proportion of correctly selected valid instruments decreases when $s < 5$, but increases again after $s \geq 5$. 

The behaviors of the proportions of correctly selected invalid and valid instruments from Figures \ref{fig:equalZcorrSPercent}, \ref{fig:withinZcorrSPercent}, and \ref{fig:interZcorrSPercent} reaffirms our previous observation that there is a strong association between the proportion of correctly selected invalid instruments and the median absolute deviation of $\hat{\beta}_\lambda$, $|\hat{\beta}_{\lambda} - \beta^*|$. In particular, from Figure 3 of the main manuscript and Figures \ref{fig:withinZcorrS} and \ref{fig:interZcorrS} of the Supplementary Materials, when $s < 5$, sisVIVE's median absolute deviation is just as small as the oracle two stage least squares. However, when $s \geq 5$, sisVIVE's median absolute deviation is just as large as the naive two stage least squares. The proportion of correctly selected invalid instruments in Figures \ref{fig:equalZcorrSPercent}, \ref{fig:withinZcorrSPercent}, and \ref{fig:interZcorrSPercent} closely corresponds to this sharp change in behavior between $s < 5$ and $s \geq 5$. In contrast, the proportion of correctly selected valid instruments does not have this sharp behavior at $s = 5$ across all the figures. 

Overall, by measuring the estimation performance of $\hat{\bm{\alpha}}_{\lambda}$ using the two proportion-based error metrics, we notice a strong relationship between the proportion of correctly selected invalid instruments and the median absolute deviation of $\hat{\beta}_{\lambda}$. For any type of correlation structure between instruments and different variations on endogeneity and $s$, sisVIVE deviates far from the truth if we incorrectly select the invalid instruments. Hence, it is much more important to correctly select invalid instruments at the expense of incorrectly selecting valid instruments for better estimation of $\beta^*$. This relationship makes sense since using invalid instruments creates bias whereas using at least one valid instrument and not using other valid instruments does not create bias, but just reduces efficiency. The relationship also suggests that when we choose the tuning parameter $\lambda$, which controls the number of non-zero $\hat{\bm{\alpha}}_{\lambda}$ and consequently, controls the proportion of correctly selected valid and invalid instruments, we should choose $\lambda$ that correctly selects the invalid instruments, even if some valid instruments are selected as invalid. In particular, $\lambda$ should generally be small so that there is less $\ell_1$ penalty on $\|\bm{\alpha}\|_1$, but not too small so that the penalty has no effect. As a result, few elements of $\hat{\bm{\alpha}}_{\lambda}$ will be zero and more instruments will be selected as invalid.  We discuss the choice of $\lambda$ in more detail in Section \ref{sec:choiceofLambda}.

\subsection{Varying Instrument Strength} \label{sec:strengthvary}
In this section, we extend the simulation study in Section 4 of the main manuscript by considering other  types of instrument strength beyond those considered in the main manuscript. Specifically, we look at two cases where the invalid instruments are ``stronger'' than the valid instruments and the valid instruments are ``stronger'' than the invalid instruments. To simulate these two new cases, we first fix the concentration parameter, a global/overall measure of instrument strength, similar to the simulation setup in the main manuscript. Second, given a concentration parameter, for the case when the invalid instruments are stronger than the valid instruments, we find $\bm{\gamma}^*$ where $\gamma_j^* = 2 * \gamma_k^*$ for $j \in$supp$(\bm{\alpha}^*)$ (i.e. set of invalid instruments) and $k \in$supp$(\bm{\alpha}^*)^C$ (i.e. set of invalid instruments). In other words, the $\gamma_j^*$s associated with invalid instruments have twice the magnitude of the $\gamma_j^*$s associated with the valid instruments. For the case when the valid instruments are stronger than the invalid instruments, we flip the roles of $j$ and $k$ where $j$ now belongs to supp$(\bm{\alpha}^*)^C$ and $k$ belongs to supp$(\bm{\alpha}^*)^C$. Finally, we rerun the simulation setup in Section 4 of the main manuscript and Sections \ref{sec:varysimstructure} and \ref{sec:propBasedError} of the Supplementary Materials, except we replace the ``Equal'' and ``Variable'' strengths with the two new types of instrument strength introduced in this Section, denoted as ``Stronger Invalid'' (i.e. the case when the invalid instruments are stronger than the valid instruments) and ''Stronger Valid'' (i.e. the case when the valid instruments are stronger than the invalid instruments). We also reduce the number of simulations $1000$ to $500$ for computational reasons.

In addition, for each of the simulation setups, we repeat the exercise we did in Section \ref{sec:rhomuCorollary2} of the Supplementary Materials where we compute $\rho$ and $\mu$ that appear in Corollary 2 of the main manuscript. Table 3 and 4 show the results when the instruments have the identical pairwise correlation; for other correlation structures, the condition on $s$ is similar and hence, they are not presented (see Section \ref{sec:varysimstructure} of the Supplementary Materials for discussion on this). The column and row labels in the two tables are identical as those found in Section \ref{sec:varysimstructure} of the Supplementary Materials, except the new headings ``Stronger Invalid'' and ``Stronger Valid.''
\begin{table}[h!]
\centering
\begin{tabular}{p{2cm} p{3.4cm} p{3.4cm} p{3.4cm} p{3.4cm} }
\multicolumn{5}{c}{Table 3. Values of $\rho$ defined in Corollary 2 for the Simulation Study} \\ \hline \hline
Instrument Corr. $(\mu)$ & Strong Instrument, Stronger Invalid & Strong Instrument, Stronger Valid & Weak Instrument, Stronger Invalid & Weak Instrument, Stronger Valid \\ \hline
0 & 0.41 & 0.33 & 0.28 & 0.18 \\
0.25 & 0.60 & 0.54 & 0.47 & 0.33 \\
0.5 & 0.75 & 0.71 & 0.64 & 0.49 \\
0.75 & 0.88 & 0.86 & 0.81 & 0.70
\end{tabular}
\label{tab:varystrengthrho}
\end{table}
\begin{table}[h!]
\centering
\begin{tabular}{p{2cm} p{3.4cm} p{3.4cm} p{3.4cm} p{3.4cm} }
\multicolumn{5}{c}{Table 4. Condition on $s$ in Corollary 2 for the Simulation Study} 
\\ \hline \hline
Instrument Corr. $(\mu)$ & Strong Instrument, Stronger Invalid & Strong Instrument, Stronger Valid & Weak Instrument, Stronger Invalid & Weak Instrument, Stronger Valid \\ \hline
0 & 0.60 & 0.90 & 1.27 & 3.02 \\
0.25 & 0.28 & 0.33 & 0.33 & 0.33 \\
0.5 & 0.17 & 0.17 & 0.17 & 0.17 \\
0.75 & 0.11 & 0.11 & 0.11 & 0.11
\end{tabular}
\label{tab:varystrengthS}
\end{table}

Figures \ref{fig:equalZcorrEndo-awkward} to \ref{fig:equalZcorrSPercent-awkward} represent the cases where the instruments have identical pairwise correlation $\mu$. When we vary endogeneity, but fix $s = 3$ (Figure \ref{fig:equalZcorrEndo-awkward}), sisVIVE performs as well as the oracle for strong instruments. For weak instruments, sisVIVE does better when the valid instruments are stronger than the invalid instruments (i.e. ``Stronger Valid'') than when the invalid instruments are stronger than the valid instruments (i.e. ``Stronger Invalid''). In both the strong and weak cases, sisVIVE does much better than the next best alternative, naive two stage least squares. 

When we vary $s$, but fix endogeneity to $0.8$ (Figure \ref{fig:equalZcorrS-awkward}), sisVIVE deviates from the oracle at $s = 4$ for the case when the invalid instruments are stronger than the valid instruments (i.e. ``Stronger Invalid'') and at $s = 7$ for the case when the valid instruments are stronger than the invalid instruments (i.e. ``Stronger Valid''). When sisVIVE deviates from oracle TSLS, sisVIVE's performance is no worse than naive two stage least squares. 

When we look at the proportion-based error metrics for estimating $\bm{\alpha}_\lambda^*$ (Figures \ref{fig:equalZcorrEndoPercent-awkward} and \ref{fig:equalZcorrSPercent-awkward}), the behavior of the two curves are similar to what we observed in Section \ref{sec:propBasedError}. That is, whenever sisVIVE performs badly, there is a large decrease in the proportion of correctly selected invalid instruments. Also, there is no relationship between sisVIVE's median absolute bias of $\hat{\beta}_{\lambda}$ and the proportion of correctly selected valid instruments. When we vary endogeneity (Figure \ref{fig:equalZcorrEndoPercent-awkward}), the proportion of correctly selected invalid instruments remain at 1 except when the overall strength of the instruments is weak and the invalid instruments are stronger than the valid instruments (i.e. ``Stronger Invalid''). However, in all cases, a smaller median absolute deviation in Figure \ref{fig:equalZcorrEndo-awkward} corresponds with having a high proportion of correctly selected invalid instruments in Figure \ref{fig:equalZcorrEndoPercent-awkward}. In contrast, the proportion of correctly selected valid instruments remains below $1$ if the invalid instruments are stronger than the valid instruments (i.e. ``Stronger Invalid'') and close to $1$ if the valid instruments are stronger than the invalid instruments (i.e. ``Stronger Valid'').

Similarly, when we vary $s$ (Figure \ref{fig:equalZcorrSPercent-awkward}) and are under the case where the invalid instruments are stronger than the valid instruments (i.e. ``Stronger Invalid''), the proportion of correctly selected invalid instruments move away from $1$ at $s = 4$ when the overall strength of the instruments is strong and at $s = 3$ when the overall strength of the instruments is weak. When the valid instruments are stronger than the invalid instruments (i.e. ``Stronger Valid''), the proportion of correctly selected invalid instruments move away from $1$ at $s = 7$ for strong instruments and $s = 6$ for weak instruments. Again, similar to what we observed in Section \ref{sec:propBasedError} of the Supplementary Materials, these points of $s$ correspond to sisVIVE's deviation from the oracle in Figure \ref{fig:equalZcorrS-awkward}. In contrast, the proportion of correctly selected valid instruments vary widely in Figure \ref{fig:equalZcorrSPercent-awkward} and there does not seem to be any relationship between it and sisVIVE's deviation from the oracle.

For other correlation structures, specifically when (i) there is only correlation within valid and invalid instruments, and (ii) there is only correlation between valid and invalid instruments, we observe the same phenomena as the case where all the instruments are correlated. This is in alignment with Sections \ref{sec:varysimstructure} and \ref{sec:propBasedError}. The result from the two correlation structures under the different types of instrument strengths considered in this Section are  in Figures \ref{fig:withinZcorrEndo-awkward} to \ref{fig:interZcorrSPercent-awkward}.

The simulation study in this Section showed that in vast majority of cases, sisVIVE estimates the causal effect of interest better than the next best alternative, naive two stage least squares and in many cases, sisSIVE's performance is similar to the oracle. However, when the invalid instruments are stronger than the valid instruments (i.e. ``Stronger Invalid''), sisVIVE's performance does not do as well relative to the oracle, even though by the identification result in Corollary 1 of the main manuscript, at $s = 4$, identification is guaranteed. The degradation in performance of sisVIVE may be due to a number of reasons. It may follow from the fact that the condition in Corollary 2 are not met since Table 4 shows that in the ``Stronger Invalid'' case, $s$ has to be less than $1$ or $2$. It may be that we chose a bad tuning parameter $\lambda$; based on the results on the proportion of correctly selected invalid instruments, we may need a smaller $\lambda$ than what we used was chosen by cross validation. A closer analysis of this particular case more closely is a topic for future research. Regardless, even when sisVIVE's performance degrades, it does no worse than the next best alternative, naive two stage least squares.

In addition, the simulation study reaffirmed the points mentioned in Sections \ref{sec:varysimstructure} and \ref{sec:propBasedError} of the Supplementary Materials that (i) sisVIVE seems to do well under different correlation structures, and (ii) $\hat{\beta}_{\lambda}$'s deviation from $\beta^*$ depends heavily on the proportion of correctly selected invalid instruments more so than the proportion of correctly selected valid instruments.

\subsection{Number of potential instruments} \label{sec:Lchange}
In this section, we extend the simulation study in Section 4 of the main manuscript by increasing the potential number of instruments from $L = 10$ to $L = 100$. We note that in Mendelian randomization settings, it is rare to have 100 potential genetic instruments since all $100$ of the genetic instruments must affect the exposure (see the Introduction and Section 3.1 of the main manuscript for details). Usually, the number of potential instruments is far less than $100$ (see citations in the Introduction of our main manuscript for examples). However, for completeness, we demonstrate sisVIVE's performance when $L = 100$ potential instruments are present.

We rerun the simulation setup in Section 4 of the main manuscript and Section \ref{sec:propBasedError} in the Supplementary Materials except $L = 100$ and when we vary endogeneity, we fix the number of invalid instruments to be $30$ (instead of 3); note that based on the simulation results in Section \ref{sec:varysimstructure} where other correlation structures did not impact the performance of sisVIVE, we only consider the correlation structure in the main manuscript, specifically where all the instruments are correlated to each other with pairwise correlation $\mu$. Also, for computational reasons, we reduce the simulation number from $1000$ to $500$. Finally, we repeat the exercise in Section \ref{sec:rhomuCorollary2} by computing $\rho$ and $\mu$ defined in Corollary 2. Table 5 and 6 show the results. 
\begin{table}[h!]
\centering
\begin{tabular}{p{2cm} p{3.4cm} p{3.4cm} p{3.4cm} p{3.4cm} }
\multicolumn{5}{c}{Table 5. Values of $\rho$ defined in Corollary 2 for the Simulation Study} \\ \hline \hline
Instrument Corr. $(\mu)$ & Strong Instrument, Equal Strength & Strong Instrument, Variable Strength & Weak Instrument, Equal Strength & Weak Instrument, Variable Strength \\ \hline
0 & 0.15 & 0.17 & 0.16 & 0.17 \\
0.25 & 0.54 & 0.54 & 0.53 & 0.53 \\
0.5 & 0.73 & 0.73 & 0.53 & 0.73 \\
0.75 & 0.87 & 0.87 & 0.88 & 0.87
\end{tabular}
\label{tab:tablerhoL100}
\end{table}

\begin{table}[h!]
\centering
\begin{tabular}{p{2cm} p{3.4cm} p{3.4cm} p{3.4cm} p{3.4cm} }
\multicolumn{5}{c}{Table 6. Condition on $s$ in Corollary 2 for the Simulation Study} \\ \hline \hline
 Instrument Corr. $(\mu)$ & Strong Instrument, Equal Strength & Strong Instrument, Variable Strength & Weak Instrument, Equal Strength & Weak Instrument, Variable Strength \\ \hline
0 & 4.2 & 3.3 & 4.0 & 3.4 \\
0.25 & 0.33 & 0.33 & 0.33 & 0.33 \\
0.5 & 0.17 & 0.17 & 0.17 & 0.17 \\
0.75 & 0.11 & 0.11 & 0.11 & 0.11
\end{tabular}
\label{tab:tableSLimitL100}
\end{table}

Figures \ref{fig:equalZcorrEndoL100} and \ref{fig:equalZcorrSL100} represent the results from the simulation setup when we measure the median of $|\hat{\beta}^* - \beta^*|$ over 500 simulations; this setup is identical to Section 4 in the main manuscript except for the exceptions mentioned in the previous paragraph. The behavior of all four estimators are similar to Figures 2 and 3 in the main manuscript. For example, when we vary endogeneity (Figure \ref{fig:equalZcorrEndoL100}), sisVIVE tends to perform slightly worse when the overall strength of the instruments is weak. Also, when the number of invalid instruments, $s$, is varied (Figure \ref{fig:equalZcorrSL100}), sisVIVE has a sharp peak at $s = 50$, similar to the sharp peak at $s = 5$ in Figures 3 of the main manuscript. 

Figures \ref{fig:equalZcorrEndoL100Percent} and \ref{fig:equalZcorrSL100Percent} represent the simulation setups in Section \ref{sec:propBasedError} of the Supplementary Materials. Similar to what we observed in Section \ref{sec:propBasedError} when $L = 10$, when we vary endogeneity (Figure \ref{fig:equalZcorrEndoL100Percent}), but fix the number of invalid instruments to $30$, we see that the proportion of correctly selected invalid instruments are $1$. When we vary $s$ (Figure \ref{fig:equalZcorrSL100Percent}), we again notice a sharp decrease in the proportion of correctly selected valid invalid instruments around $s=50$ for all instrument strength and magnitude of the correlation. 

Overall, the simulation study suggests that sisVIVE does scale as $L$ increases and that its performance at large values of $L$ is similar to its performance at smaller values of $L$, such as $L = 10$.

\subsection{Choice of $\lambda$} \label{sec:choiceofLambda}
In this section, we look at different ways to select $\lambda$. As discussed in the main manuscript, the choice of $\lambda$ impacts the performance of sisVIVE where a high value of $\lambda$ will push most elements of $\hat{\bm{\alpha}}_{\lambda}$ to zero while a low value of $\lambda$ will do the opposite. In Section 3.3 of the main manuscript, we suggested cross-validation with the ``one standard error'' rule as a data-driven method to choosing the tuning parameter. In addition, in Section 3.4, we provided theoretical results which suggested choosing a $\lambda$ that is greater than $ 3 \|\mathbf{Z}^T \mathbf{P}_{\hat{\mathbf{D}}^\perp}  \bm{\epsilon}\|_\infty$. We explore these two possible choices of $\lambda$ and their impact on estimation.

We begin with a simulation study similar to the one in the main manuscript. In particular, we have $L = 10$ instruments of which the pairwise correlation between all instruments is $0.75$ and the endogeneity is fixed at $0.8$. We vary $s$, the number of invalid instruments and vary instruments' absolute strength, relative strength, and other strengths considered in Section \ref{sec:strengthvary} of the Supplementary Materials. In short, the simulation setups we consider correspond to the last row of Figure 3 in the main manuscript and the last row of Figure \ref{fig:equalZcorrS-awkward} in the Supplementary Materials. We do not simulate other correlation structures or different $L$ because the simulation results in Sections \ref{sec:varysimstructure} and \ref{sec:Lchange} of the Supplementary Materials showed sisVIVE behaves similarly as the cases we consider in this Section.

Table 7 shows the different values of $\lambda$ averaged across 500 simulations where the overall, absolute instrument strength is strong (see Section 4 of the main manuscript for details on the definition of an absolute instrument strength). We use the same column heading labels in Figure 3 of the main manuscript and Figure \ref{fig:equalZcorrS-awkward} in the Supplementary Materials. We also use the column labeled ``CV'' to denote the average $\lambda$s based on cross validation laid out in Section 3.3 of the main manuscript. Also, the column labeled ``Theory'' denotes the average $\lambda$s based on Theorem 2, specifically the average of $3 \|\mathbf{Z}^T \mathbf{P}_{\hat{\mathbf{D}}^\perp}  \bm{\epsilon}\|_\infty$ over 500 simulations. In almost all cases, cross validation tends to choose a smaller $\lambda$ than one prescribed by Theorem 2, with the exception of $s = 9$ in the ``Equal'' column and $s = 7,8$, and $9$ in the ``Stronger Valid'' column. Except for these cases, cross validation tends to prefer a small $\lambda$, thereby preferring $\hat{\bm{\alpha}}_\lambda$ to have more non-zero entries than zero entries and more instruments selected as invalid instruments than valid instruments. 
\begin{table}[h!]
\centering
\begin{tabular}{c c c c c c c c c}
\multicolumn{9}{p{12cm}}{Table 7. Average $\lambda$ from cross validation and Theorem 2 after $500$ simulations for instruments whose overall strength is strong.} \\ \hline \hline
& \multicolumn{2}{c}{Equal} & \multicolumn{2}{c}{Variable} & \multicolumn{2}{c}{Stronger Invalid} & \multicolumn{2}{c}{Stronger Valid} \\ \cline{2-9}
$s$ & CV & Theory & CV & Theory & CV & Theory & CV & Theory \\ \hline
1 & 1.88 & 2.70 & 2.04 & 2.71 & 1.53 & 2.70 & 2.06 & 2.72 \\
2 & 1.36 & 2.66 & 1.39 & 2.67 & 0.95 & 2.65 & 1.58 & 2.68 \\
3 & 1.06 & 2.64 & 1.12 & 2.66 & 0.84 & 2.64 & 1.33 & 2.68 \\ 
4 & 0.84 & 2.64 & 0.86 & 2.65 & 1.08 & 2.63 & 1.16 & 2.68 \\
5 & 1.70 & 2.63 & 1.33 & 2.64 & 0.87 & 2.62 & 0.99 & 2.67 \\
6 & 1.78 & 2.62 & 1.10 & 2.63 & 0.85 & 2.61 & 0.96 & 2.67 \\
7 & 2.02 & 2.62 & 0.79 & 2.64 & 0.91 & 2.61 & 3.40 & 2.68 \\
8 & 2.41 & 2.62 & 0.86 & 2.62 & 1.01 & 2.61 & 3.74 & 2.67 \\
9 & 3.19 & 2.62 & 0.45 & 2.62 & 1.31 & 2.60 & 6.03 & 2.67 
\end{tabular}
\end{table}

Table 8 shows the estimation performance of sisVIVE, the median of $|\beta^* - \hat{\beta}_{\lambda}|$ over 500 simulations, based on two different $\lambda$s, one based on cross validation and one based on Theorem 2. In most cases, sisVIVE with a cross validated $\lambda$ performs just as well as sisVIVE with a theory-based $\lambda$. For the ``Equal'' and ''Variable'' case, when $s < 5$, sisVIVE with a cross-validated $\lambda$ performs better than sisVIVE with a theory-based $\lambda$. For the ``Stronger Invalid'' case, when $s < 3$, sisVIVE with a cross validated $\lambda$ performs better than sisVIVE with a theory-based $\lambda$. However, when $s \geq 3$, sisVIVE with a cross validated $\lambda$ performs worse than sisVIVE with a theory-based $\lambda$, although the differences between the two decrease as $s$ increases. For the ``Stronger Valid'' case, sisVIVE with a cross validated $\lambda$ always dominates sisVIVE with a theory-based $\lambda$, although the differences between the two are slight when $s \geq 7$.
\begin{table}[h!]
\centering
\begin{tabular}{c c c c c c c c c}
\multicolumn{9}{p{12cm}}{Table 8. Median absolute estimation error ($|\beta^* - \hat{\beta}_{\lambda}|$) after 500 simulations from $\lambda$ chosen by cross-validation and Theorem 2. The table only considers instruments whose overall strength is strong.} \\ \hline \hline
& \multicolumn{2}{c}{Equal} & \multicolumn{2}{c}{Variable} & \multicolumn{2}{c}{Stronger Invalid} & \multicolumn{2}{c}{Stronger Valid} \\ \cline{2-9}
$s$ & CV & Theory & CV & Theory & CV & Theory & CV & Theory \\ \hline
1 & 0.13 & 0.17 & 0.14 & 0.16 & 0.13 & 0.19 & 0.14 & 0.16 \\
2 & 0.16 & 0.27 & 0.16 & 0.27 & 0.16 & 0.34 & 0.16 & 0.24 \\
3 & 0.18 & 0.39 & 0.18 & 0.37 & 0.24 & 0.54 & 0.18 & 0.32 \\ 
4 & 0.21 & 0.53 & 0.22 & 0.53 & 1.57 & 1.34 & 0.20 & 0.41 \\
5 & 0.71 & 1.15 & 0.76 & 1.43 & 1.43 & 1.25 & 0.23 & 0.55 \\
6 & 2.43 & 2.34 & 2.05 & 1.93 & 1.35 & 1.23 & 0.28 & 0.71 \\
7 & 2.42 & 2.37 & 1.83 & 1.95 & 1.28 & 1.21 & 3.83 & 3.95 \\
8 & 2.35 & 2.34 & 1.98 & 2.05 & 1.22 & 1.18 & 4.24 & 4.39 \\
9 & 2.29 & 3.01 & 1.23 & 1.37 & 1.17 & 1.16 & 4.34 & 4.51
\end{tabular}
\end{table}

Table 9 considers the same setup as Table 7, except we now look at instruments where their overall, absolute strength is weak. Under this case, we see drastic differences between $\lambda$s chosen based on cross validation and Theorem 2. For example, for the ``Equal'' and ``Variable'' cases, when $s < 5$,  $\lambda$ chosen based on cross validation is, on average, smaller than $\lambda$ chosen based on Theorem 2. When $s \geq 5$, $\lambda$ chosen based on cross validation is, on average, bigger than $\lambda$ chosen based on Theorem 2. For the ``Stronger Invalid'' case, when $s < 3$, $\lambda$ based on cross validation is, on average, smaller than $\lambda$ based on Theorem 2. But, when $s \geq 3$, the opposite is the case. Finally, for the ``Stronger Valid'' case, this phenomena occurs at $s = 6$.

\begin{table}[h!]
\centering
\begin{tabular}{c c c c c c c c c}
\multicolumn{9}{p{12cm}}{Table 9. Average $\lambda$ from cross validation and Theorem 2 after $500$ simulations for instruments whose overall strength is weak.} \\ \hline \hline
& \multicolumn{2}{c}{Equal} & \multicolumn{2}{c}{Variable} & \multicolumn{2}{c}{Stronger Invalid} & \multicolumn{2}{c}{Stronger Valid} \\ \cline{2-9}
$s$ & CV & Theory & CV & Theory & CV & Theory & CV & Theory \\ \hline
1 & 1.36 & 3.20 & 1.56 & 3.23 & 1.05 & 3.13 & 1.52 & 3.24 \\
2 & 1.25 & 3.00 & 1.22 & 3.01 & 0.93 & 2.92 & 1.47 & 3.07 \\
3 & 1.12 & 2.91 & 1.11 & 2.94 & 3.67 & 2.81 & 1.26 & 3.00 \\
4 & 2.06 & 2.86 & 1.83 & 2.89 & 9.47 & 2.75 & 1.13 & 2.97 \\
5 & 6.30 & 2.80 & 4.34 & 2.84 & 10.52 & 2.71 & 1.20 & 2.92 \\
6 & 11.99 & 2.78 & 7.48 & 2.80 & 10.74 & 2.69 & 3.36 & 2.93 \\
7 & 14.14 & 2.76 & 5.92 & 2.77 & 10.58 & 2.67 & 7.79 & 2.93 \\
8 & 14.04 & 2.75 & 5.94 & 2.75 & 9.92 & 2.66 & 9.70 & 2.93 \\
9 & 13.16 & 2.74 & 2.02 & 2.68 & 9.47 & 2.64 & 7.09 & 2.96
\end{tabular}
\end{table}

Table 10 considers the same setup as Table 8, except we now look at instruments where their overall, absolute strength is weak. Similar to Table 8, sisVIVE with a cross validated $\lambda$ performs better than sisVIVE with a theory-based $\lambda$, with the only  exception at $s = 5$ under ``Equal'' column. In fact, sisVIVE with a cross validated $\lambda$ performs drastically better than sisVIVE based on Theorem 2 in the following cases: $s < 5$ (for ``Equal'' and ``Variable'' cases), $s < 3$ (for ``Stronger Invalid'' case), and $s < 7$ (for ``Stronger Valid'' case).

\begin{table}[h!]
\centering
\begin{tabular}{c c c c c c c c c}
\multicolumn{9}{p{12cm}}{Table 10. Median absolute estimation error ($|\beta^* - \hat{\beta}_{\lambda}|$) after 500 simulations from $\lambda$ chosen by cross-validation and Theorem 2. The table only considers instruments whose overall strength is weak.} \\ \hline \hline
& \multicolumn{2}{c}{Equal} & \multicolumn{2}{c}{Variable} & \multicolumn{2}{c}{Stronger Invalid} & \multicolumn{2}{c}{Stronger Valid} \\ \cline{2-9}
$s$ & CV & Theory & CV & Theory & CV & Theory & CV & Theory \\ \hline
1 & 0.44 & 0.63 & 0.44 & 0.60 & 0.43 & 0.69 & 0.44 & 0.61 \\
2 & 0.51 & 0.96 & 0.50 & 0.94 & 0.50 & 1.13 & 0.52 & 0.88 \\
3 & 0.55 & 1.30 & 0.55 & 1.26 & 0.70 & 1.86 & 0.56 & 1.13 \\
4 & 0.61 & 1.74 & 0.61 & 1.75 & 3.19 & 3.77 & 0.58 & 1.43 \\
5 & 4.10 & 3.80 & 3.98 & 3.93 & 3.25 & 3.78 & 0.62 & 1.83 \\
6 & 5.28 & 6.03 & 5.28 & 5.54 & 3.36 & 3.79 & 0.73 & 2.52 \\
7 & 5.84 & 6.55 & 5.58 & 5.63 & 3.47 & 3.77 & 7.51 & 7.68 \\
8 & 6.29 & 6.75 & 6.19 & 6.19 & 3.52 & 3.70 & 9.69 & 9.77 \\
9 & 6.72 & 6.90 & 4.18 & 4.34 & 3.56 & 3.64 & 10.86 & 10.91 
\end{tabular}
\end{table}

Based on these simulations, sisVIVE based on cross-validation generally performs better than sisVIVE based on Theorem 2, especially when the overall instrument strength is weak. We also note that cross validation tends to choose a smaller $\lambda$ than the one based on Theorem 2, suggesting that for better estimation, it is preferable to set only a few elements of $\hat{\bm{\alpha}}_\lambda$ to zero and declare more instruments to be invalid than valid. This observation was also seen in our simulation in Section \ref{sec:propBasedError} where low median absolute error, $|\beta^* - \hat{\beta}_{\lambda}|$, was tied to high proportion of correctly chosen invalid instruments. We note that this observation is in contrast with estimating sparse vectors in typical high dimensional regression settings where many zeroed elements are desirable in the estimated sparse vector.

Despite the simulation evidence suggesting the use of cross validation to choose $\lambda$ over Theorem 2 to choose $\lambda$, unfortunately, there is little theory to justify the use of cross validation in $\ell_1$ penalization settings \citep{hastie_elements_2009, buhlmann_statistics_2011}. However, Section 2.5.1 of \citet{buhlmann_statistics_2011} does provide limited theoretical results suggesting that $\lambda$ based on cross validation tends to set few elements of $\hat{\bm{\alpha}}_{\lambda}$ to zero, a desirable property in our setting where we want to select more instruments to be invalid than valid for better estimation performance of $\hat{\beta}_\lambda$.

Besides cross validation and Theorem 2, there is another way to choose $\lambda$ if we assume Corollary 1 holds for our data. That is, if we are in the always identified region where $s < U \leq L/2$, one possible method of choosing $\lambda$ would be to find the $\lambda$ where exactly $U = L/2$, say $\lambda_{L/2}$. From there, we grid the values of potential $\lambda$s between $0$ and $\lambda_{L/2}$ and choose the $\lambda$ that minimizes the estimating equation $||\mathbf{P}_{\mathbf{Z}}(\mathbf{Y} - \mathbf{Z} \bm{\alpha} - \mathbf{D} \beta)||_2$. It would be interesting to investigate this method in future research.

\section{Additional Discussion about Theorem 2}
Theorem 2 is written in terms of the restricted isometry type (RIP) condition while its corresponding Corollary 2 is written in terms of the mutual incoherence property (MIP) condition. As the main text states, the RIP condition implies the MIP condition, but not vice versa. We illustrate this relationship with the following simple example. Suppose the matrix of instruments $\mathbf{Z}$ is an $n$ by $L$ matrix where each entry $Z_{ij}$ are from i.i.d. standard Normal. Based on Theorem 5.2 in \citet{Baraniuk}, when $n \geq Cs\log (L/s)$ for some $C$ not dependent on $L$ and $s$, we are able to ensure the RIP condition $2\delta_{2s}^-(\mathbf{Z}) > 3\delta_{3s}^+(\mathbf{Z})$ with high probability. Here, $2\delta_{2s}^-(\mathbf{Z}) > 3\delta_{3s}^+(\mathbf{Z})$ is a stronger condition than $2\delta_{2s}^-(\mathbf{Z}) > \delta_{3s}^+(\mathbf{Z}) + 2 \delta_{2s}^+(\mathbf{P}_{\hat{\mathbf{D}}}\mathbf{Z})$, the RIP condition we need for Theorem 2. However, based on Theorem 8 in \citet{Cai_Fan_Jiang}, to guarantee our MIP condition $\mu < \frac{1}{12s}$, we need $n \geq C s^2\log L$ for some $C$ not dependent on $L$ and $s$. In short, when the order of $n$ is between $s\log (L/s)$ and $s^2\log L$, $\mathbf{Z}$ meet the RIP condition but not the MIP condition, with high probability.

\section{Wisconsin Longitudinal Data}
\subsection{Background of Data}
This research uses data from the Wisconsin Longitudinal Study (WLS) of the University of Wisconsin-Madison. Since 1991, the WLS has been supported principally by the National Institute on Aging (AG-9775, AG-21079, AG-033285, and AG-041868), with additional support from the Vilas Estate Trust, the National Science Foundation, the Spencer Foundation, and the Graduate School of the University of Wisconsin-Madison. Since 1992, data have been collected by the University of Wisconsin Survey Center. A public use file of data from the Wisconsin Longitudinal Study is available from the Wisconsin Longitudinal Study, University of Wisconsin-Madison, 1180 Observatory Drive, Madison, Wisconsin 53706 and at http://www.ssc.wisc.edu/wlsresearch/data/. The opinions expressed herein are those of the authors.

\subsection{Reduced form estimates}
Tables 7 and 8 summarize the reduced form estimates for the data analysis in the Main manuscript. The reduced form estimates are computed by using ordinary least squares (OLS) where the genetic instruments are the explanatory variables and the dependent variables are body mass index (BMI) and Health Utility Index Mark 3 (HUI3). 

\begin{table}[h!]
\centering
\begin{tabular}{l r r}
\multicolumn{3}{c}{Table 7. Reduced Form Estimates for HUI-3 and BMI for Three Instruments} \\ \hline \hline
Instruments & BMI (SE) & HUI-3 (SE) \\ \hline 
rs1421085 & -0.20 (0.07) & 0.0003 (0.004) \\ 
rs1501299 & 0.03 (0.08) & 0.002 (0.005) \\ 
rs2241766 & -0.04 (0.11) & -0.0001 (0.007)  
\end{tabular}
\end{table} 

\begin{table}[h!]
\centering
\begin{tabular}{l r r}
\multicolumn{3}{c}{Table 8. Reduced Form Estimates for HUI-3 and BMI for Four Instruments} \\ \hline \hline
Instruments & BMI (SE) & HUI-3 (SE) \\ \hline 
rs1421085 & -0.20 (0.07) & 0.0004 (0.004) \\ 
rs1501299 & 0.03 (0.08) & 0.002 (0.005) \\ 
rs2241766 & -0.04 (0.11) & -0.0004 (0.007)  \\
rs6265 & -0.008 (0.08) & -0.008 (0.005) 
\end{tabular}
\end{table} 

\subsection{Sargan overidentification test} 
For the data analysis with three SNPs, the Sargan overidentification test \citet{sargan_estimation_1958}, which tests assumptions (A2) and (A3) in the presence of multiple instruments, gives a Chi-squared value of $0.12$ (p-value: $0.94$), retaining the null hypothesis that the instruments are all valid under the $0.05$ significance level. For the data analysis with four SNPs, the Sargan overidentification test gives a Chi-squared value of $2.49$ (p-value: $0.48$).

\section{Proofs}
We adopt the following notations for the proofs. For any sets $A,B \subseteq \{1,\ldots,L\}$, denote $A \cap B$ to be the intersection of sets $A$ and $B$, $A \cup B$ to be the union of sets $A$ and $B$, and $A^C$ and $B^C$ to be the complement of sets $A$ and $B$, respectively. If $A \subseteq B$, denote $B \setminus A$ to be the set that comprises of all the elements of $B$ except those that are in $A$. Let $|A|$ and $|B|$ denote the cardinality of the sets $A$ and $B$, respectively. 

For any vector $\bm{\alpha} \in \reals^L$ and set $A \subseteq \{1,\ldots,L\}$, denote $\bm{\alpha}_{A} \in \reals^L$ to be the vector where all the elements except whose indices are in $A$ are zero. Also, denote the $j$th element as $\alpha_j$. Let supp$(\bm{\alpha}) \subseteq \{1,\ldots,L\}$ to be the support of the vector $\bm{\alpha}$ and supp$(\bm{\alpha})^C$ be the complement set. For any matrix $\mathbf{M} \in \reals^{n \times L}$ and set $A \subseteq \{1,\ldots,p\}$, let  $\mathbf{M}_A \in \reals^{n \times L}$ be an $n$ by $|A|$ matrix where the columns are specified by set $A$.

\subsection{Proof of Theorem 1}
First, we prove that, $\beta^*$ is a unique solution if and only if $\bm{\alpha}^{*}$ is a unique solution. Suppose $\beta^*$ has a unique solution; that is, for any two solutions $\bm{\alpha}^{(1)}$ $\beta^{(1)}$ and $\bm{\alpha}^{(2)}, \beta^{(2)}$, in equation (7)
\begin{subequations}
\label{eq:momentCond4}
\begin{align}
\bm{\alpha}^{(1)} + \bm{\gamma}^* \beta^{(1)} &= \mathbf{\Gamma}^* \\
 \bm{\alpha}^{(2)} + \bm{\gamma}^* \beta^{(2)} &= \mathbf{\Gamma}^* 
\end{align}
\end{subequations}
we have $\beta^{(1)} = \beta^{(2)}$. Subtracting $\bm{\gamma}^* \beta^{(1)}$ from equations \eqref{eq:momentCond4} gives $\bm{\alpha}^{(1)} = \bm{\alpha}^{(2)}$. Now, suppose $\bm{\alpha}^*$ is unique, which implies $\bm{\alpha}^{(1)} = \bm{\alpha}^{(2)}$. Again, subtracting $\bm{\alpha}^{(1)}$ from \eqref{eq:momentCond4} reveals $\beta^{(1)} = \beta^{(2)}$. 

Second, we prove the necessary and sufficient conditions for Theorem 1. Suppose the subspace conditions on $\bm{\gamma}^*$ and $\bm{\Gamma}^*$ hold, specifically $q_m = q_{m'}$ for any $m \neq m'$, but there are two distinct sets of parameters, $\bm{\alpha}^{(1)}, \beta^{(1)}$ and $\bm{\alpha}^{(2)}, \beta^{(2)}$ that solve the moment equation in equation \eqref{eq:momentCond4}. Let $A^{(1)} = $supp$(\bm{\alpha}^{(1)})$ and $A^{(2)} = $supp$(\bm{\alpha}^{(2)})$ be the sets of invalid instruments for the two distinct parameter sets, not equal to each other; if the supports are equal to each other, we have the degenerate case whereby from equation \eqref{eq:momentCond4}, for any $j \in A^{(1)} = A^{(2)}$ $\gamma_j^* \beta^{(1)} = \Gamma_j^*$ and $\gamma_j^* \beta^{(2)} = \Gamma_j^*$, which implies that $\beta^{(1)} = \beta^{(2)}$ and $\bm{\alpha}^{(1)} = \bm{\alpha}^{(2)}$, a contradiction. Because the number of invalid instruments, $s$, is less than $U$, $s < U$, the number of valid instruments, $L - s$, must be greater than $L - U$, $L -s > L - U$. Thus, $|(A^{(1)})^C|, |(A^{(2)})^C| > L - U$. 

Now, pick any subsets, $(A^{(1')})^C$ and $(A^{(2')})^C$, of $(A^{(1)})^C$ and $(A^{(2)})^C$, respectively, where $|(A^{(1')})^C| = |(A^{(2')})^C| = L - U + 1$. These subsets $(A^{(1')})^C$ and $(A^{(2')})^C$ inherit the following property from their larger sets $(A^{(1)})^C$ and $(A^{(2)})^C$, respectively.
\begin{align*}
\alpha_{j}^{(1)} + \gamma_{j}^* \beta^{(1)} = \gamma_j^* \beta^{(1)} = \Gamma_j^*, \quad{} j \in (A^{(1')})^C \subseteq (A^{(1)})^C \\
 \alpha_{k}^{(2)} + \gamma_{k}^* \beta^{(2)} = \gamma_k^* \beta^{(2)} = \Gamma_k^*, \quad{} k \in (A^{(2')})^C \subseteq (A^{(2)})^C
\end{align*}
The subspace condition on $\bm{\gamma}^*$ and $\bm{\Gamma}^*$ in Theorem 1 state that for any sets $C_m$ with size $|C_m| = L - U + 1$ and with the property that $\gamma_j q_m = \Gamma_j, j \in C_m$, we have $q_m = q_{m'}$ for any $m, m'$. The subsets we constructed, $(A^{(1')})^C$ and $(A^{(2')})^C$, satisfy these subspace condition with constants $q_{1'} = \beta^{(1)}$ and $q_{2'} = \beta^{(2)}$. Hence, $\beta^{(1)} = q_{1'} = q_{2'} = \beta^{(2)}$, which is a contradiction. Hence, the two sets of parameters $\bm{\alpha}^{(1)}, \beta^{(1)}$ and $\bm{\alpha}^{(2)}, \beta^{(2)}$ are identical to each other and the solution is unique.

Now, suppose the solution is unique. Then, we show that the subspace conditions on $\bm{\gamma}^*$ and $\bm{\Gamma}^*$ must hold. Pick any two sets $A^{(1)}, A^{(2)} \subseteq \{1,\ldots,L\}$ with their complements having the size $|(A^{(1)})^C| = |(A^{(2)})^C| = L - U + 1$ and corresponding constants $q_1$ and $q_2$, respectively, defined in the Theorem. We have to show that $q_1 = q_2$ for any pair of two sets.

Note that at least one set of these sets and its corresponding constant $q$ must exist because at the true parameter values, $\bm{\alpha}^*$ and $\beta^*$, equation (7) is satisfied. Specifically, if $A^* = $supp$(\bm{\alpha}^*)$ where, by $s < U$, $|(A^*)^C| = |$supp$(\bm{\alpha}^*)^C| > L - U$, we can take any subset $(A^{(*')})^C \subseteq (A^{*})^C$ of size $|(A^{(*')})^C| = L - U + 1$. For any $j \in (A^{(*')})^C$, by equation (7), $\gamma_j^* \beta^* = \Gamma_j^*$ and thus, its corresponding constant $q_{*'}$ is $q_{*'} = \beta^*$. If there is exactly one set $A^{(1)}$, the subspace condition holds automatically. 

Suppose there are two or more sets and let $A^{(1)}$ and $^{(2)}$ be any pair of the sets. Based on the sets $A^{(1)}$ and $A^{(2)}$ and their corresponding constants $q_1$ and $q_2$, we construct the following sets of parameters $\bm{\alpha}^{(1)}, \beta^{(1)}$ and $\bm{\alpha}^{(2)}, \beta^{(2)}$
\begin{align*}
\beta^{(1)} = q_1, &\quad{} \alpha_j^{(1)} = \begin{cases}
0 & j \in (A^{(1)})^C \\
\Gamma_j^* - q_1 \gamma_j^* & j \in A^{(1)}
\end{cases} \\
\beta^{(2)} = q_2, &\quad{} \alpha_j^{(2)} = \begin{cases}
0 & j \in (A^{(2)})^C \\
\Gamma_j^* - q_2 \gamma_j^* & j \in A^{(2)}
\end{cases}
\end{align*}
The cardinality of $\bm{\alpha}^{(1)}$ and $\bm{\alpha}^{(2)}$ are less than $U$. In addition, they satisfy the moment equation in equation (7).
\begin{align*}
\alpha_j^{(1)} + \gamma_j^* \beta^{(1)} &= \begin{cases}
\gamma_j^*  q_1 = \Gamma_j^* & j \in (A^{(1)})^C \\
 \Gamma_j^* - q_1 \gamma_j^* + \gamma_j^* q_1 = \Gamma_j^*  & j \in A^{(1)}  
\end{cases} \\
\alpha_j^{(2)} + \gamma_j^* \beta^{(2)} &= \begin{cases}
\gamma_j^*   q_2 = \Gamma_j^* & j \in (A^{(2)})^C \\
 \Gamma_j^* - q_2 \gamma_j^* + \gamma_j^* q_2 = \Gamma_j^* & j \in A^{(2)}  
\end{cases} 
\end{align*}
Since the equation has only one unique solution, this implies that $\beta^{(1)} = \beta^{(2)}$, or $q_1 = q_2$. Since this holds for any two sets $(A^{(1)})^C, (A^{(2)})^C$ with constants $q_1$ and $q_2$ and cardinality $L - U + 1$, we arrive at the subspace condition $q_m = q_{m'}$ for any $m, m'$.
\qed

\subsection{Proof of Corollary 1}
Consider any two sets $C_m$ and $C_{m'}$ with the constants $q_m$ and $q_{m'}$ in Theorem 1. Take an element $j$ from the intersection $C_m \cap C_{m'}$; this intersection is non-empty because $|C_m| = |C_{m'}| = L - U + 1 \geq L/2 + 1$. At element $j \in C_m \cap C_{m'}$, we have $ \gamma_j^* q_m = \Gamma_j^*$ and $\gamma_j^* q_{m'} = \Gamma_j^*$, which implies $q_m = q_{m'}$. Since this holds for any two sets $C_m$ and $C_{m'}$, $q_m = q_{m'}$ for $m,m'$, the subspace restriction condition in Theorem 1 always holds whenever $U \geq L/2$ and we have identification.
\qed

\subsection{Proof of Theorem 2}

We begin by introducing some notations and terminologies. For $\aalpha \in \reals^p$ and $s \in \{1,\ldots,p\}$, $\aalpha_{\max(s)}$ is defined as the vector where all but the largest $s$ elements set to zero and $\aalpha_{-\max(s)}$ is defined as $\aalpha - \aalpha_{\max(s)}$.
\begin{definition}
The restricted orthogonal constant (ROC) of single matrix of order $k_1$ and $k_2$, denoted as $\theta_{k_1,k_2}(\M)$, is the smallest $\theta_{k_1,k_2}(\M)$ where for any $k_1$-sparse vector $\aalpha_1$ and $k_2$-sparse vector $\aalpha_2$ with non-overlapping support, we have
\[
|\langle \M\aalpha_1,\M\aalpha_2\rangle | \leq \theta_{k_1,k_2}(\M) \|\aalpha_1\|_2 \|\aalpha_2\|_2.
\]
\end{definition}

Next, we introduce two lemmas. The first Lemma relates the RIP and ROC constants.
\begin{lemma} \label{lm:RIP-ROC}
For any matrix $\M$ and positive integers $s_1$ and $s_2$, 
\[
\theta_{s_1,s_2}(\M) \leq \frac{1}{2}\left(\delta_{s_1 + s_2}^{+}(\M) - \delta_{s_1 + s_2}^{-}(\M)\right).
\]
\end{lemma}

\begin{proof}
For any vectors $x$ and $y$ with disjoint supports and $\|x\|_2 = \|y\|_2 = 1$, we must have $x+y$, $x-y$ are both $(s_1+s_2)$-sparse and $\|x+y\|_2^2 = \|x-y\|_2^2 =2$. Hence,
\begin{equation*}
\begin{split}
|\langle \M x, \M y\rangle| = & \frac{1}{4} \left|\|\M(x+y)\|_2^2 - \|\M(x-y)\|_2^2\right| \\
 = & \frac{1}{4}\max\left\{\|\M(x+y)\|_2^2 - \|\M(x-y)\|_2^2, \|\M(x-y)\|_2^2 - \|\M(x+y)\|_2^2\right\}\\
 \leq & \frac{1}{4}\max\Big\{\delta_{s_1+s_2}^+(\M)\|x+y\|_2^2 - \delta_{s_1+s_2}^-(\M)\|x-y\|_2^2,\\
 &  \delta_{s_1+s_2}^+(\M)\|x-y\|_2^2 - \delta_{s_1+s_2}^-(\M)\|x+y\|_2^2\Big\}\\
 \leq & \frac{1}{2}\left(\delta_{s_1+s_2}^+(\M) - \delta_{s_1+s_2}^-(\M)\right),
\end{split}
\end{equation*}
which implies $\theta_{s_1, s_2}(\M) \leq \frac{1}{2}\left(\delta_{s_1+s_2}^+(\M) - \delta_{s_1+s_2}^-(\M)\right)$.
\end{proof}

The second Lemma proves a standard property of the Lasso.
\begin{lemma}\label{lm:Lasso_RIP}
Suppose we have the model $Y_i = \Z_{i.}^T \aalpha^*+ \epsilon_i$ where 
%$E(\bm{\epsilon}_i | \Z_{i.} ) = 0$ and 
$\aalpha^*$ is $s$-sparse. Further suppose that matrix $\Z$ has upper and lower RIP constants $\delta_{s}^{+}(\Z)$ and $\delta_{s}^{-}(\Z)$, respectively. Define $\hat{\aalpha}$ as the Lasso estimator
\begin{equation} \label{eq:lasso1}
\hat{\aalpha}_{\lambda} = \amin{\aalpha} \frac{1}{2}\|Y - \Z\aalpha\|_2^2 + \lambda\|\aalpha \|_1
\end{equation}
and let $h = \hat{\aalpha}_{\lambda} - \aalpha^*$ measure the errors of the estimator. 

If $r\|\Z^T\bm{\epsilon}\|_\infty \leq \lambda$ for some $r>1$, we have
\begin{equation} \label{eq:lasso2}
\|h_{-\max(s)}\|_1 \leq \frac{r+1}{r-1} \|h_{\max(s)}\|_1.
\end{equation}
Furthermore, if $ (r+1)\delta_{2s}^+(\Z) <  (3r-1)\delta_{2s}^-(\Z) $, 
\begin{equation} \label{eq:lasso3}
\|h_{\max(s)}\|_2 \leq \frac{2\lambda\sqrt{s}(r-1)(r+1)/r}{(3r-1)\delta_{2s}^-(\Z)-(r+1)\delta_{2s}^+(\Z)}.
\end{equation}
\end{lemma}

\begin{proof}
Since $\hat{\aalpha}_{\lambda}$ is the minimizer of \eqref{eq:lasso1} ,  we have
\[
\frac{1}{2}\|Y - \Z\hat{\aalpha}_{\lambda}\|_2^2 + \lambda\|\hat{\aalpha}_{\lambda}\|_1 \leq \frac{1}{2}\|y - \Z\aalpha^*\|_2^2 + \lambda\|\aalpha^*\|_1.
\]
By the assumed model $Y_i = \Z_i^T \aalpha^*+\epsilon_i$, we have
\begin{equation}\label{eq:lasso4}
 \frac{1}{2}\left(\|\bm{\epsilon} - \Z h\|_2^2 - \|\bm{\epsilon}\|_2^2\right) \leq \lambda(\|\aalpha^*\|_1 - \|\hat{\aalpha}_{\lambda}\|_1).
\end{equation}
For the upper bound of \eqref{eq:lasso4}, the fact that $\aalpha^*$ is $s$-sparse gives a useful bound. Specifically, 
\begin{align*}
\|\aalpha^*\|_1 - \|\hat{\aalpha}_{\lambda}\|_1 & = \|\aalpha_{supp(\aalpha^*)}^*\|_1 - \|\hat\aalpha_{supp(\aalpha^*)}\|_1 - \|\hat\aalpha_{supp(\aalpha^*)^c}\|_1 \\ 
 & \leq \|\aalpha_{supp(\aalpha^*)}^* - \hat\aalpha_{supp(\aalpha^*)}\|_1 - \|h_{supp(\aalpha^*)^c}\|_1\\
 & \leq \|h_{supp(\aalpha^*)}\|_1 - \|h_{supp(\aalpha^*)^c}\|_1\\
&\leq \|h_{\max(s)}\|_1 - \|h_{-\max(s)}\|_1.
\end{align*}
For the lower bound of \eqref{eq:lasso4}, $\|\bm{\epsilon} - \Z h\|_2^2 - \|\bm{\epsilon}\|_2^2$, we can simplify as
\begin{align*}
 \frac{1}{2}\left(\|\bm{\epsilon} - \Z h\|_2^2 - \|\bm{\epsilon}\|_2^2\right) &= -\frac{1}{2}(\Z h)^T(2\bm{\epsilon} - \Z h) \geq -h^T\Z^T\bm{\epsilon} 
 \geq -  \|\Z^T\bm{\epsilon}\|_\infty \|h\|_1 \\
&=- \|\Z^T \bm{\epsilon}\|_{\infty} (\|h_{\max(s)}\|_1 + \|h_{-\max(s)}\|_1).
\end{align*}
Hence, by \eqref{eq:lasso4} and the condition $r\|\Z^T \bm{\epsilon}\|_{\infty} \leq \lambda$ where $r > 1$, we have 
\[
r (\|h_{\max(s)}\|_1 - \|h_{-\max(s)}\|_1) \geq - (\|h_{\max(s)}\|_1 + \|h_{-\max(s)}\|_1).
\]
which yields \eqref{eq:lasso2}, the first part of the theorem.

For \eqref{eq:lasso3}, the second part of the theorem, suppose $(r+1)\delta_{2s}^{+}(\Z) < (3r - 1)\delta_{2s}^{-}(\Z)$ holds. By the Karush-Kuhn-Tucker (KKT) condition of the minimization problem in \eqref{eq:lasso1}, we we have $\|\Z^T(y - \Z\hat\aalpha)\|_\infty \leq \lambda$ and
\[
\|\Z^T\Z h\|_\infty \leq \|\Z^T(y-\Z\hat\aalpha)\|_\infty + \|\Z^T(y-\Z\aalpha^*)\|_\infty\leq \lambda+ \|\Z^T \bm{\epsilon}\|_{\infty}.
\]
Lemma 5.1 in \citet{cai_compressed_2013} with $\lambda = \max(\|h_{-\max(s)}\|_\infty, \|h_{-\max(s)}\|_1/s)$ implies
\begin{align*}
|\langle \Z h_{\max(s)}, \Z h_{-\max(s)}\rangle | & \leq \theta_{s,s}(\Z) \|h_{\max(s)}\|_2\cdot \sqrt{s} \cdot \max(\|h_{-\max(s)}\|_\infty, \|h_{-\max(s)}\|_1/s) \\
& \leq \sqrt{s}\theta_{s,s}(\Z) \|h_{\max(s)}\|_2\cdot\frac{r+1}{r-1}\|h_{\max(s)}\|_1/s \\
& \leq \theta_{s,s}(\Z) \frac{r+1}{r-1}\|h_{\max(s)}\|_2^2,
\end{align*}
where the last inequality uses \eqref{eq:lasso2}. We then have
\begin{align*}
\sqrt{s}(\lambda+\|\Z^T \bm{\epsilon} \|_{\infty} ) \|h_{\max(s)}\|_2 & \geq (\lambda+\|\Z^T \bm{\epsilon} \|_{\infty} )\|h_{\max(s)}\|_1
  \geq \langle \Z^T\Z h, h_{\max(s)}\rangle \\
 & = \langle \Z h_{\max(s)}, \Z h_{\max(s)}\rangle + \langle \Z h_{\max(s)}, \Z h_{-\max(s)}\rangle \\
 & \geq \|\Z h_{\max(s)}\|_2^2 - \theta_{s,s} \frac{r+1}{r-1}\|h_{\max(s)}\|_2^2 \\
 & = \left(\delta_{2s}^-(\Z) - \theta_{s,s}(\Z)\frac{r+1}{r-1}\right)\|h_{\max(s)}\|_2^2\\
 & \geq \left(\frac{3r-1}{2(r-1)}\delta_{2s}^-(\Z) - \frac{r+1}{2(r-1)}\delta_{2s}^+\right)\|h_{\max(s)}\|_2^2,
\end{align*}
where the last inequality uses Lemma \ref{lm:RIP-ROC}. Moving $\|h_{\max(s)}\|$ to the right hand side and using the condition $r\|\Z^T \bm{\epsilon}\|_{\infty} \leq \lambda$ where $r > 1$ yields \eqref{eq:lasso3}.
\end{proof}

Now we move on to the proof of Theorem 2. Section 3.5 in the  main paper states that the original estimation method can be reinterpreted as a two-step method where the first step is the Lasso step and the second step is a dot product. The proof will first analyze step 1 using the lemmas about Lasso performance and use it to analyze step 2.

First, in lieu of step 1, the model in equation (3) from the original paper can be modified to 
\begin{equation} \label{eq:model4}
\PP_{\hat{\D}^{\perp}} \PP_{\Z}Y = \PP_{\hat{\D}^{\perp}} \Z \aalpha^{*} + \PP_{\hat{\D}^{\perp}}\PP_{\Z} \bm{\epsilon}.
\end{equation}
Here, $\PP_{\hat{\D}^{\perp}} \Z$ becomes the design matrix, $\PP_{\hat{\D}^{\perp}} \PP_{\Z}Y$ becomes the outcome, and $\PP_{\hat{\D}^{\perp}}\PP_{\Z} \bm{\epsilon}$ is the new error term. In addition, from the condition $3\|\Z^T \PP_{\hat{\D}^\perp} \bm{\epsilon} \| \leq \lambda$, we have
\[
\lambda \geq  3 \|\Z^T (I- \PP_{\hat{\D}}) \bm{\epsilon} \|_{\infty} = 3 \|\Z^T (\PP_{\Z}- \PP_{\hat{\D}}) \bm{\epsilon} \|_{\infty} = 3 \|\Z^T (I- \PP_{\hat{\D}}) \PP_{\Z} \bm{\epsilon} \|_{\infty} = 3\| (\PP_{\hat{\D}^{\perp}} \Z)^T \PP_{\Z} \bm{\epsilon} \|_{\infty}.
\]

Second, note that \eqref{eq:bound1} is in terms of the RIP constants of $\PP_{\hat{\D}^\perp}\Z$. To relate the RIP constants of $\PP_{\hat{\D}^\perp}\Z$ with that of $\Z$, we see that for any $2s$-sparse vector $x \in \reals^L$,
$
\|\PP_{\hat{\D}^{\perp}} \Z x\|_{2}^2 = \|\Z x\|_2^2 - \| \PP_{\hat{\D}} \Z x\|_2^2 \leq \| \Z x \|_2^2 \leq \delta_{2s}^{+}(\Z) \| x \|_2^2.
$
By the definition of $\delta_{2s}^+(\PP_{\hat{\D}^{\perp}} \Z)$, this implies 
\begin{equation}\label{eq:ripconvert1}
\delta_{2s}^+(\PP_{\hat{\D}^{\perp}} \Z) \leq  \delta_{2s}^+(\Z).
\end{equation}
In addition, we have 
$
\|\PP_{\hat{\D}^{\perp}} \Z x\|_{2}^2 = \|\Z x\|_2^2 - \| \PP_{\hat{\D}} \Z x\|_2^2 \geq \delta_{2s}^{-}(\Z) \|x\|_2^2 - \delta_{2s}^{+}(\PP_{\hat{\D}} \Z) \|x\|_2^2.
$
By the definition of $\delta_{2s}^-(\PP_{\hat{\D}^{\perp}} \Z)$, this also implies
\begin{equation} \label{eq:ripconvert2}
\delta_{2s}^-(\PP_{\hat{\D}^{\perp}} \Z) \geq \delta_{2s}^{-}(\Z) - \delta_{2s}^+(\PP_{\hat{\D}} \Z).
\end{equation}

Combining \eqref{eq:ripconvert1}, \eqref{eq:ripconvert2} with assumption that $2\delta^-_{2s}(\Z) > \delta_{2s}^+(\Z) + 2\delta_{2s}^+(\PP_{\hat \D}\Z)$, we know $2\delta_{2s}^-(\PP_{\hat \D^\perp}\Z) > \delta_{2s}^-(\PP_{\hat \D^\perp}\Z)$. By Lemma \ref{lm:Lasso_RIP}, where we set $r=3$ in assumption $r\|\Z^T \bm{\epsilon}\|_{\infty} \leq \lambda$ and the model is rewritten as \eqref{eq:model4},
\begin{equation} \label{eq:bound1}
\|h_{\max(s)} \|_{2} \leq \frac{4/3 \lambda \sqrt{s}}{2\delta_{2s}^{-}(\PP_{\hat{\D}^\perp}\Z) - \delta_{2s}^+(\PP_{\hat{\D}^\perp}\Z)}
\end{equation}
and 
\begin{equation} \label{eq:bound1.5}
\|h_{-\max(s)} \|_1 \leq 2 \|h_{\max(s)} \|_1. 
\end{equation}
Combining the RIP relations established by \eqref{eq:ripconvert1} and \eqref{eq:ripconvert2}, we can rewrite \eqref{eq:bound1} as
\begin{equation} \label{eq:bound2}
\|h_{\max(s)} \|_{2} \leq \frac{4/3 \lambda \sqrt{s}}{2\delta_{2s}^{-}(\Z) - \delta_{2s}^+(\Z) - 2\delta_{2s}^+(\PP_{\hat \D}\Z)}.
\end{equation}

Third, we establish a bound for $\|\PP_{\hat{\D}} \Z h\|_2$. This bound is needed to bound step 2 in Section 3.5 of the original paper because
\[
\hat{\beta}_{\lambda} = \frac{\hat{\D}^T \PP_{\hat{\D}} (Y - \Z\hat{\aalpha}_{\lambda})}{\|\hat{\D}\|_2^2} = \frac{\hat{\D}^T \PP_{\hat{\D}} (\Z\aalpha^* + \D\beta^* + \bm{\epsilon} - \Z\hat{\aalpha}_{\lambda})}{\|\hat{\D}\|_2^2} = \beta^* - \frac{\hat{\D}^T \PP_{\hat{\D}} \Z h}{\|\hat{\D}\|_2^2}  + \frac{\hat{\D}^T \PP_{\hat{\D}} \bm{\epsilon}}{\|\hat{\D}\|_2^2}.
\]
Rearranging terms and taking norms on both sides give
\begin{equation} \label{eq:bound3}
\|\hat{\beta}_{\lambda} - \beta^{*}\|_2 \leq \frac{\| \hat{\D}^T \PP_{\hat{\D}} \Z h \|_2}{\|\hat{\D}\|_2^2}  + \frac{ \| \hat{\D}^T \PP_{\hat{\D}} \bm{\epsilon} \|_2}{\|\hat{\D}\|_2^2} \leq \frac{\| \PP_{\hat{\D}} \Z h \|_2}{\|\hat{\D}\|_2}  + \frac{| \hat{\D}^T \bm{\epsilon} |}{\|\hat{\D}\|_2^2}.
\end{equation}
Hence, a bound on $\| \PP_{\hat{\D}} \Z h \|_2 $ is necessary to bound $\|\hat{\beta}_{\lambda} - \beta^{*}\|_2$. To start off, we apply Lemma 1.1 in \citet{cai_sharp_2013} to represent $h_{-\max(s)}$ as a weighted mean of $s$-sparse vectors. This lemma allows us to convert the bound for $h_{\max(s)}$ in \eqref{eq:bound2} to the bound for $\| \PP_{\hat{\D}} \Z h \|_2$. Specifically, the lemma states we can find $\lambda_i \geq 0$ and $s$-sparse $v_i \in \reals^L$ where $i=1,\ldots,N$ such that $\sum_{i=1}^N \lambda_i = 1$ and $h_{-\max(s)} = \sum_{i=1}^N \lambda_i v_i$. Hence, $h =\sum_{i=1}^N \lambda_i(h_{\max(s)} + v_i)$. Furthermore, we have 
\[
supp(v_i) \subseteq supp(h_{-\max(s)}), \quad{} \|v_i\|_\infty \leq \max \left(\|h_{-\max(s)}\|_\infty, \frac{\|h_{-\max(s)}\|_1}{s} \right), \quad{} \|v_i\|_1 = \|h_{-\max(s)}\|_1,
\]
which yields
\[
\|v_i\|_\infty \leq \max \left( \frac{\|h_{\max(s)}\|_1}{s},\frac{2\|h_{\max(s)}\|_1}{s} \right)= \frac{2 \|h_{\max(s)} \|_1}{s}, \quad{} \|v_i \|_1\leq 2\| h_{\max(s)}\|_1 
\]
and
$
\|h_{\max(s)} + v_i\|_2^2 = \|h_{\max(s)}\|_2^2 + \|v_i\|_2^2 \leq \|h_{\max(s)}\|_2^2 + \|v_i \|_1 \|v_i\|_\infty  \leq 5\|h_{\max(s)}\|_2^2.
$
Combining all these together with \eqref{eq:bound2}, we have
\begin{align*}
\|\PP_{\hat{\D}} \Z h\|_2 &\leq \sum_{i=1}^N \lambda_i \|\PP_{\hat{\D}} \Z (h_{\max(s)} + v_i)\|_2 
\leq \sum_{i=1}^N \lambda_i \sqrt{5 \delta_{2s}^{+}(\PP_{\hat{\D}}\Z)} \|h_{\max(s)}\|_2 \\
&\leq \sqrt{5\delta_{2s}^+(\PP_{\hat \D}\Z) } \frac{4/3 \lambda \sqrt{s}}{2\delta_{2s}^{-}(\Z) - \delta_{2s}^+(\Z) - 2\delta^+_{2s}(\PP_{\hat \D}\Z)} \\
&=  \frac{4\sqrt{5}/3 \lambda \sqrt{s\delta^+_{2s}(\PP_{\hat \D}\Z)}}{2\delta_{2s}^{-}(\Z) - \delta_{2s}^+(\Z) - 2\delta^+_{2s}(\PP_{\hat \D}\Z)}.
\end{align*}
Finally, using the relation \eqref{eq:bound3} gives us the desired bound for Theorem 2. \qed

Of independent interest is that the proof of Theorem 2 can be generalized to a matrix of $\D$ instead of a vector of $\D$. That is, the proof can consider models where there are more than one endogenous variables in the data-generating model. However, for clarity of presentation, we don't explore this route.

\subsection{Proof of Corollary 2}

Now, we establish Corollary 2 as a Corollary to Theorem 2. Specifically, the task is to convert the RIP constants $\delta_{2s}^+(\Z)$, $\delta_{2s}^-(\Z)$, $\delta_{2s}^+(\PP_{\hat \D}\Z)$ and the constraint of $2\delta_{2s}^{-}(\Z) - \delta_{2s}^+(\Z) - 2\delta_{2s}^+(\PP_{\hat \D}\Z)>0 $ into $\mu$ and a similar constraint on $s$. To do this, note that for any $s$-sparse vector $\aalpha$
\begin{align*}
\|\Z \aalpha \|_2^2 &= \sum_{j \in supp(\aalpha)} \|\Z_{.j}\|_2^2 \aalpha_j^2 + \sum_{i < j, i,j \in supp(\aalpha)} 2\aalpha_i \aalpha_j \langle \Z_{.i}, \Z_{.j} \rangle 
\leq \sum_{j \in supp(\aalpha)} \aalpha_j^2 + \sum_{i < j, i,j \in supp(\aalpha)} (\aalpha_i^2 + \aalpha_j^2) \mu \\
&= (1 + (s - 1)\mu) \sum_{j \in supp(\aalpha)} \aalpha_j^2 
= (1 + (s -1)\mu) \|\aalpha\|_2^2
\end{align*}
and 
\begin{align*}
\|\Z \aalpha\|_2^2 &= \sum_{j \in supp(\aalpha)} \|\Z_{.j}\|_2^2 \aalpha_j^2 + \sum_{i < j, i,j \in supp(\aalpha)} 2\aalpha_i \aalpha_j \langle \Z_{.i}, \Z_{.j} \rangle 
\geq \sum_{j \in supp(\aalpha)} \aalpha_j^2 - \sum_{i < j, i,j \in supp(\aalpha)} (\aalpha_i^2 + \aalpha_j^2) \mu  \\
&= (1 - (s-1) \mu)  \|\aalpha\|_2^2.
\end{align*}
The upper and lower bounds on $\|\Z\aalpha\|_2^2$ imply 
\[
\delta_{s}^{+}(\Z) \leq (1 + (s - 1) \mu), \quad \mbox{and}\quad
\delta_{s}^{-}(\Z) \geq (1 - (s-1) \mu);
\]
For $\PP_{\hat{\D}^\perp}\Z$ and all $2s$-sparse vector $x$, we have
\begin{align*}
 \|\PP_{\hat{\D}} \Z x\|_2^2 &\leq \left(\sum_{j \in supp(x)} \| \PP_{\hat{\D}} \Z_{.j} x_j\|_2 \right)^2 
\leq 2s \sum_{j \in supp(x)} \|\PP_{\hat{\D}} \Z_{.j} x_j \|_2^2 \\
&= 2s \sum_{j \in supp(x)} \|\PP_{\hat{\D}} \Z_{.j}\|_2^2 x_j^2 
= 2s \sum_{j \in supp(x)} \frac{\| \PP_{\hat{\D}} \Z_{.j} \|_2^2}{\|\Z_{.j}\|_2^2} \|\Z_{.j} x_{j}\|_2^2 \\
&\leq  2s \rho^2  \delta_{1}^{+}(\Z) \sum_{j\in supp(x)}x_j^2 \leq 2s \rho^2  \delta_{2s}^{+}(\Z) \|x\|_2^2.
\end{align*}
Again, by the definition of $\delta_{2s}^+(\PP_{\hat{\D}}\Z)$, this implies that 
\begin{equation} \label{eq:ripconvert3}
\delta_{2s}^+(\PP_{\hat{\D}}\Z) \leq 2s \rho^2  \delta_{2s}^{+}(\Z). 
\end{equation}

Under the condition $s < \min\left(\frac{1}{12 \mu}, \frac{1}{10 \rho^2}\right)$, the denominator of the bound in Theorem 2 becomes
\begin{align*}
2 \delta_{2s}^{-}(\Z) - \delta_{2s}^{+}(\Z) - 2\delta_{2s}^+(\PP_{\hat \D}\Z) & \geq 2\delta_{2s}^-(\Z) - (1 + 4s\rho^2)\delta_{2s}^+(\Z) \\
& \geq 2(1 - (2s - 1)\mu) - (1 + 4s \rho^2)(1 + (2s -1) \mu) \\
&= 1 - 6s \mu + 3\mu - 4s\rho^2 - 8s^2 \rho^2 \mu + 4s \rho^2 \mu \\
%&\geq 1 - 6s \mu - 4s \rho^2 - 8 s \rho^2 (s \mu) \\
&\geq 1 - 6s \mu - 5s \rho^2 > 0.
\end{align*}
For the numerator of the bound in Theorem 2, we have
\begin{align*}
\frac{4\sqrt{5}}{3} \lambda \sqrt{s\delta_{2s}^+(\PP_{\hat \D}\Z)} &\leq \frac{4\sqrt{5}}{3} \lambda \sqrt{2s^2\rho^2\delta_{2s}^+(\Z)} \leq \frac{4\sqrt{10}}{3} \lambda s \rho \sqrt{1 + (2s-1) \mu} \\
& \leq \frac{4\sqrt{10}}{3} \lambda s \rho \sqrt{1 + 2s \mu} \leq \frac{4\sqrt{10}}{3} \lambda s \rho \sqrt{1 + 1/6} 
=\frac{4 \sqrt{105}}{9} \lambda s \rho.
\end{align*}
Combining them together leads to the desired bound. Note that one can improve the constants in the constraint of $s$ with a bit more care on the above inequalities.
\qed

\subsection{Proof of Theorem 3}

The original estimation method can be rewritten as follows
\begin{align*}
\hat{\bm{\alpha}}_{\lambda}, \hat{\beta}_{\lambda} = &\amin{\bm{\alpha},\beta} \frac{1}{2} \|\mathbf{P}_{\mathbf{Z}}(\mathbf{Y} - \mathbf{Z}\bm{\alpha} - \mathbf{D}\beta)\|_2^2 + \lambda ||\bm{\alpha}||_1 \\
=&\amin{\bm{\alpha},\beta} \frac{1}{2} ||(\mathbf{P}_{\hat{\mathbf{D}}} + \mathbf{P}_{\hat{\mathbf{D}}^{\perp}}) \mathbf{P}_{\mathbf{Z}}(\mathbf{Y} - \mathbf{Z}\bm{\alpha} - \mathbf{D}\beta)||_2^2 +  \lambda ||\bm{\alpha}||_1 \\
=& \amin{\bm{\alpha},\beta} \frac{1}{2}||\mathbf{P}_{\hat{\mathbf{D}}} \mathbf{P}_{\mathbf{Z}}(\mathbf{Y} - \mathbf{Z}\bm{\alpha} - \mathbf{D}\beta)||_2^2  + \frac{1}{2} ||\mathbf{P}_{\hat{\mathbf{D}}^{\perp}}\mathbf{P}_{\mathbf{Z}} (\mathbf{Y} - \mathbf{Z}\bm{\alpha} - \mathbf{D}\beta)||_2^2 + \lambda ||\bm{\alpha}||_1 \\
=& \amin{\bm{\alpha},\beta} \frac{1}{2}||\mathbf{P}_{\hat{\mathbf{D}}} (\mathbf{Y} - \mathbf{Z}\bm{\alpha}) - \hat{\mathbf{D}}\beta||_2^2 + \frac{1}{2}||\mathbf{P}_{\hat{\mathbf{D}}^{\perp}}\mathbf{P
}_{\mathbf{Z}} \mathbf{Y} - \mathbf{P}_{\hat{\mathbf{D}}^{\perp}}\mathbf{Z}\bm{\alpha} ||_2^2 + \lambda ||\bm{\alpha}||_1.
\end{align*}
The first term, $\frac{1}{2}||\mathbf{P}_{\hat{\mathbf{D}}} (\mathbf{Y} - \mathbf{Z}\bm{\alpha}) - \hat{\mathbf{D}}\beta||_2^2 $ is always zero for any given $\bm{\alpha} \in \reals^L$ because $\mathbf{P}_{\hat{\mathbf{D}}}(\mathbf{Y} - \mathbf{Z}\bm{\alpha})$ lies in the span of $\hat{\mathbf{D}}$ and thus, we can pick $\beta$ such that the first term is zero. The second term, $\frac{1}{2}||\mathbf{P}_{\hat{\mathbf{D}}^{\perp}}\mathbf{P}_{\mathbf{Z}} (\mathbf{Y} - \mathbf{Z}\bm{\alpha} )||_2^2 + \lambda ||\bm{\alpha}||_1$, is the traditional Lasso problem where the outcome is $\mathbf{P}_{\hat{\mathbf{D}}^\perp} \mathbf{P}_{\mathbf{Z}}\mathbf{Y}$ and the design matrix is $\mathbf{P}_{\hat{\mathbf{D}}^{\perp}} \mathbf{Z}$. Hence, the minimizer for this Lasso problem is also the minimizer for the original method.
\qed

\section{Figures}
\begin{figure}[htbp!]
\centering
\includegraphics[width=7in,height=7.3in]{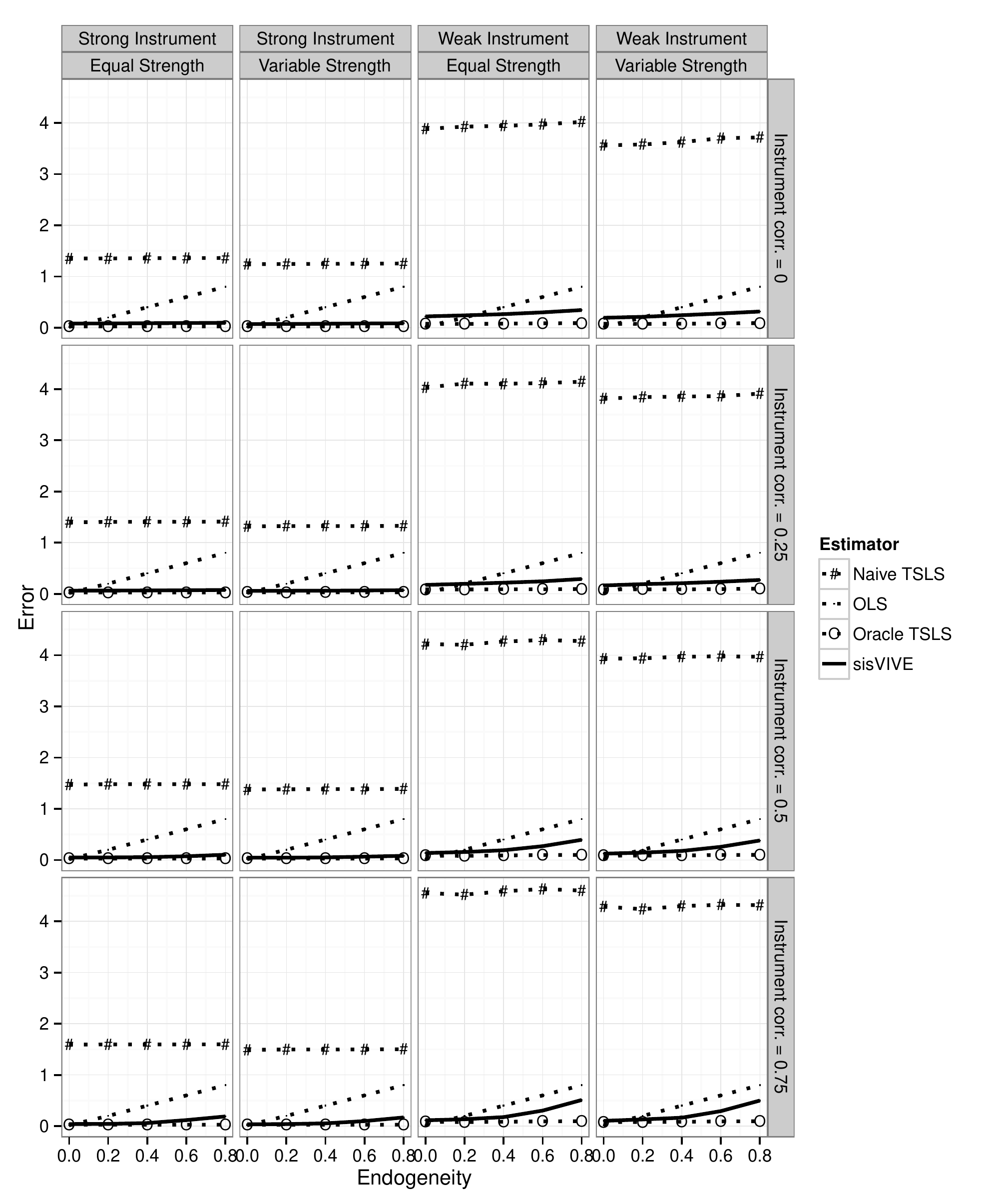}
\caption{Simulation Study of Estimation Performance Varying Endogeneity and Correlation Only Exists Within Valid and Invalid Instruments. There are ten $(L = 10)$ instruments. Each line represents the median absolute estimation error ($|\beta^* - \hat{\beta}|$) after 500 simulations. We fix the number of invalid instruments to $s = 3$. Each column in the plot corresponds to different variation of instruments' absolute and relative strength. There are two types of absolute strengths, ``Strong'' and ``Weak'', measured by the concentration parameter. There are two types of relative strengths, ``Equal'' and ``Variable'', measured by varying $\bm{\gamma}^*$ while holding the absolute strength (i.e. concentration parameter) fixed. Each row corresponds to the maximum correlation between instruments, but  correlation only exists within valid and invalid instruments.}
\label{fig:withinZcorrEndo}
\end{figure}

\begin{figure}[htbp!]
\centering
\includegraphics[width=7in,height=7.3in]{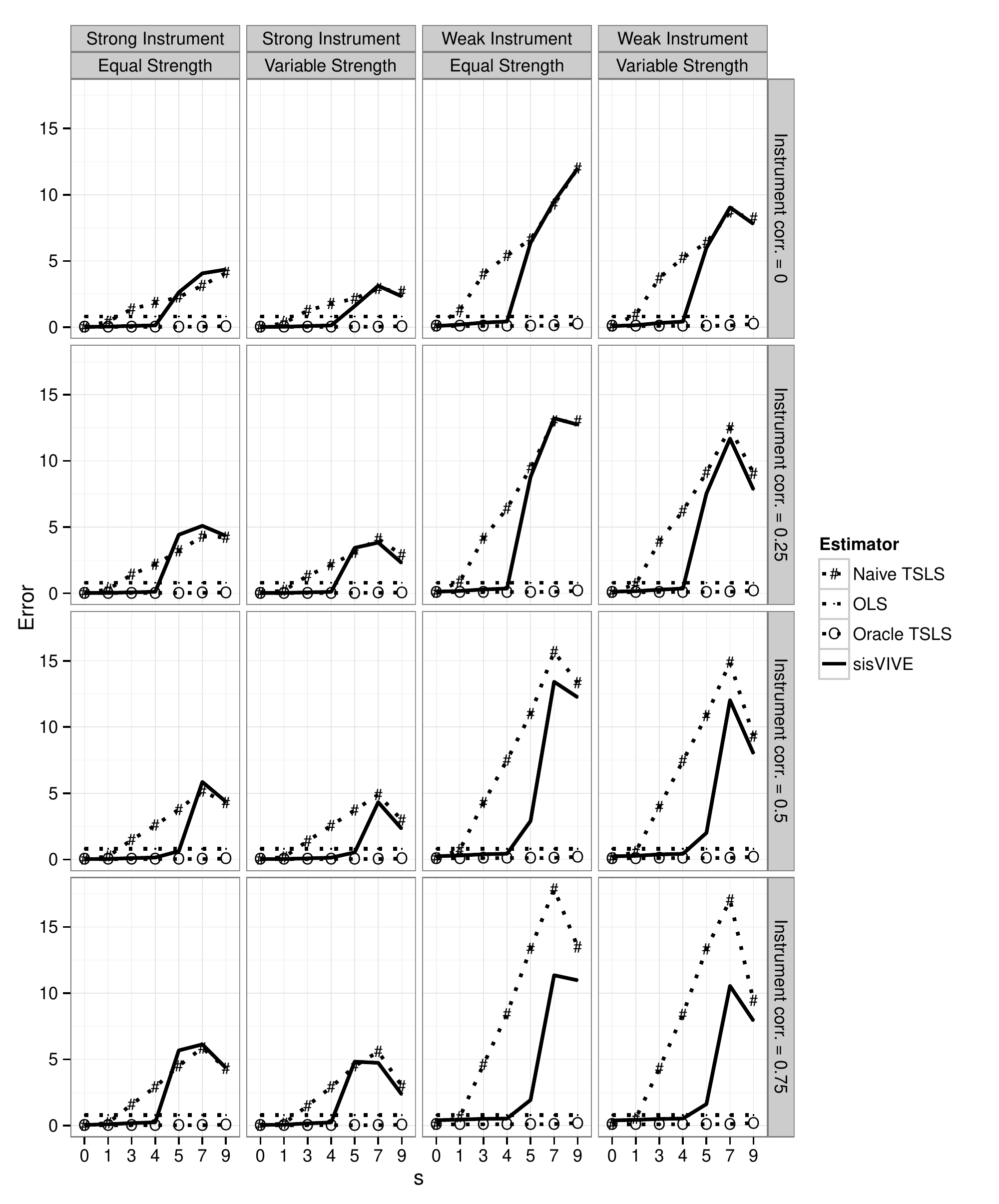}
\caption{Simulation Study of Estimation Performance Varying the Number of Invalid Instruments ($s$) and Correlation Only Exists Within Valid and and Invalid Instruments. There are ten $(L = 10)$ instruments. Each line represents the median absolute estimation error ($|\beta^* - \hat{\beta}|$) after 500 simulations. We fix the endogeneity $\sigma_{\epsilon \xi}^*$ to $\sigma_{\epsilon \xi}^* = 0.8$. Each column in the plot corresponds to a different variation of instruments' absolute and relative strength. There are two types of absolute strengths, ``Strong'' and ``Weak'', measured by the concentration parameter. There are two types of relative strengths, ``Equal'' and ``Variable'', measured by varying $\bm{\gamma}^*$ while holding the absolute strength fixed. Each row corresponds to the maximum correlation between instruments, but correlation only exists within valid and invalid instruments.}
\label{fig:withinZcorrS}
\end{figure}

\begin{figure}[htbp!]
\centering
\includegraphics[width=7in,height=7.3in]{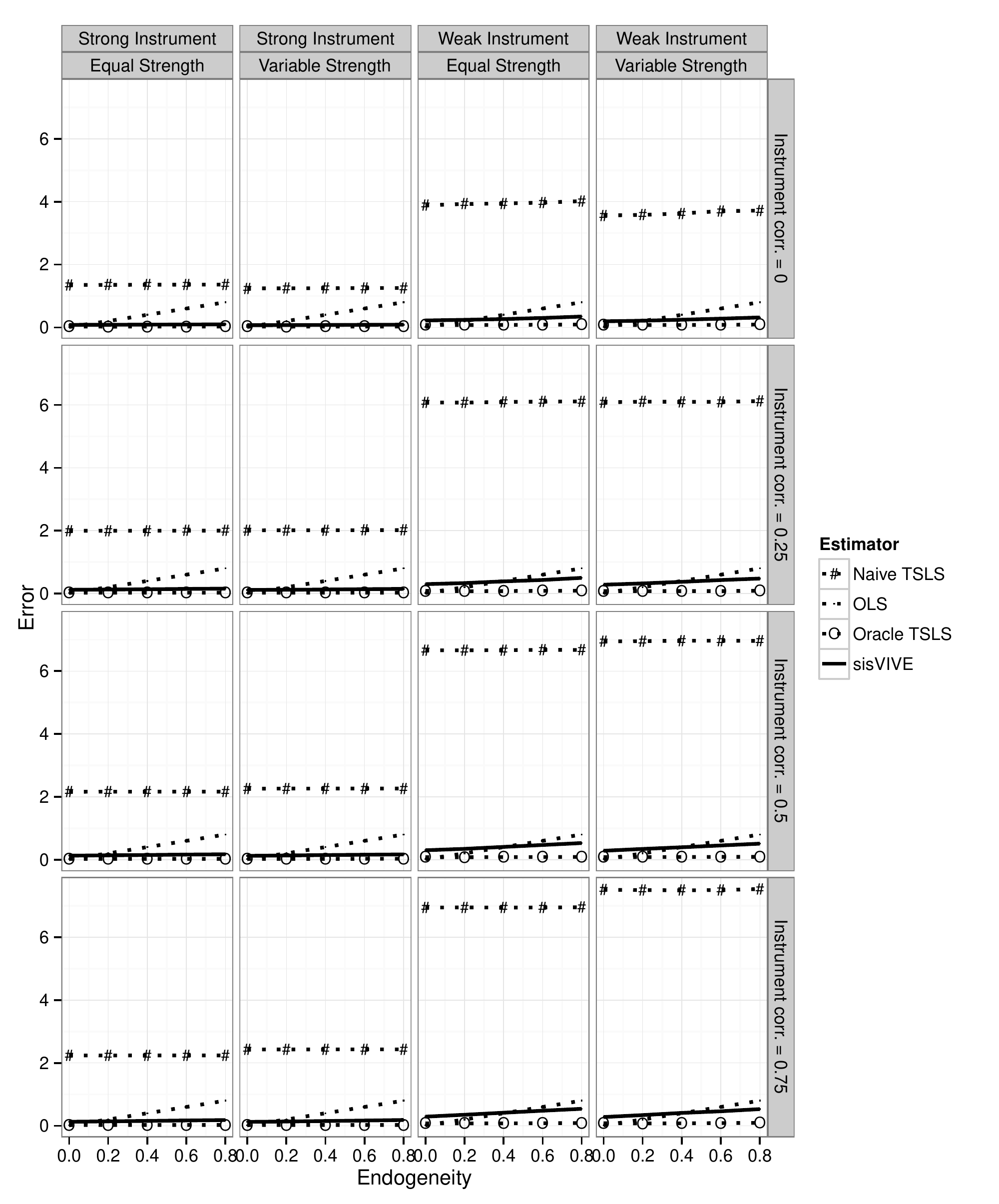}
\caption{Simulation Study of Estimation Performance Varying Endogeneity and Correlation Only Exists Between Valid and Invalid Instruments. There are ten $(L = 10)$ instruments. Each line represents the median absolute estimation error ($|\beta^* - \hat{\beta}|$) after 500 simulations. We fix the number of invalid instruments to $s = 3$. Each column in the plot corresponds to a different variation of instruments' absolute and relative strength. There are two types of absolute strengths, ``Strong'' and ``Weak'', measured by the concentration parameter. There are two types of relative strengths, ``Equal'' and ``Variable'', measured by varying $\bm{\gamma}^*$ while holding the absolute strength (i.e. concentration parameter) fixed. Each row corresponds to the maximum correlation between instruments, but correlation only exists between valid and invalid instruments.}
\label{fig:interZcorrEndo}
\end{figure}

\begin{figure}[htbp!]
\centering
\includegraphics[width=7in,height=7.3in]{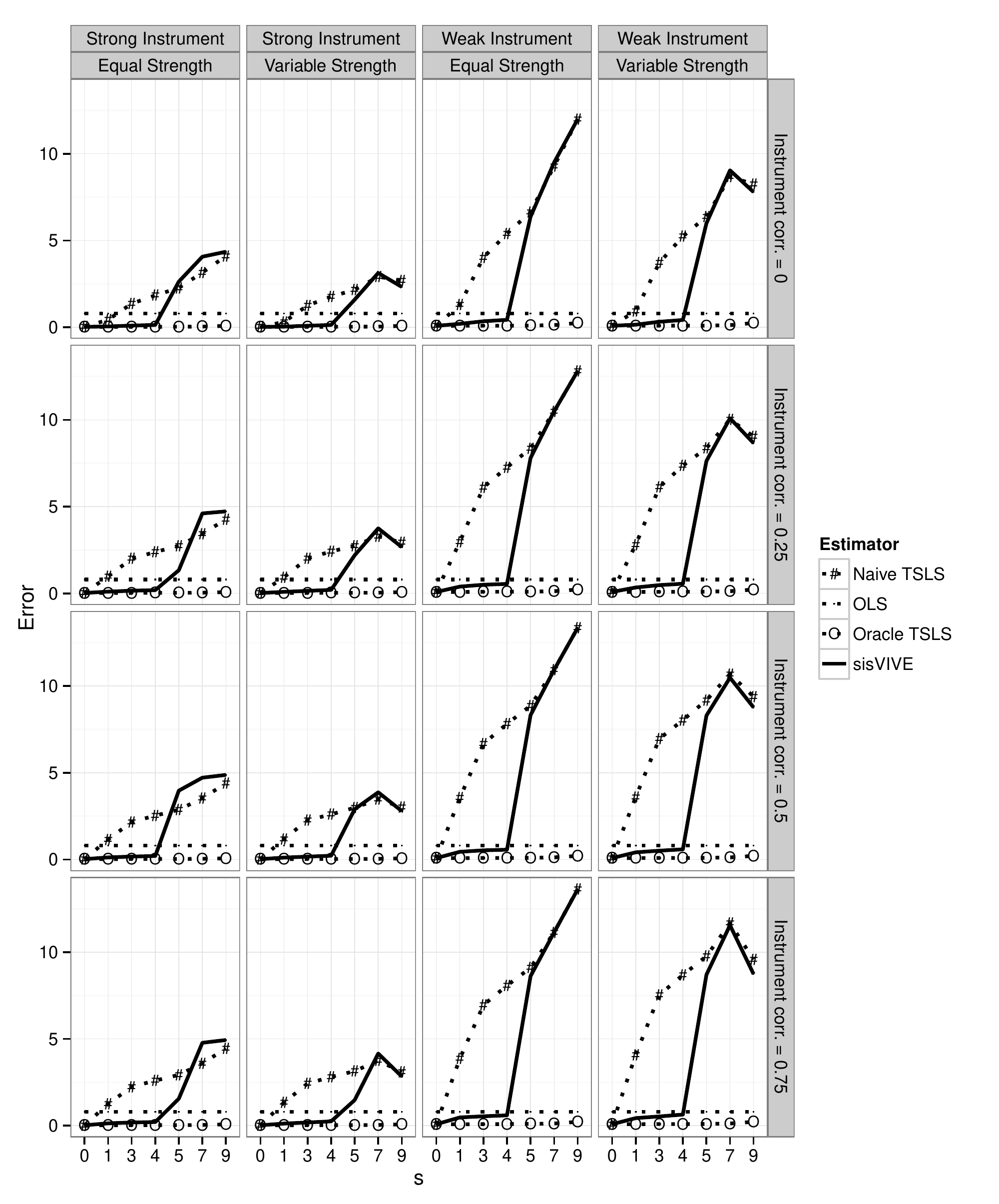}
\caption{Simulation Study of Estimation Performance Varying the Number of Invalid Instruments ($s$) and Correlation Only Exists Between Valid and and Invalid Instruments. There are ten $(L = 10)$ instruments. Each line represents the median absolute estimation error ($|\beta^* - \hat{\beta}|$) after 500 simulations. We fix the endogeneity $\sigma_{\epsilon \xi}^*$ to $\sigma_{\epsilon \xi}^* = 0.8$. Each column in the plot corresponds to a different variation of instruments' absolute and relative strength. There are two types of absolute strengths, ``Strong'' and ``Weak'', measured by the concentration parameter. There are two types of relative strengths, ``Equal'' and ``Variable'', measured by varying $\bm{\gamma}^*$ while holding the absolute strength fixed. Each row corresponds to the maximum correlation between instruments, but correlation only exists between valid and invalid instruments.}
\label{fig:interZcorrS}
\end{figure}

\begin{figure}[htbp!]
\centering
\includegraphics[width=7in,height=7.3in]{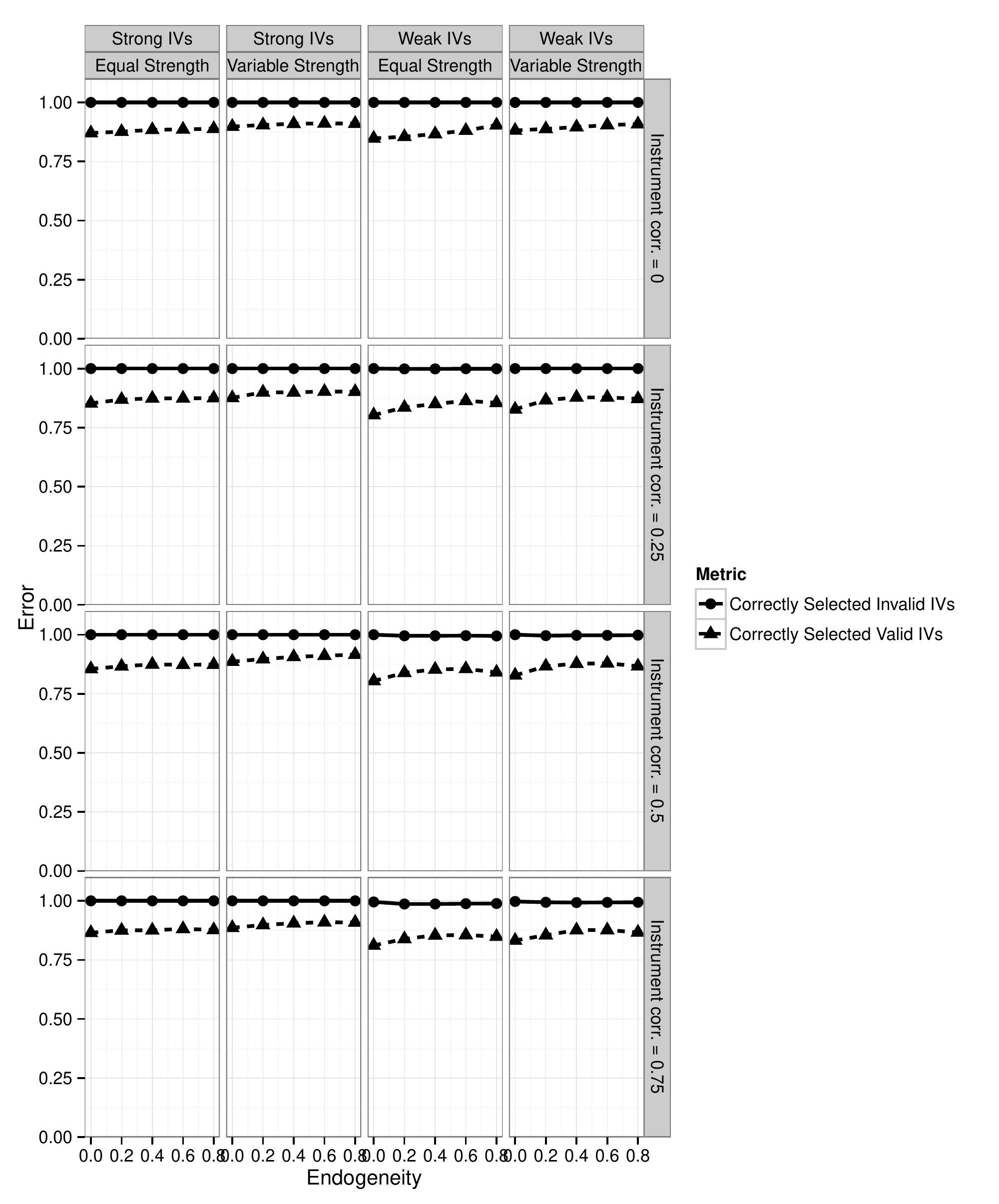}
\caption{Simulation Study Varying Endogeneity and Correlation Exists Between All Instruments. There are ten $(L = 10)$ instruments. Each line represents the average proportions of correctly selected valid instruments and correctly selected invalid instruments after 500 simulations. We fix the number of invalid instruments to $s = 3$. Each column in the plot corresponds to a different variation of instruments' absolute and relative strength. There are two types of absolute strengths, ``Strong'' and ``Weak'', measured by the concentration parameter. There are two types of relative strengths, ``Equal'' and ``Variable'', measured by varying $\bm{\gamma}^*$ while holding the absolute strength (i.e. concentration parameter) fixed. Each row corresponds to the maximum correlation between all instruments.}
\label{fig:equalZcorrEndoPercent}
\end{figure}

\begin{figure}[htbp!]
\centering
\includegraphics[width=7in,height=7.3in]{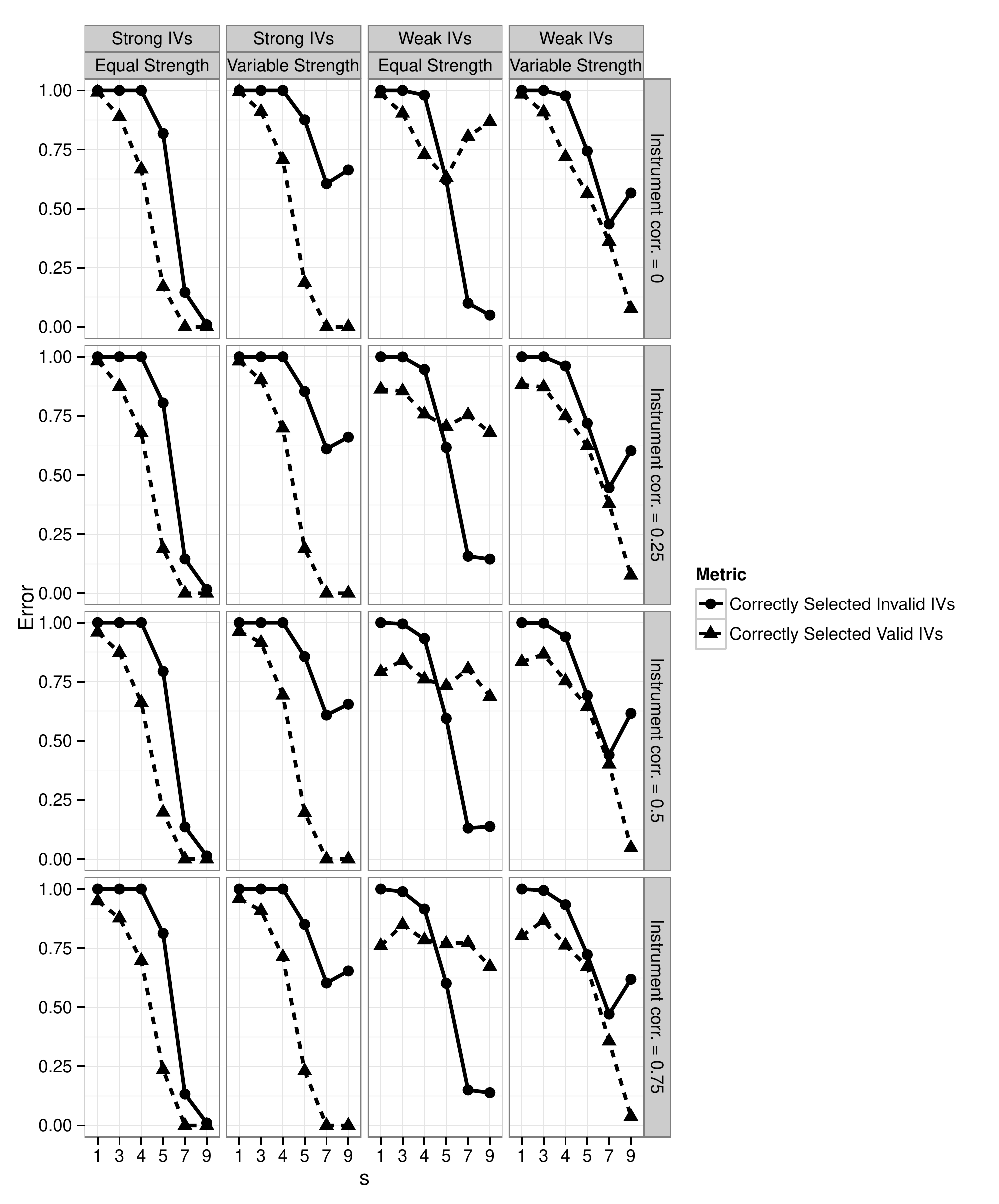}
\caption{Simulation Study Varying the Number of Invalid Instruments ($s$) and Correlation Exists Between All Instruments. There are ten $(L = 10)$ instruments. Each line represents the average proportions of correctly selected valid instruments and correctly selected invalid instruments after 500 simulations. We fix the endogeneity $\sigma_{\epsilon \xi}^*$ to $\sigma_{\epsilon \xi}^* = 0.8$. Each column in the plot corresponds to a different variation of instruments' absolute and relative strength. There are two types of absolute strengths, ``Strong'' and ``Weak'', measured by the concentration parameter. There are two types of relative strengths, ``Equal'' and ``Variable'', measured by varying $\bm{\gamma}^*$ while holding the absolute strength fixed. Each row corresponds to the maximum correlation between all instruments. }
\label{fig:equalZcorrSPercent}
\end{figure}

\begin{figure}[htbp!]
\centering
\includegraphics[width=7in,height=7.3in]{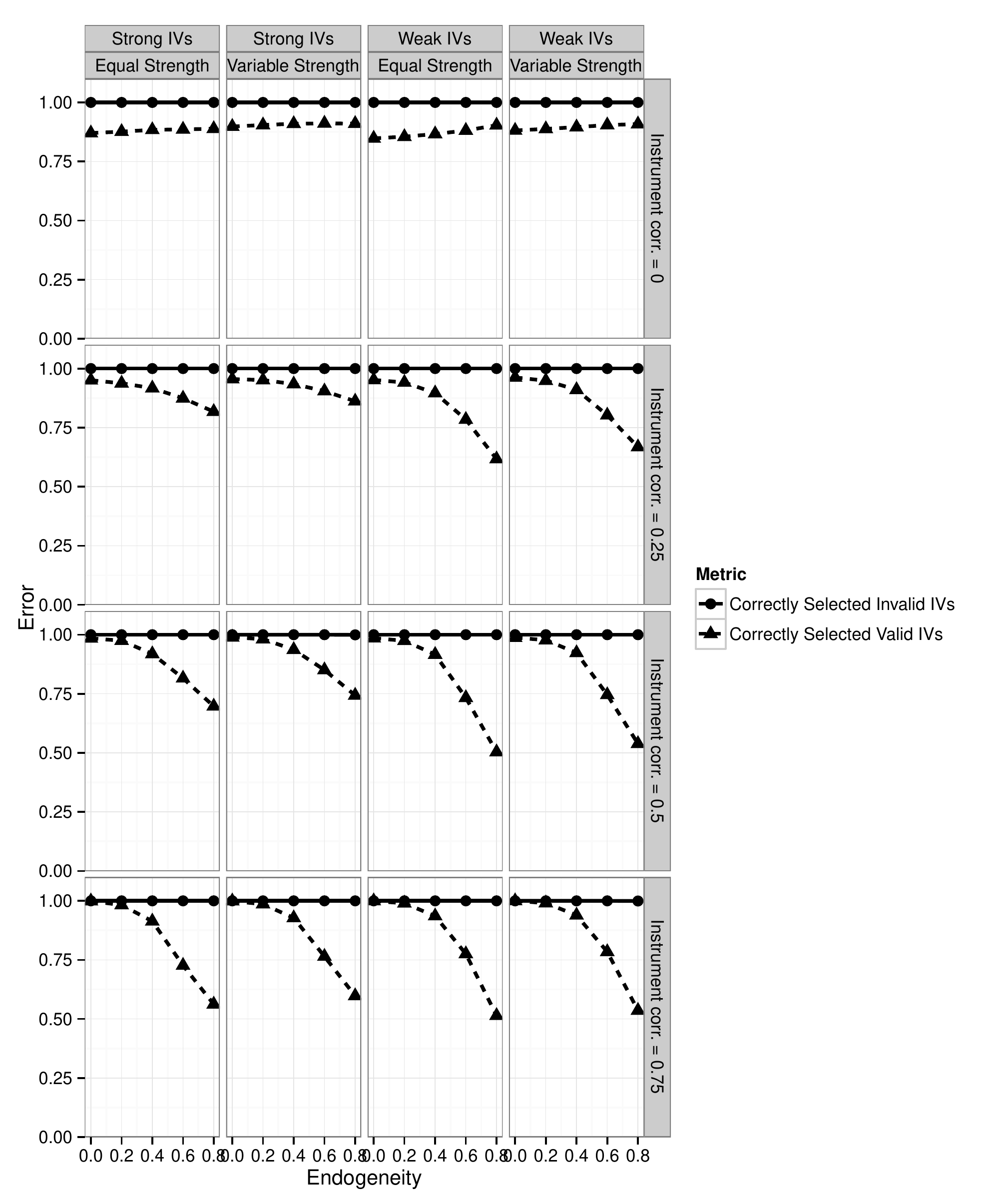}
\caption{Simulation Study Varying Endogeneity and Correlation Only Exists Within Valid and Invalid Instruments. There are ten $(L = 10)$ instruments. Each line represents the average proportions of correctly selected valid instruments and correctly selected invalid instruments after 500 simulations. We fix the number of invalid instruments to $s = 3$. Each column in the plot corresponds to a different variation of instruments' absolute and relative strength. There are two types of absolute strengths, ``Strong'' and ``Weak'', measured by the concentration parameter. There are two types of relative strengths, ``Equal'' and ``Variable'', measured by varying $\bm{\gamma}^*$ while holding the absolute strength (i.e. concentration parameter) fixed. Each row corresponds to the maximum correlation between instruments, but correlation only exists within valid and invalid instruments.}
\label{fig:withinZcorrEndoPercent}
\end{figure}

\begin{figure}[htbp!]
\centering
\includegraphics[width=7in,height=7.3in]{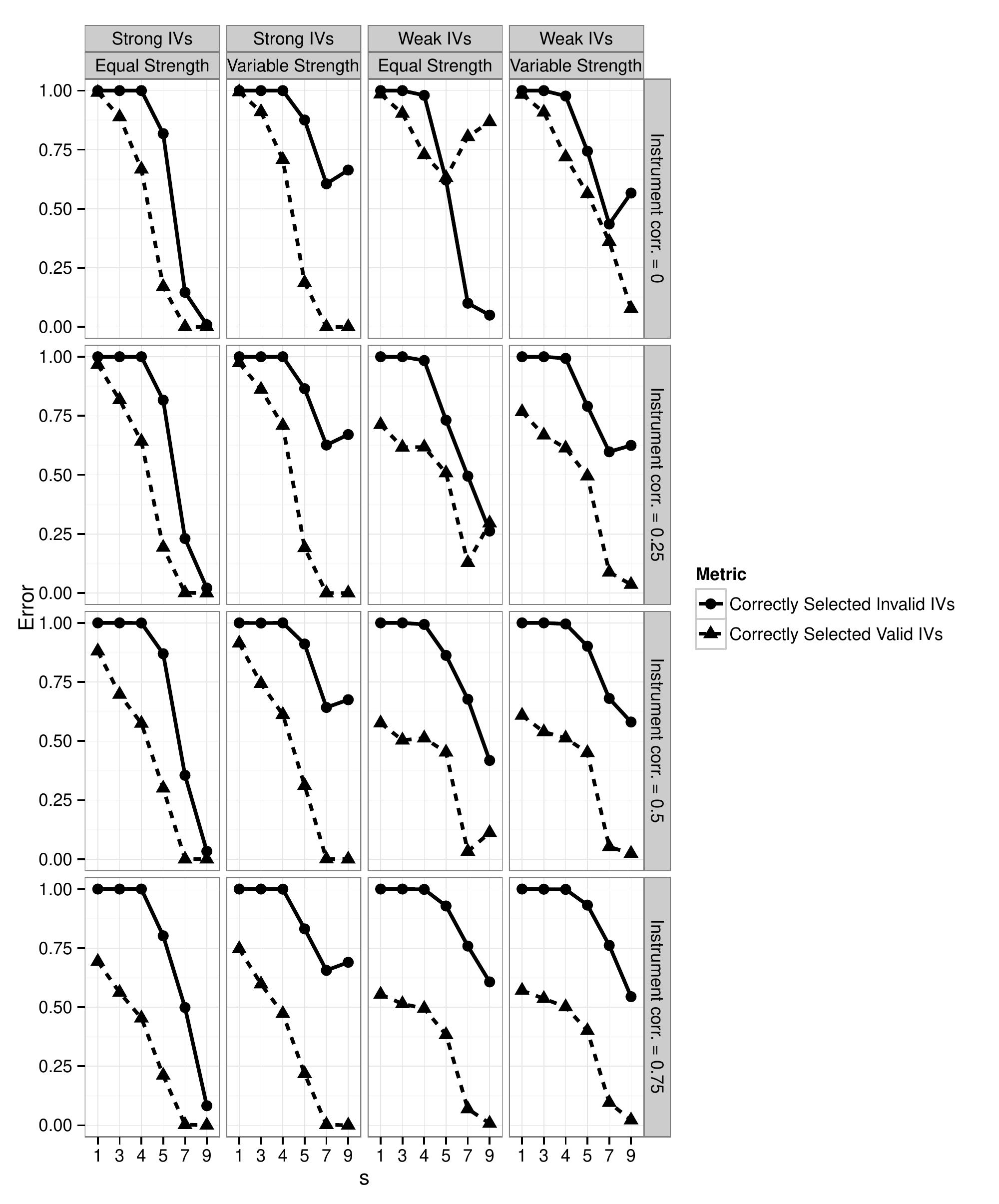}
\caption{Simulation Study Varying the Number of Invalid Instruments ($s$) and Correlation Only Exists Within Valid and and Invalid Instruments. There are ten $(L = 10)$ instruments. Each line represents the average proportions of correctly selected valid instruments and correctly selected invalid instruments after 500 simulations. We fix the endogeneity $\sigma_{\epsilon \xi}^*$ to $\sigma_{\epsilon \xi}^* = 0.8$. Each column in the plot corresponds to a different variation of instruments' absolute and relative strength. There are two types of absolute strengths, ``Strong'' and ``Weak'', measured by the concentration parameter. There are two types of relative strengths, ``Equal'' and ``Variable'', measured by varying $\bm{\gamma}^*$ while holding the absolute strength fixed. Each row corresponds to the maximum correlation between instruments, but correlation only exists within valid and invalid instruments.}
\label{fig:withinZcorrSPercent}
\end{figure}

\begin{figure}[htbp!]
\centering
\includegraphics[width=7in,height=7.3in]{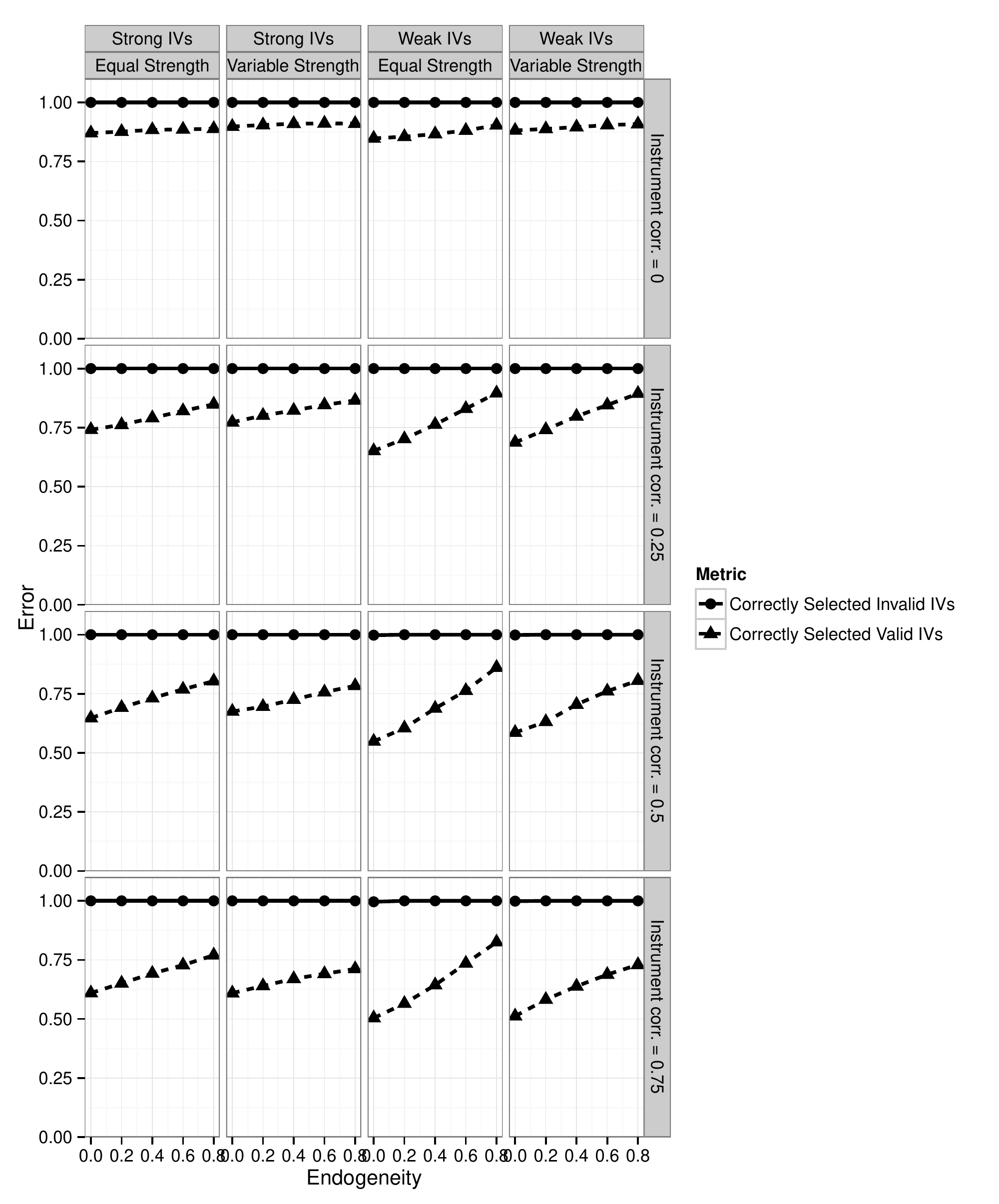}
\caption{Simulation Study Varying Endogeneity and Correlation Only Exists Between Valid and Invalid Instruments. There are ten $(L = 10)$ instruments. Each line represents the average proportions of correctly selected valid instruments and correctly selected invalid instruments after 500 simulations. We fix the number of invalid instruments to $s = 3$. Each column in the plot corresponds to a different variation of instruments' absolute and relative strength. There are two types of absolute strengths, ``Strong'' and ``Weak'', measured by the concentration parameter. There are two types of relative strengths, ``Equal'' and ``Variable'', measured by varying $\bm{\gamma}^*$ while holding the absolute strength (i.e. concentration parameter) fixed. Each row corresponds to the maximum correlation between instruments, but correlation only exists between valid and invalid instruments.}
\label{fig:interZcorrEndoPercent}
\end{figure}

\begin{figure}[htbp!]
\centering
\includegraphics[width=7in,height=7.3in]{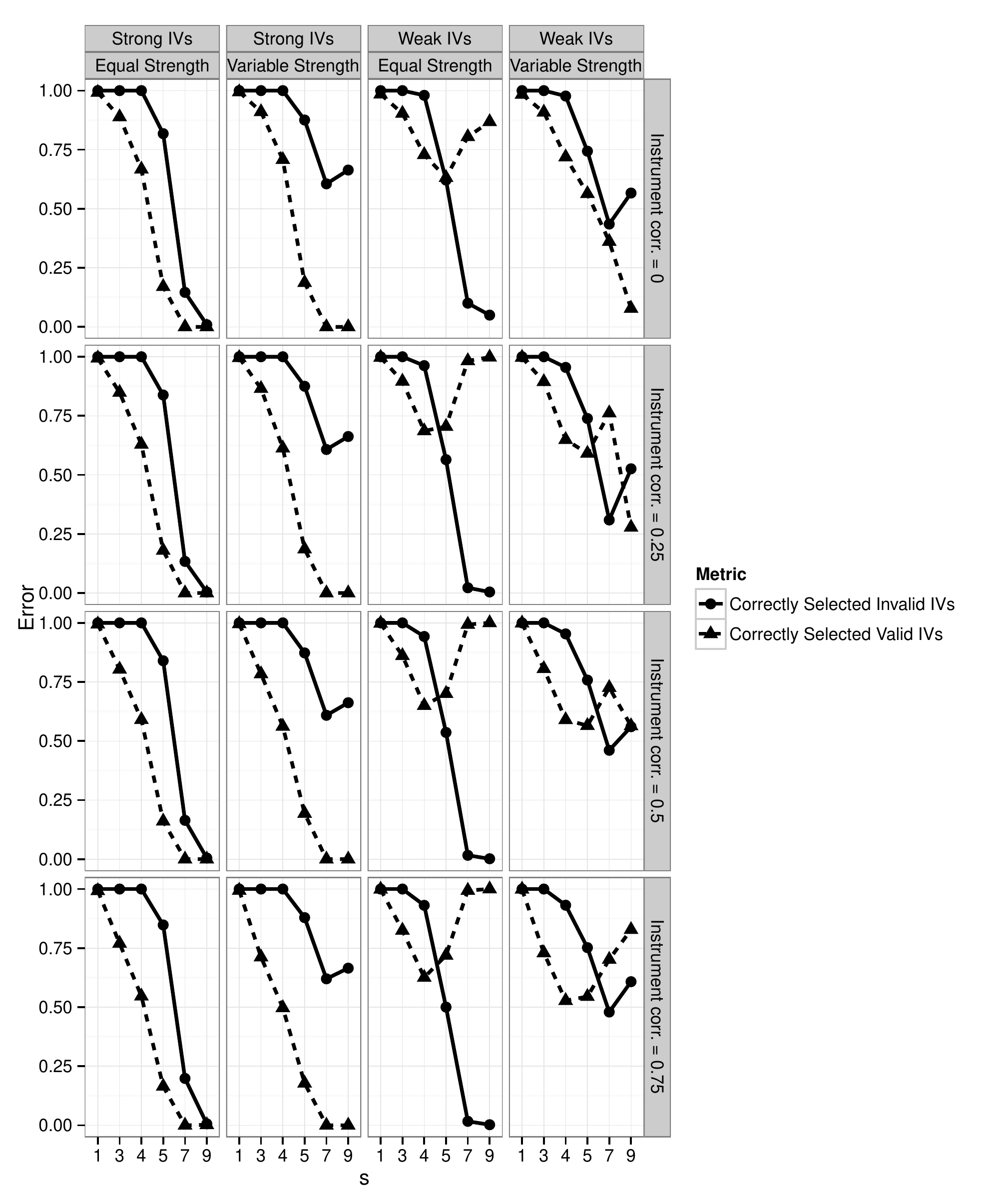}
\caption{Simulation Study Varying the Number of Invalid Instruments ($s$) and Correlation Only Exists Between Valid and and Invalid Instruments. There are ten $(L = 10)$ instruments. Each line represents the average proportions of correctly selected valid instruments and correctly selected invalid instruments after 500 simulations. We fix the endogeneity $\sigma_{\epsilon \xi}^*$ to $\sigma_{\epsilon \xi}^* = 0.8$. Each column in the plot corresponds to a different variation of instruments' absolute and relative strength. There are two types of absolute strengths, ``Strong'' and ``Weak'', measured by the concentration parameter. There are two types of relative strengths, ``Equal'' and ``Variable'', measured by varying $\bm{\gamma}^*$ while holding the absolute strength fixed. Each row corresponds to the maximum correlation between instruments, but correlation only exists between valid and invalid instruments.}
\label{fig:interZcorrSPercent}
\end{figure}

\begin{figure}[htbp!]
\centering
\includegraphics[width=7in,height=7.3in]{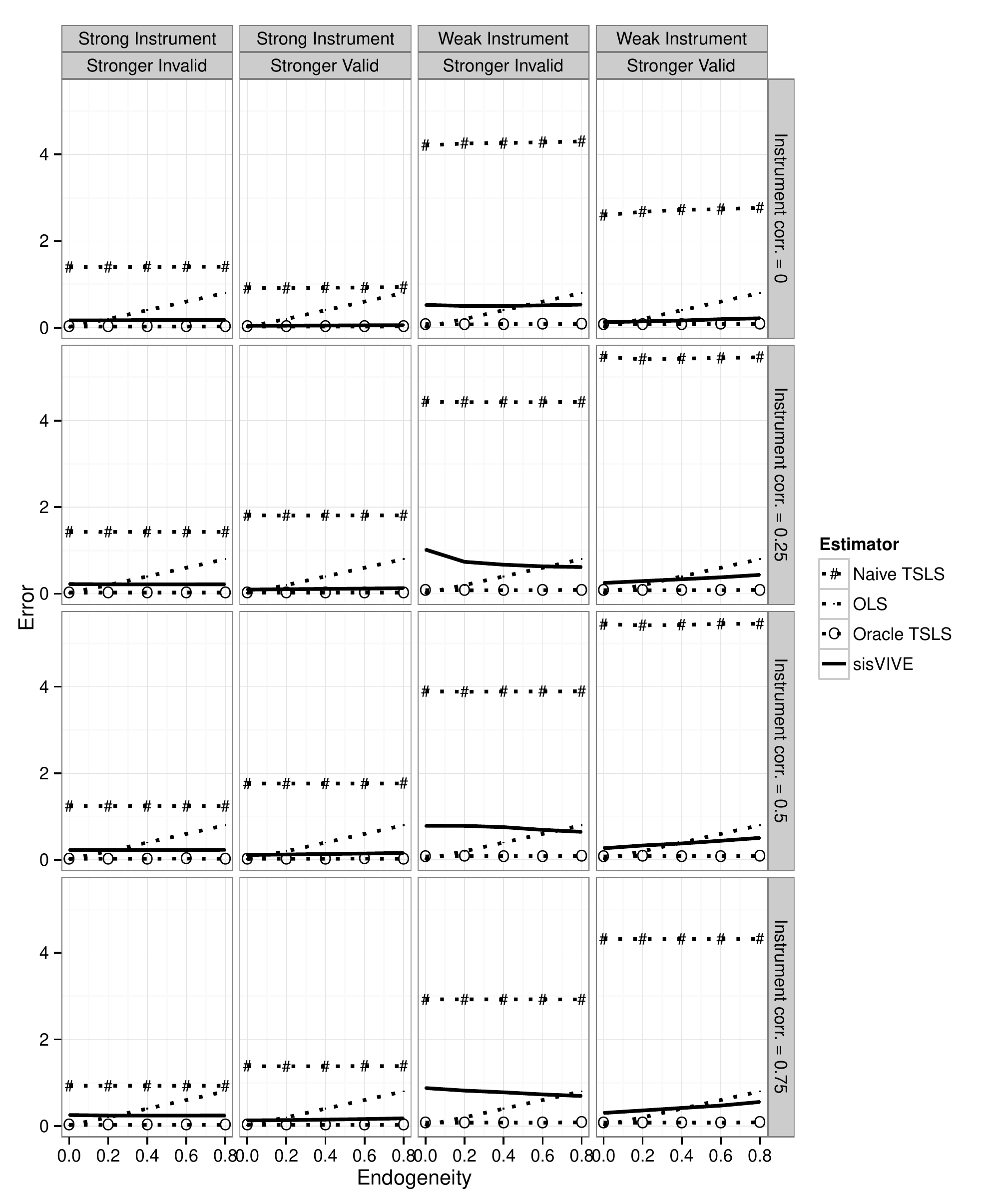}
\caption{Simulation Study Varying Endogeneity and Correlation Exists Between All Instruments. We also vary the instrument strength of valid and invalid instruments. There are ten $(L = 10)$ instruments. Each line represents the median absolute estimation error ($|\beta^* - \hat{\beta}|$) after 500 simulations. We fix the number of invalid instruments to $s = 3$. Each column in the plot corresponds to a different variation of instruments' absolute and relative strength. There are two types of absolute strengths, ``Strong'' and ``Weak'', measured by the concentration parameter. There are two types of strengths for valid and invalid instruments, ``Stronger Invalid'' and ``Stronger Valid'', determined by varying $\bm{\gamma}^*$ while holding the absolute strength fixed. Each row corresponds to the maximum correlation between instruments.}
\label{fig:equalZcorrEndo-awkward}
\end{figure}

\begin{figure}[htbp!]
\centering
\includegraphics[width=7in,height=7.3in]{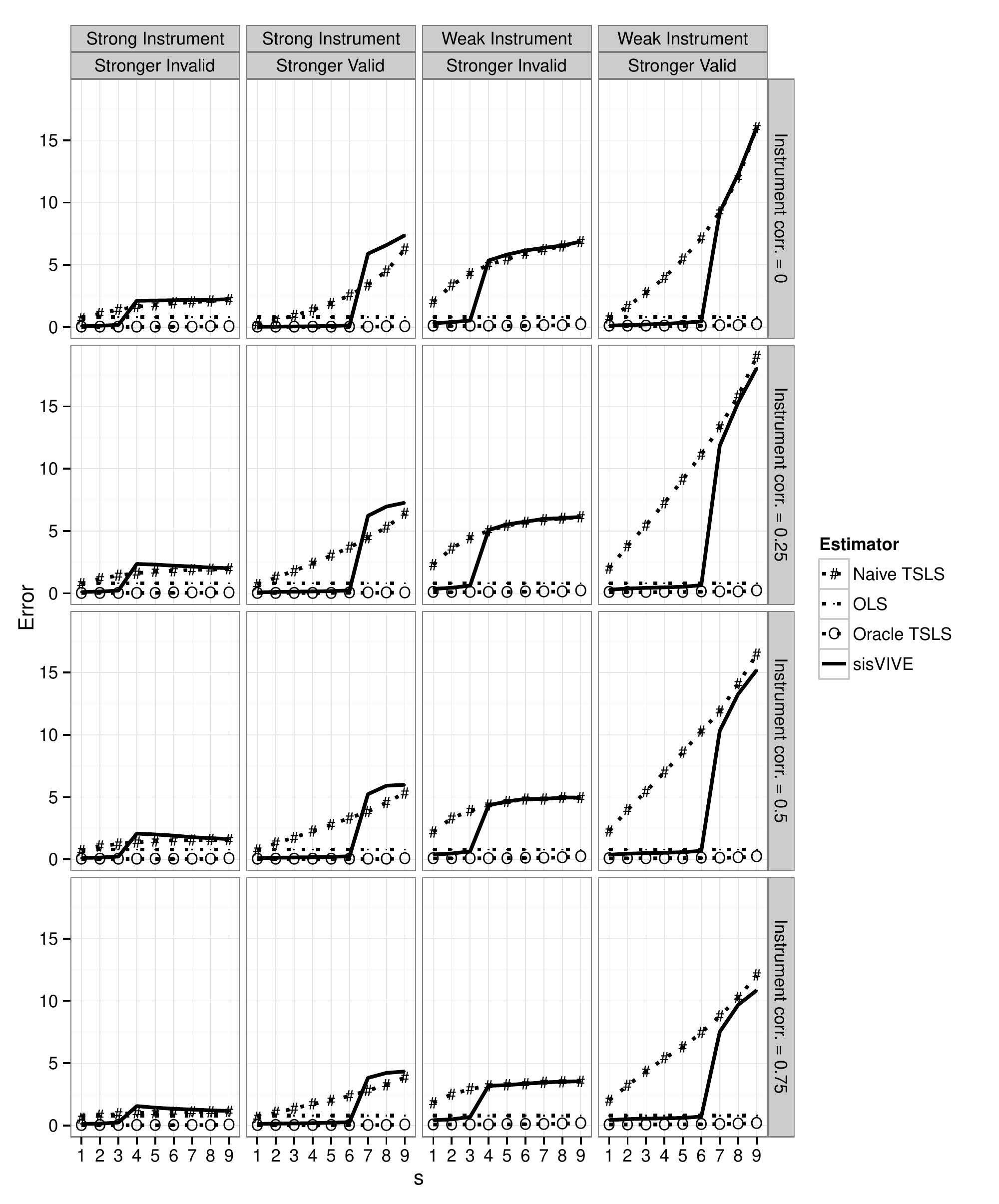}
\caption{Simulation Study Varying the Number of Invalid Instruments ($s$) and Correlation Exists Between All Instruments. We also vary the instrument strength of valid and invalid instruments. There are ten $(L = 10)$ instruments. Each line represents the median absolute estimation error ($|\beta^* - \hat{\beta}|$) after 500 simulations. We fix the endogeneity $\sigma_{\epsilon \xi}^*$ to $\sigma_{\epsilon \xi}^* = 0.8$. Each column in the plot corresponds to a different variation of instruments' absolute and relative strength. There are two types of absolute strengths, ``Strong'' and ``Weak'', measured by the concentration parameter. There are two types of strengths for valid and invalid instruments, ``Stronger Invalid'' and ``Stronger Valid'', determined by varying $\bm{\gamma}^*$ while holding the absolute strength fixed. Each row corresponds to the maximum correlation between instruments.}
\label{fig:equalZcorrS-awkward}
\end{figure}

\begin{figure}[htbp!]
\centering
\includegraphics[width=7in,height=7.3in]{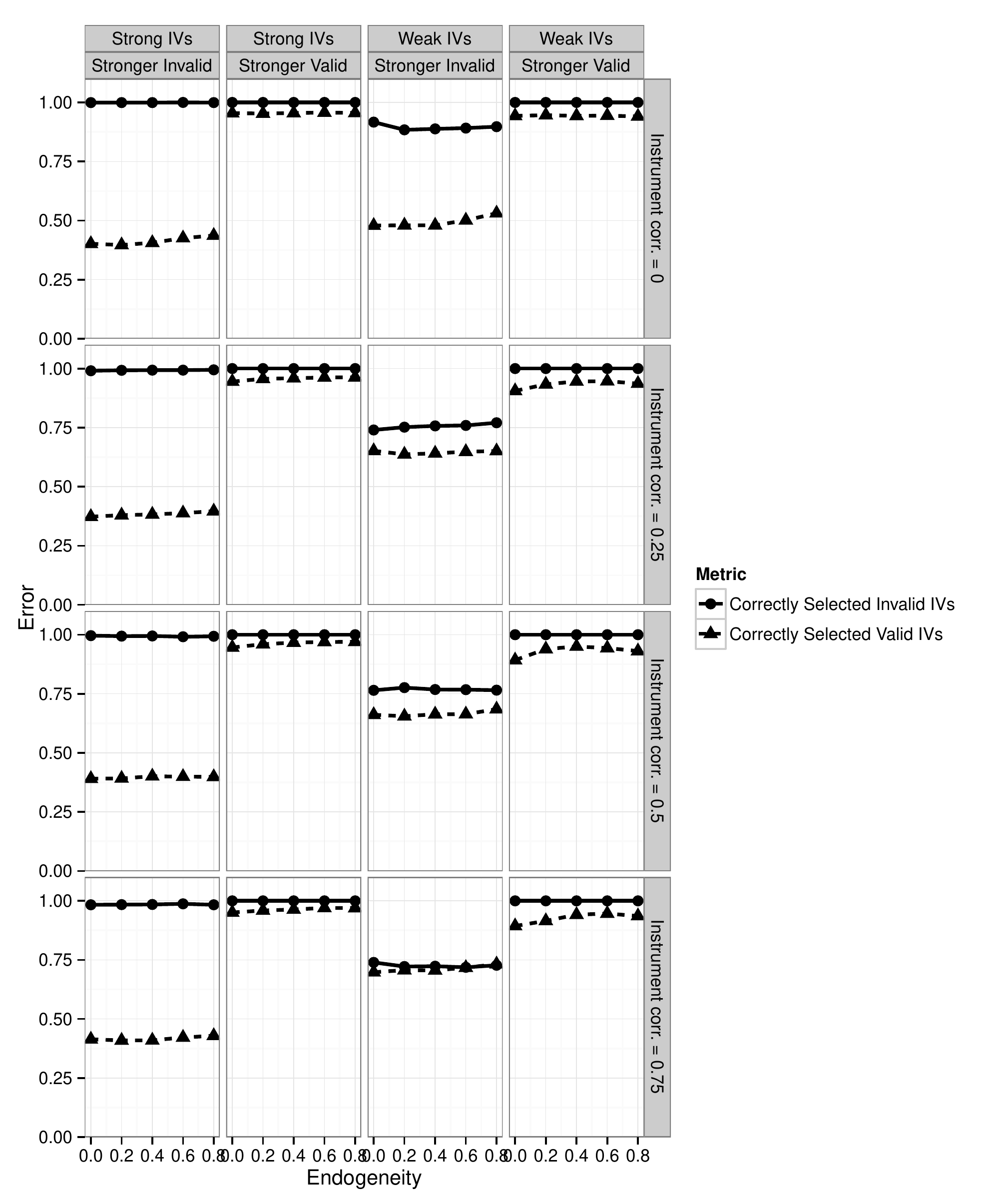}
\caption{Simulation Study Varying Endogeneity and Correlation Exists Between All Instruments. We also vary the instrument strength of valid and invalid instruments. There are ten $(L = 10)$ instruments. Each line represents the average proportions of correctly selected valid instruments and correctly selected invalid instruments after 500 simulations. We fix the number of invalid instruments to $s = 3$. Each column in the plot corresponds to a different variation of instruments' absolute and relative strength. There are two types of absolute strengths, ``Strong'' and ``Weak'', measured by the concentration parameter. There are two types of strengths for valid and invalid instruments, ``Stronger Invalid'' and ``Stronger Valid'', determined by varying $\bm{\gamma}^*$ while holding the absolute strength fixed. Each row corresponds to the maximum correlation between instruments.}
\label{fig:equalZcorrEndoPercent-awkward}
\end{figure}

\begin{figure}[htbp!]
\centering
\includegraphics[width=7in,height=7.3in]{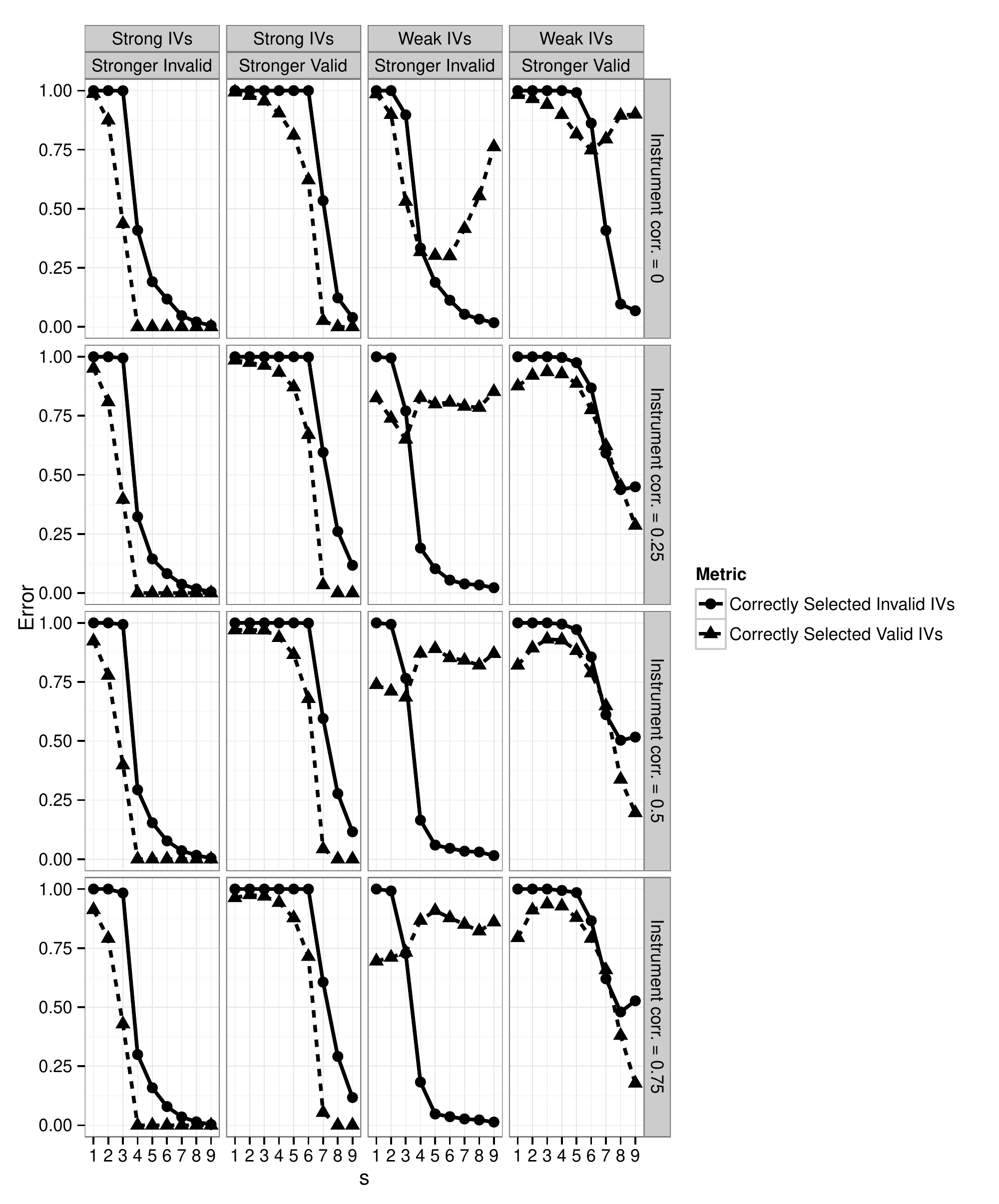}
\caption{Simulation Study Varying the Number of Invalid Instruments ($s$) and Correlation Exists Between All Instruments. We also vary the instrument strength of valid and invalid instruments. There are ten $(L = 10)$ instruments. Each line represents the average proportions of correctly selected valid instruments and correctly selected invalid instruments after 500 simulations. We fix the endogeneity $\sigma_{\epsilon \xi}^*$ to $\sigma_{\epsilon \xi}^* = 0.8$. Each column in the plot corresponds to a different variation of instruments' absolute and relative strength. There are two types of absolute strengths, ``Strong'' and ``Weak'', measured by the concentration parameter. There are two types of strengths for valid and invalid instruments, ``Stronger Invalid'' and ``Stronger Valid'', determined by varying $\bm{\gamma}^*$ while holding the absolute strength fixed. Each row corresponds to the maximum correlation between instruments.}
\label{fig:equalZcorrSPercent-awkward}
\end{figure}

\begin{figure}[htbp!]
\centering
\includegraphics[width=7in,height=7.3in]{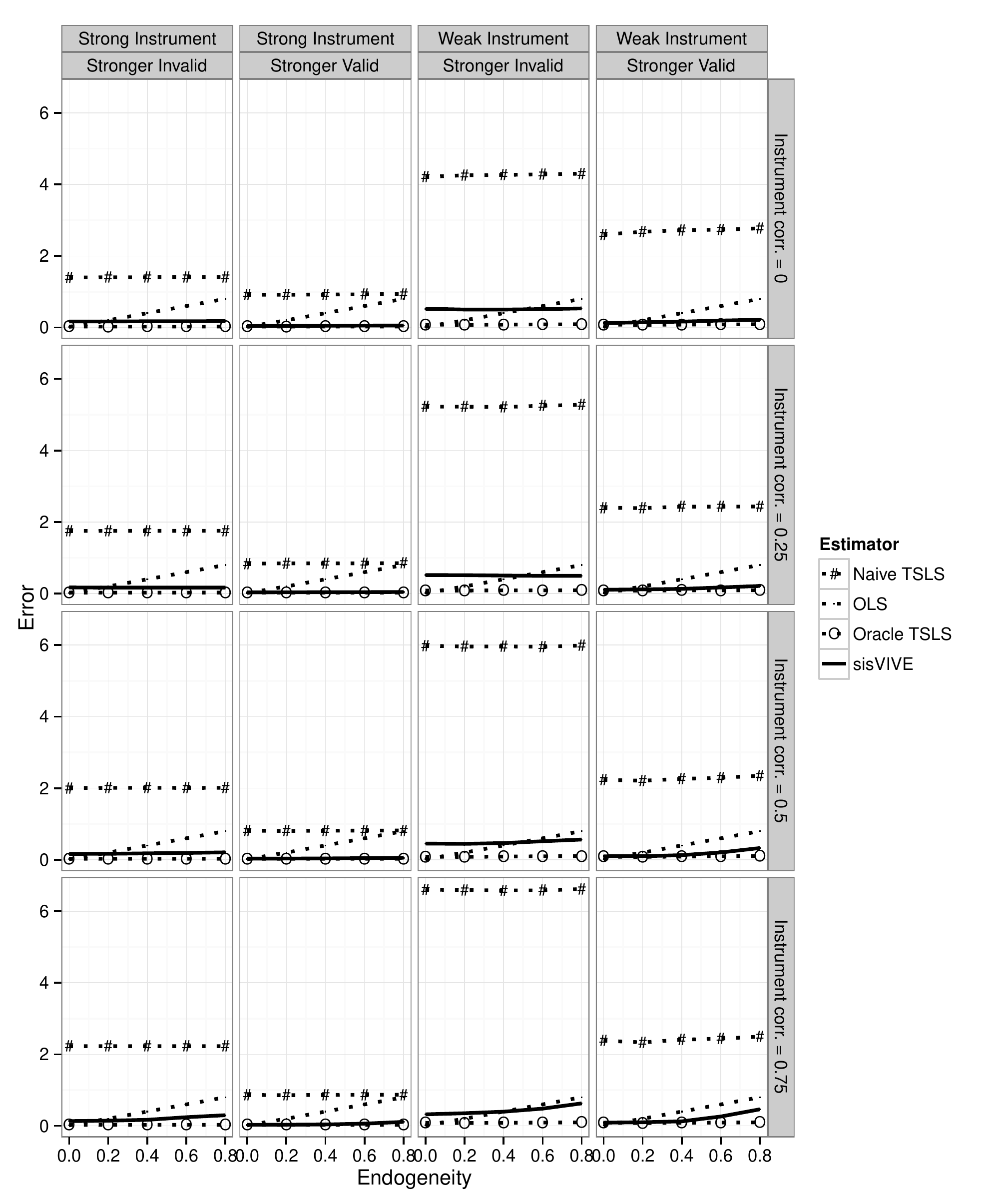}
\caption{Simulation Study Varying Endogeneity and Correlation Only Exists Within Valid and Invalid Instruments. We also vary the instrument strength of valid and invalid instruments. There are ten $(L = 10)$ instruments. Each line represents the median absolute estimation error ($|\beta^* - \hat{\beta}|$) after 500 simulations. We fix the number of invalid instruments to $s = 3$. Each column in the plot corresponds to a different variation of instruments' absolute and relative strength. There are two types of absolute strengths, ``Strong'' and ``Weak'', measured by the concentration parameter. There are two types of strengths for valid and invalid instruments, ``Stronger Invalid'' and ``Stronger Valid'', determined by varying $\bm{\gamma}^*$ while holding the absolute strength fixed. Each row corresponds to the maximum correlation between instruments, but correlation only exists within valid and invalid instruments.}
\label{fig:withinZcorrEndo-awkward}
\end{figure}

\begin{figure}[htbp!]
\centering
\includegraphics[width=7in,height=7.3in]{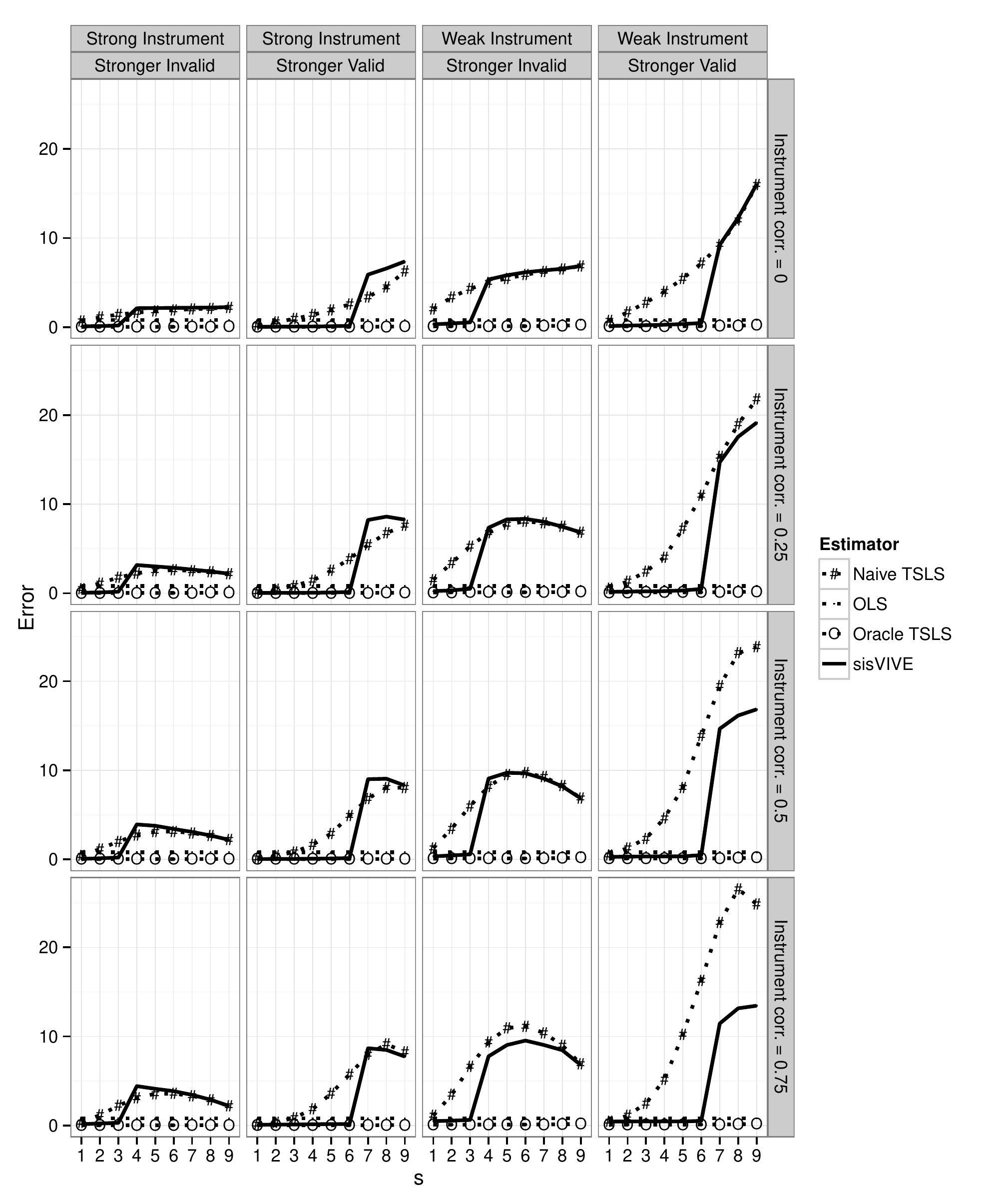}
\caption{Simulation Study Varying the Number of Invalid Instruments ($s$) and Correlation Only Exists Within Valid and Invalid Instruments. We also vary the instrument strength of valid and invalid instruments. There are ten $(L = 10)$ instruments. Each line represents the median absolute estimation error ($|\beta^* - \hat{\beta}|$) after 500 simulations. We fix the endogeneity $\sigma_{\epsilon \xi}^*$ to $\sigma_{\epsilon \xi}^* = 0.8$. Each column in the plot corresponds to a different variation of instruments' absolute and relative strength. There are two types of absolute strengths, ``Strong'' and ``Weak'', measured by the concentration parameter. There are two types of strengths for valid and invalid instruments, ``Stronger Invalid'' and ``Stronger Valid'', determined by varying $\bm{\gamma}^*$ while holding the absolute strength fixed. Each row corresponds to the maximum correlation between instruments, but correlation only exists within valid and invalid instruments.}
\label{fig:withinZcorrS-awkward}
\end{figure}

\begin{figure}[htbp!]
\centering
\includegraphics[width=7in,height=7.3in]{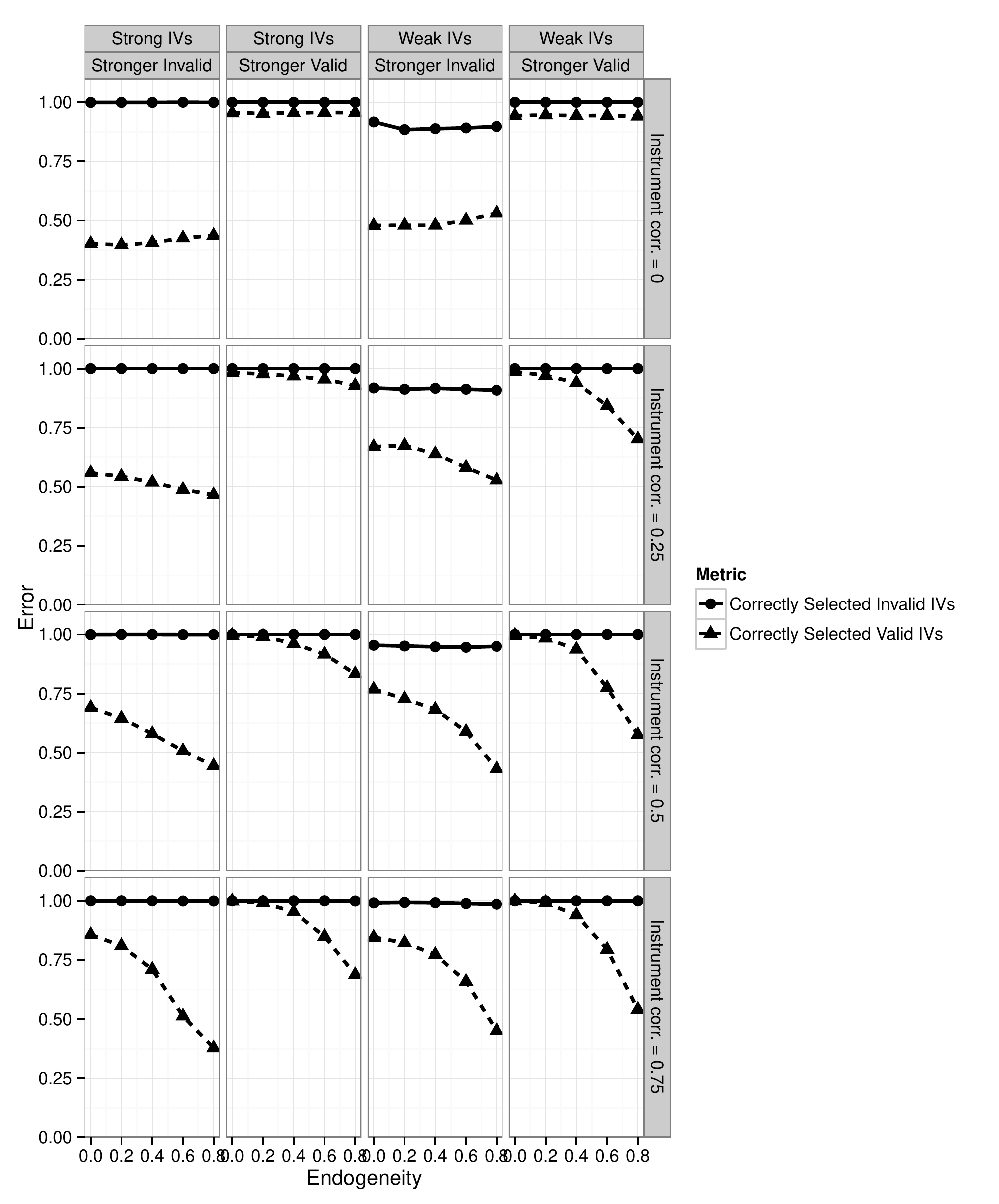}
\caption{Simulation Study Varying Endogeneity and Correlation Only Exists Within Valid and Invalid Instruments. We also vary the instrument strength of valid and invalid instruments. There are ten $(L = 10)$ instruments. Each line represents the average proportions of correctly selected valid instruments and correctly selected invalid instruments after 500 simulations. We fix the number of invalid instruments to $s = 3$. Each column in the plot corresponds to a different variation of instruments' absolute and relative strength. There are two types of absolute strengths, ``Strong'' and ``Weak'', measured by the concentration parameter. There are two types of strengths for valid and invalid instruments, ``Stronger Invalid'' and ``Stronger Valid'', determined by varying $\bm{\gamma}^*$ while holding the absolute strength fixed. Each row corresponds to the maximum correlation between instruments, but correlation only exists within valid and invalid instruments.}
\label{fig:withinZcorrEndoPercent-awkward}
\end{figure}

\begin{figure}[htbp!]
\centering
\includegraphics[width=7in,height=7.3in]{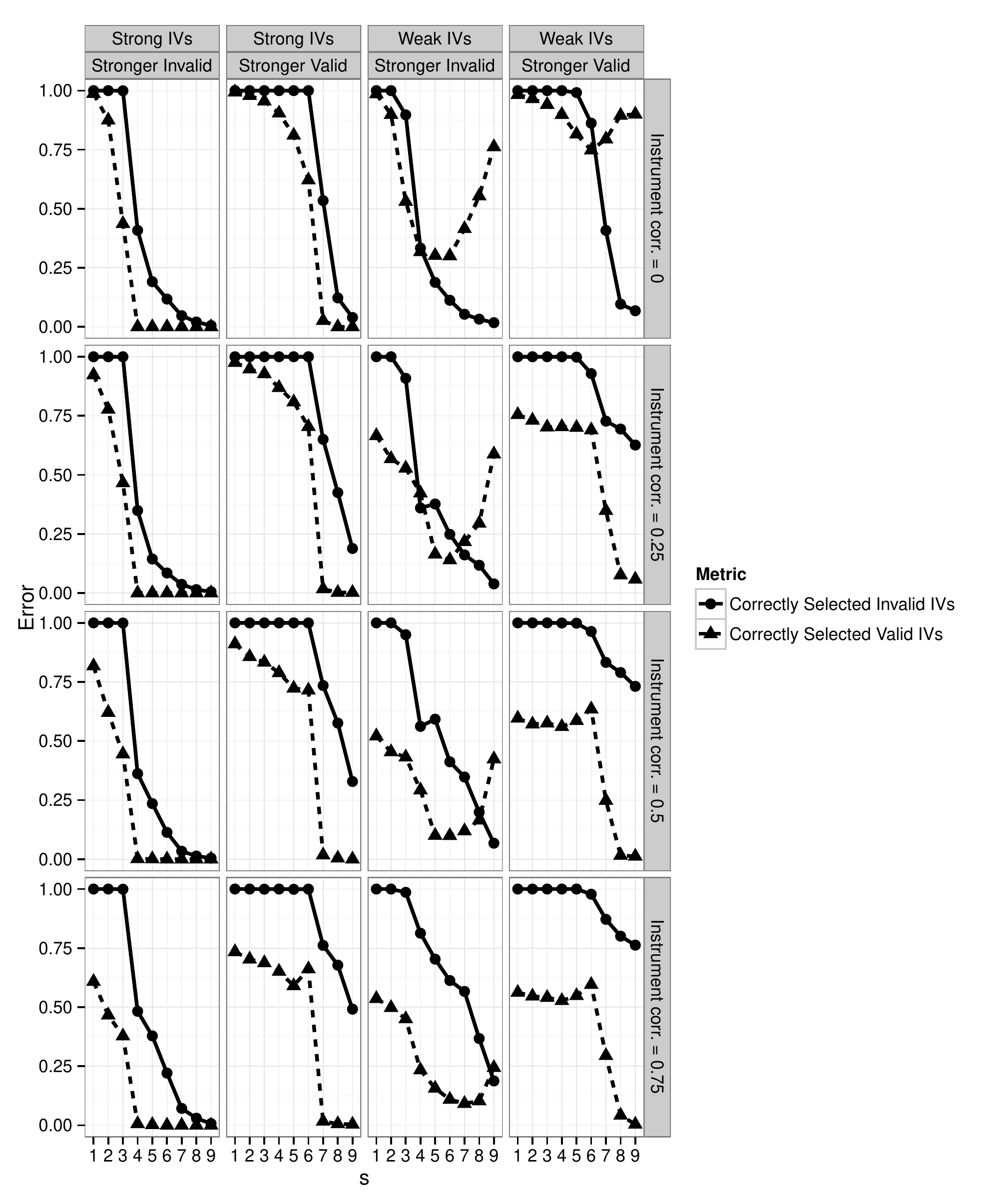}
\caption{Simulation Study Varying the Number of Invalid Instruments ($s$) and Correlation Only Exists Within Valid and Invalid Instruments. We also vary the instrument strength of valid and invalid instruments. There are ten $(L = 10)$ instruments. Each line represents the average proportions of correctly selected valid instruments and correctly selected invalid instruments after 500 simulations. We fix the endogeneity $\sigma_{\epsilon \xi}^*$ to $\sigma_{\epsilon \xi}^* = 0.8$. Each column in the plot corresponds to a different variation of instruments' absolute and relative strength. There are two types of absolute strengths, ``Strong'' and ``Weak'', measured by the concentration parameter. There are two types of strengths for valid and invalid instruments, ``Stronger Invalid'' and ``Stronger Valid'', determined by varying $\bm{\gamma}^*$ while holding the absolute strength fixed. Each row corresponds to the maximum correlation between instruments, but correlation only exists within valid and invalid instruments.}
\label{fig:withinZcorrSPercent-awkward}
\end{figure}

\begin{figure}[htbp!]
\centering
\includegraphics[width=7in,height=7.3in]{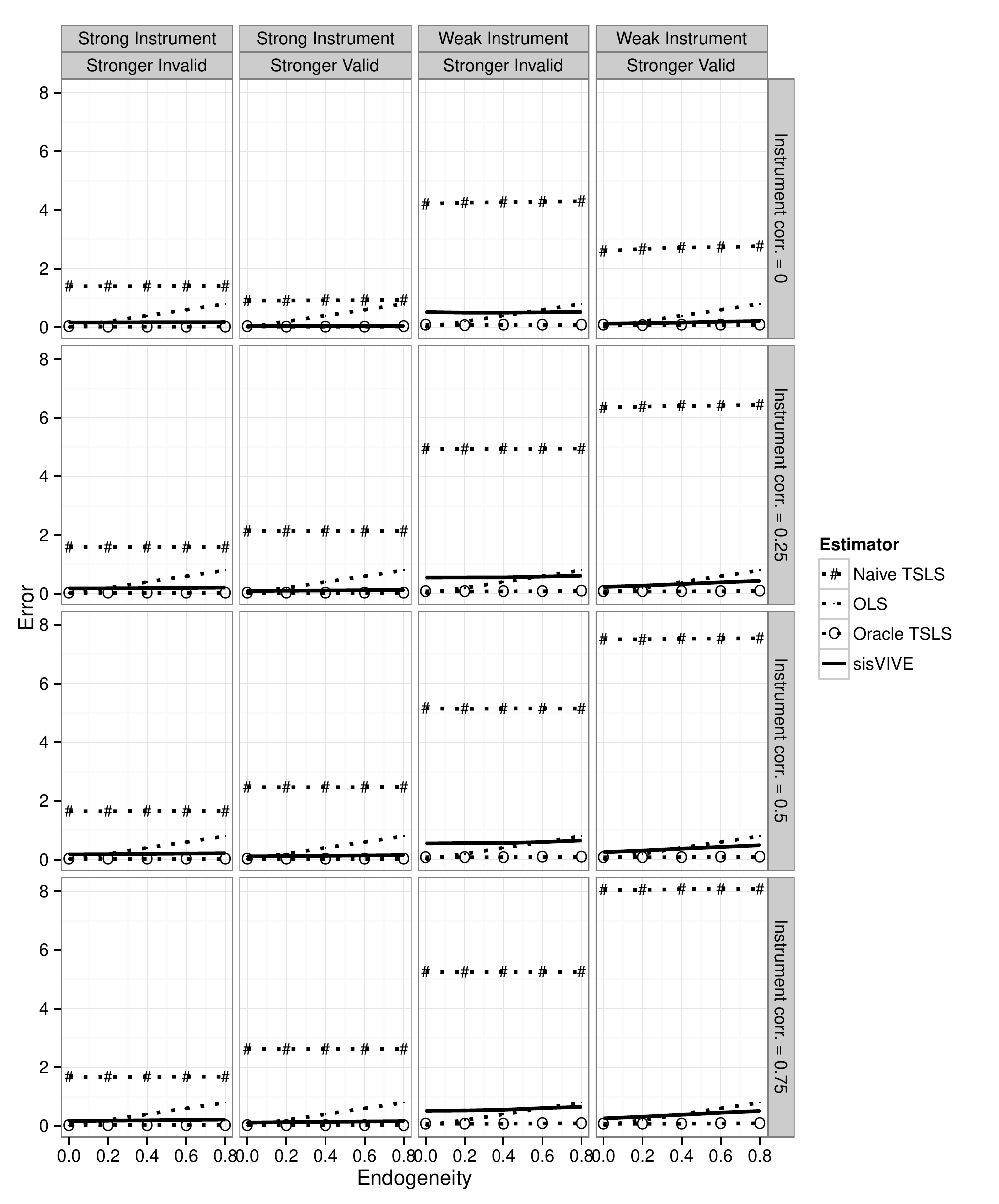}
\caption{Simulation Study Varying Endogeneity and Correlation Only Exists Between Valid and Invalid Instruments. We also vary the instrument strength of valid and invalid instruments. There are ten $(L = 10)$ instruments. Each line represents the median absolute estimation error ($|\beta^* - \hat{\beta}|$) after 500 simulations. We fix the number of invalid instruments to $s = 3$. Each column in the plot corresponds to a different variation of instruments' absolute and relative strength. There are two types of absolute strengths, ``Strong'' and ``Weak'', measured by the concentration parameter. There are two types of strengths for valid and invalid instruments, ``Stronger Invalid'' and ``Stronger Valid'', determined by varying $\bm{\gamma}^*$ while holding the absolute strength fixed. Each row corresponds to the maximum correlation between instruments, but correlation only exists between valid and invalid instruments.}
\label{fig:interZcorrEndo-awkward}
\end{figure}

\begin{figure}[htbp!]
\centering
\includegraphics[width=7in,height=7.3in]{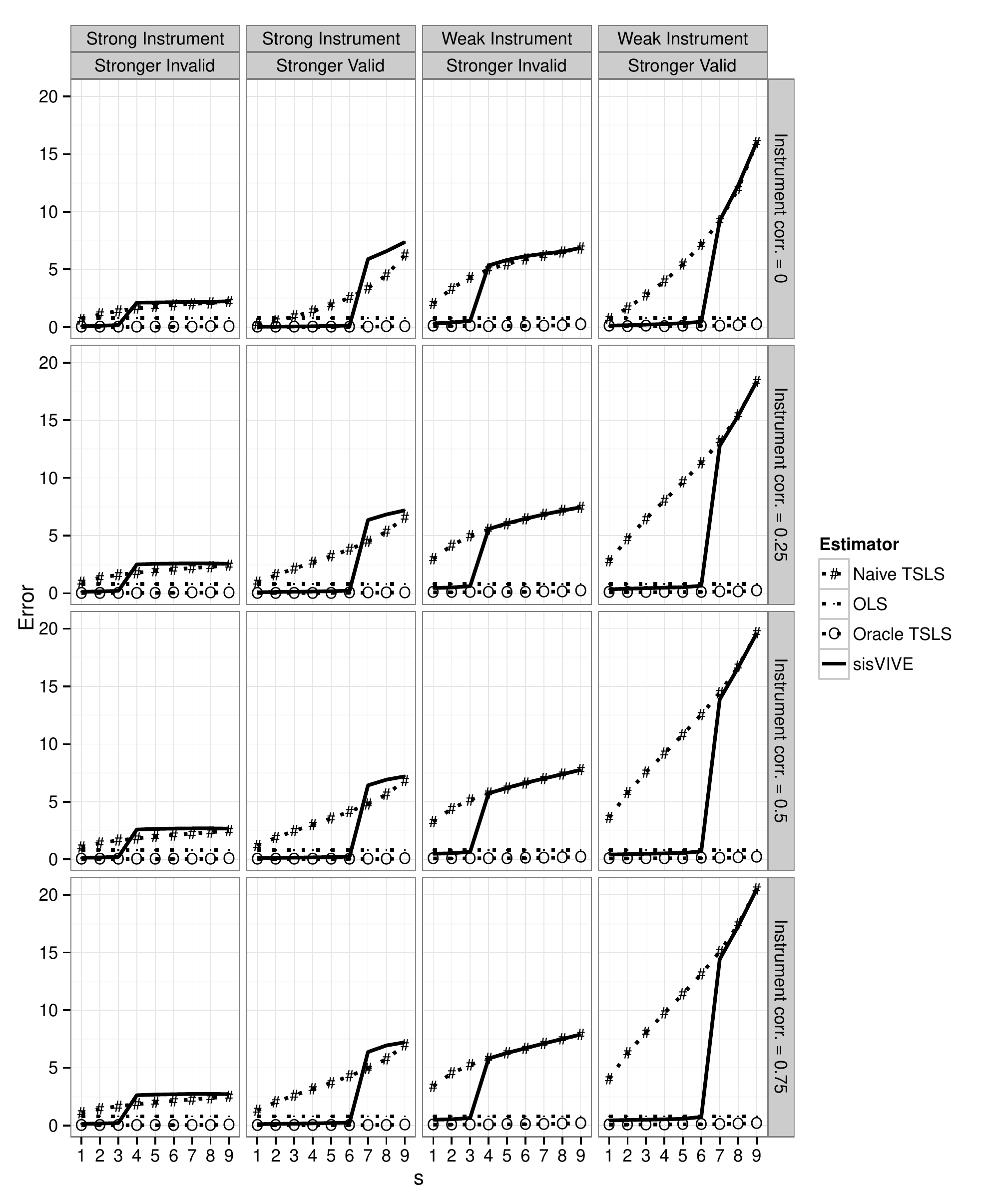}
\caption{Simulation Study Varying the Number of Invalid Instruments ($s$) and Correlation Only Exists Between Valid and Invalid Instruments. We also vary the instrument strength of valid and invalid instruments. There are ten $(L = 10)$ instruments. Each line represents median absolute estimation error ($|\beta^* - \hat{\beta}|$) after 500 simulations. We fix the endogeneity $\sigma_{\epsilon \xi}^*$ to $\sigma_{\epsilon \xi}^* = 0.8$. Each column in the plot corresponds to a different variation of instruments' absolute and relative strength. There are two types of absolute strengths, ``Strong'' and ``Weak'', measured by the concentration parameter. There are two types of strengths for valid and invalid instruments, ``Stronger Invalid'' and ``Stronger Valid'', determined by varying $\bm{\gamma}^*$ while holding the absolute strength fixed. Each row corresponds to the maximum correlation between instruments, but correlation only exists between valid and invalid instruments.}
\label{fig:interZcorrS-awkward}
\end{figure}

\begin{figure}[htbp!]
\centering
\includegraphics[width=7in,height=7.3in]{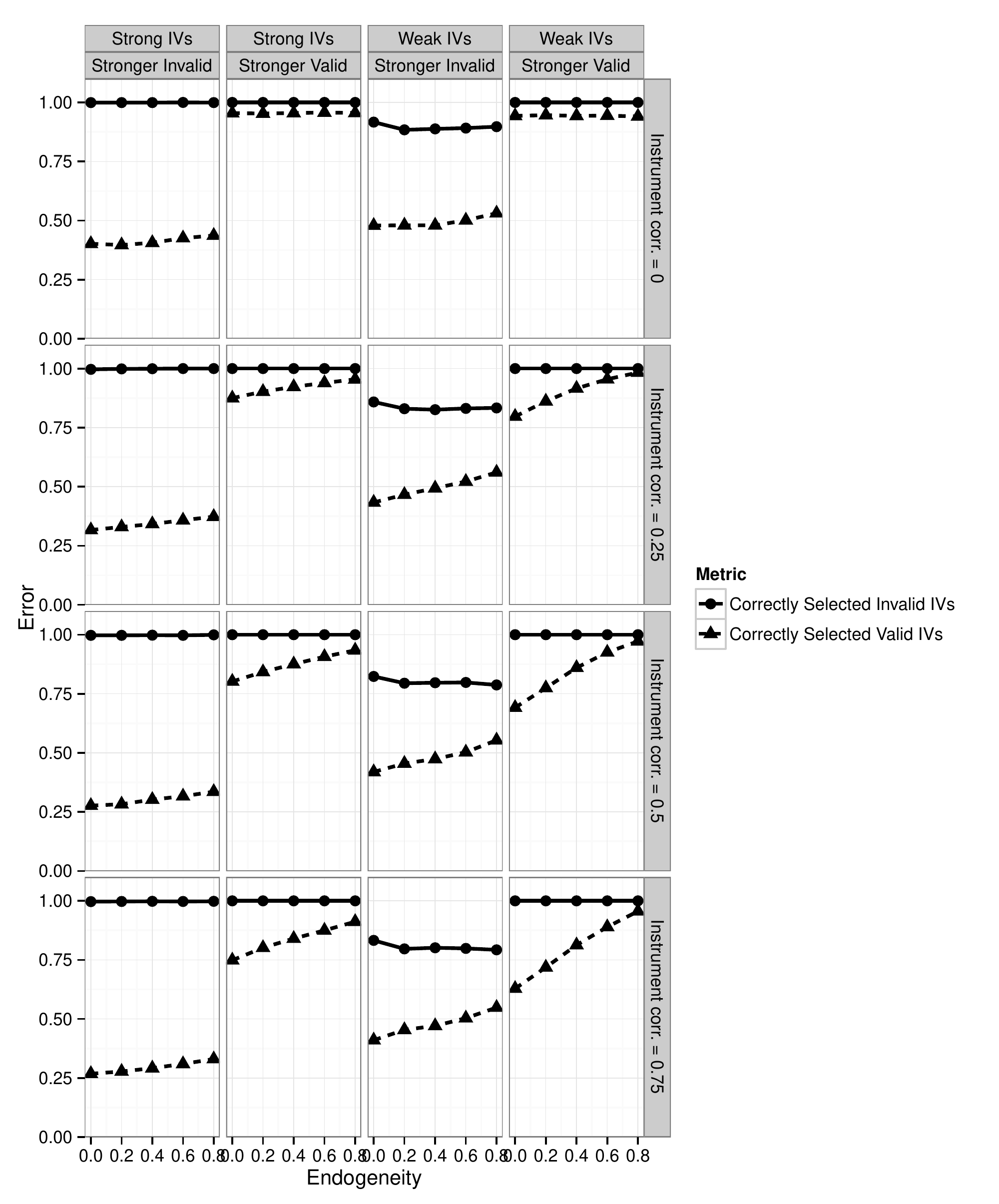}
\caption{Simulation Study Varying Endogeneity and Correlation Only Exists Between Valid and Invalid Instruments. We also vary the instrument strength of valid and invalid instruments. There are ten $(L = 10)$ instruments. Each line represents the average proportions of correctly selected valid instruments and correctly selected invalid instruments after 500 simulations. We fix the number of invalid instruments to $s = 3$. Each column in the plot corresponds to a different variation of instruments' absolute and relative strength. There are two types of absolute strengths, ``Strong'' and ``Weak'', measured by the concentration parameter. There are two types of strengths for valid and invalid instruments, ``Stronger Invalid'' and ``Stronger Valid'', determined by varying $\bm{\gamma}^*$ while holding the absolute strength fixed. Each row corresponds to maximum correlation between instruments, but correlation only exists between valid and invalid instruments.}
\label{fig:interZcorrEndoPercent-awkward}
\end{figure}

\begin{figure}[htbp!]
\centering
\includegraphics[width=7in,height=7.3in]{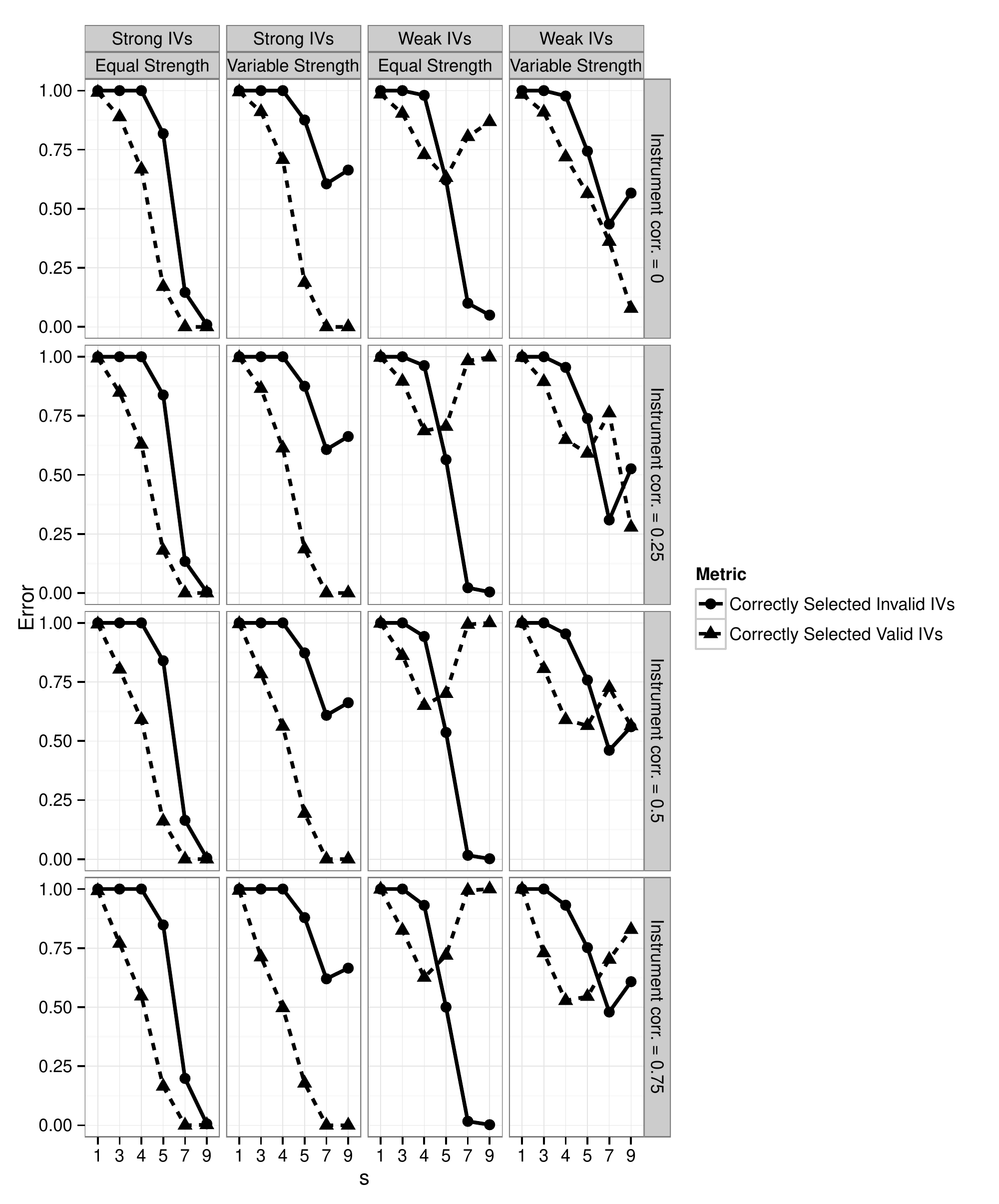}
\caption{Simulation Study Varying the Number of Invalid Instruments ($s$) and Correlation Only Exists Between Valid and Invalid Instruments. We also vary the instrument strength of valid and invalid instruments. There are ten $(L = 10)$ instruments. Each line represents the average proportions of correctly selected valid instruments and correctly selected invalid instruments after 500 simulations. We fix the endogeneity $\sigma_{\epsilon \xi}^*$ to $\sigma_{\epsilon \xi}^* = 0.8$. Each column in the plot corresponds to a different variation of instruments' absolute and relative strength. There are two types of absolute strengths, ``Strong'' and ``Weak'', measured by the concentration parameter. There are two types of strengths for valid and invalid instruments, ``Stronger Invalid'' and ``Stronger Valid'', determined by varying $\bm{\gamma}^*$ while holding the absolute strength fixed. Each row corresponds to maximum correlation between instruments, but correlation only exists between valid and invalid instruments.}
\label{fig:interZcorrSPercent-awkward}
\end{figure}

\begin{figure}[htbp!]
\centering
\includegraphics[width=7in,height=7.3in]{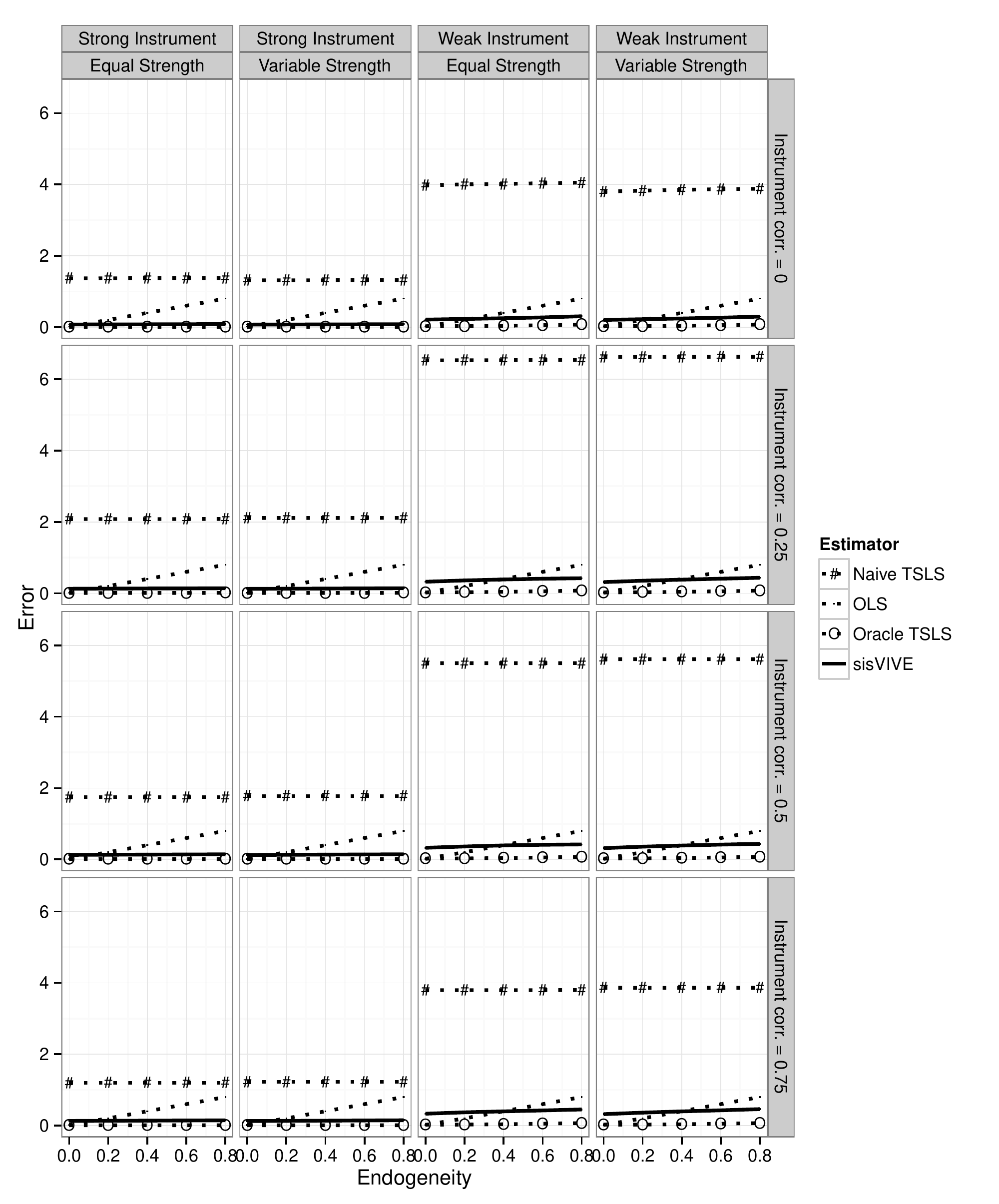}
\caption{Simulation Study of Estimation Performance Varying Endogeneity and Correlation Exists Between All Instruments. There are 100 $(L = 100)$ instruments. Each line represents the median absolute estimation error ($|\beta^* - \hat{\beta}|$) after 500 simulations. We fix the number of invalid instruments to $s = 30$. Each column in the plot corresponds to a different variation of instruments' absolute and relative strength. There are two types of absolute strengths, ``Strong'' and ``Weak'', measured by the concentration parameter. There are two types of relative strengths, ``Equal'' and ``Variable'', measured by varying $\bm{\gamma}^*$ while holding the absolute strength (i.e. concentration parameter) fixed. Each row corresponds to the maximum correlation between instruments.}
\label{fig:equalZcorrEndoL100}
\end{figure}

\begin{figure}[htbp!]
\centering
\includegraphics[width=7in,height=7.3in]{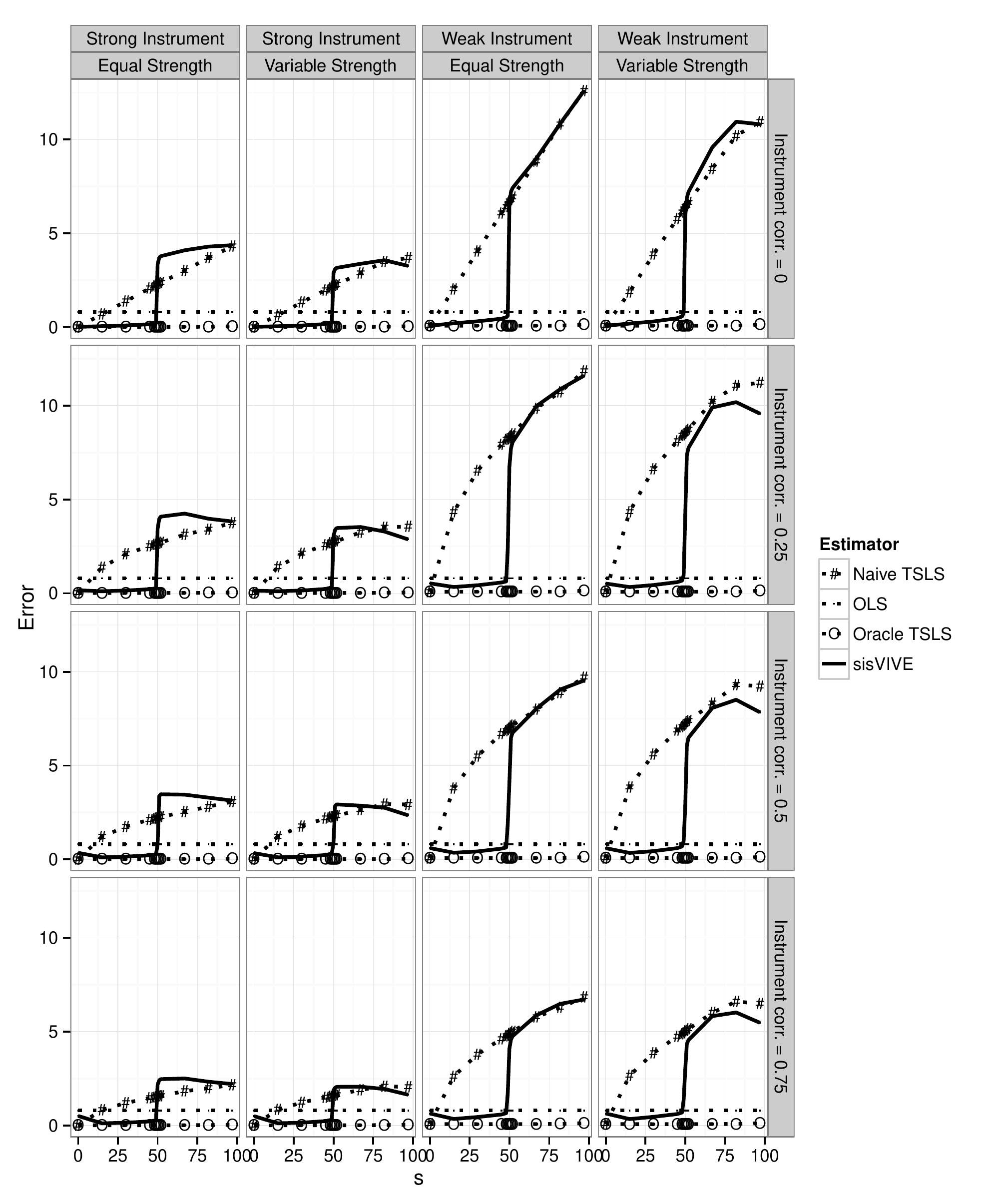}
\caption{Simulation Study of Estimation Performance Varying the Number of Invalid Instruments ($s$) and Correlation Exists Between All Instruments. There are 100 $(L = 100)$ instruments. Each line represents the median absolute estimation error ($|\beta^* - \hat{\beta}|$) after 500 simulations. We fix the endogeneity $\sigma_{\epsilon \xi}^*$ to $\sigma_{\epsilon \xi}^* = 0.8$. Each column in the plot corresponds to a different variation of instruments' absolute and relative strength. There are two types of absolute strengths, ``Strong'' and ``Weak'', measured by the concentration parameter. There are two types of relative strengths, ``Equal'' and ``Variable'', measured by varying $\bm{\gamma}^*$ while holding the absolute strength fixed. Each row corresponds to maximum correlation between instruments. }
\label{fig:equalZcorrSL100}
\end{figure}

\begin{figure}[htbp!]
\centering
\includegraphics[width=7in,height=7.3in]{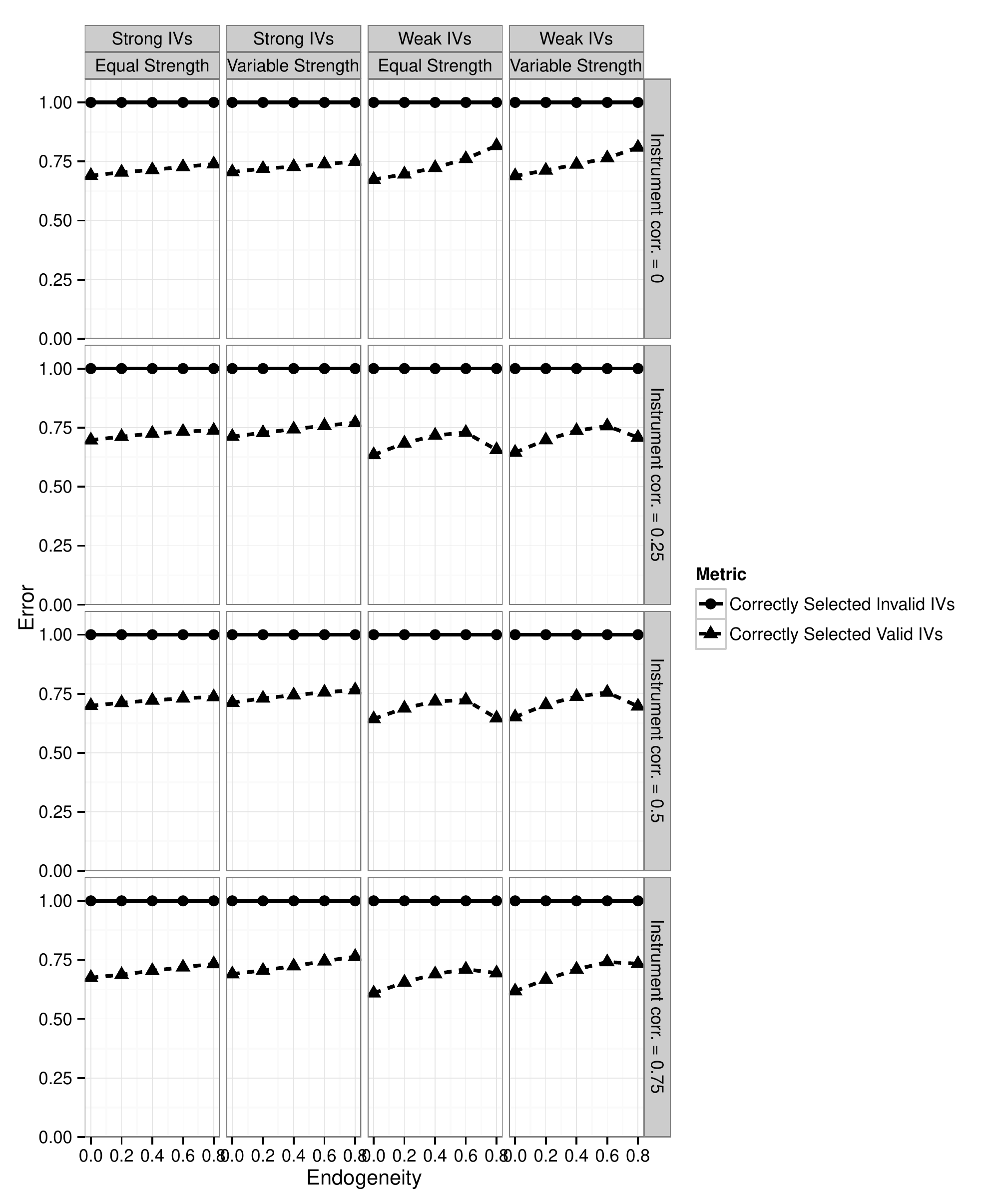}
\caption{Simulation Study Varying Endogeneity and Correlation Exists Between All Instruments. There are ten $(L = 100)$ instruments. Each line represents the average proportions of correctly selected valid instruments and correctly selected invalid instruments after 500 simulations. We fix the number of invalid instruments to $s = 30$. Each column in the plot corresponds to a different variation of instruments' absolute and relative strength. There are two types of absolute strengths, ``Strong'' and ``Weak'', measured by the concentration parameter. There are two types of relative strengths, ``Equal'' and ``Variable'', measured by varying $\bm{\gamma}^*$ while holding the absolute strength (i.e. concentration parameter) fixed. Each row corresponds to the maximum correlation between all instruments.}
\label{fig:equalZcorrEndoL100Percent}
\end{figure}

\begin{figure}[htbp!]
\centering
\includegraphics[width=7in,height=7.3in]{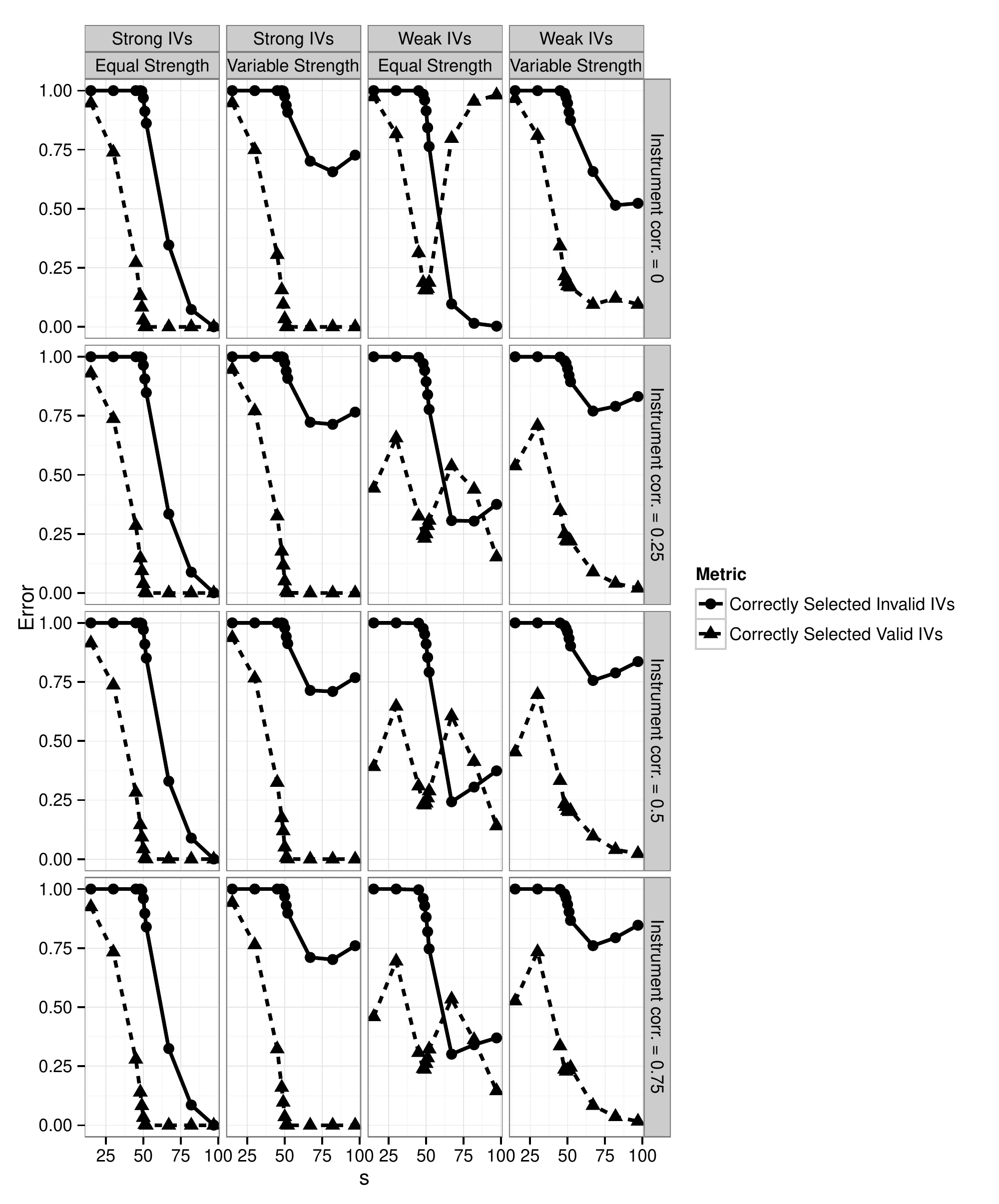}
\caption{Simulation Study Varying the Number of Invalid Instruments ($s$) and Correlation Exists Between All Instruments. There are 100 $(L = 100)$ instruments. Each line represents the average proportions of correctly selected valid instruments and correctly selected invalid instruments after 500 simulations. We fix the endogeneity $\sigma_{\epsilon \xi}^*$ to $\sigma_{\epsilon \xi}^* = 0.8$. Each column in the plot corresponds to a different variation of instruments' absolute and relative strength. There are two types of absolute strengths, ``Strong'' and ``Weak'', measured by the concentration parameter. There are two types of relative strengths, ``Equal'' and ``Variable'', measured by varying $\bm{\gamma}^*$ while holding the absolute strength fixed. Each row corresponds to maximum correlation between all instruments. }
\label{fig:equalZcorrSL100Percent}
\end{figure}

\end{document}